%% file: Journal_TIT_v2.tex
\documentclass[journal,11pt,onecolumn,draftclsnofoot]{IEEEtran}
\usepackage{cite}
\usepackage{graphicx,color,epsfig,rotating}
\usepackage{amsfonts,amsmath,amssymb,bbm}
\usepackage{setspace}
\usepackage{algorithm}
\usepackage[algo2e]{algorithm2e} 
\usepackage{subcaption}
\usepackage{mdwtab}
\usepackage{placeins}
\usepackage{psfrag, graphicx}
\usepackage[latin1]{inputenc}
\usepackage{multirow}
\usepackage{stfloats}
\usepackage{tabularx} 
\usepackage{booktabs} 
\usepackage{url}
\usepackage{bm}
\usepackage{geometry}
\usepackage{pstricks}
 \allowdisplaybreaks
\newtheorem{definition}{Definition}

\input{macros}

\begin{document}

\title{
Multi-Server Secure Aggregation with Arbitrary Collusion and Heterogeneous Security Constraints
}

\input{author_TIT.tex}

\maketitle
\IEEEpeerreviewmaketitle

\begin{abstract}
We study the fundamental limits of multi-server secure aggregation over a two-hop network where multiple servers, each connected to a disjoint subset of users, jointly compute the sum of all users' inputs. The goal is to ensure that no server can infer any information about prescribed subsets of inputs beyond the desired aggregate, even when colluding with an arbitrary subset of users.
Existing works largely focus on homogeneous security requirements, where all inputs are protected against colluding sets up to a given size. Such formulations are insufficient to capture more general scenarios in which different subsets of inputs may require protection against different collusion patterns.
In this paper, we consider a general model with heterogeneous security requirements and arbitrary user collusion. We characterize the communication rates for all parameter regimes, and determine the minimum key rate required  for secure aggregation in most regimes. In particular, we establish tight information-theoretic lower bounds and matching achievable schemes in a broad class of regimes. For the remaining regime, we derive a general lower bound together with an achievable scheme that attains it within a bounded gap. Our results reveal how the interplay between network topology and heterogeneous security constraints fundamentally determines the communication and key generation requirements, and generalize existing results on secure aggregation.

\end{abstract}


\section{Introduction}
\label{sec:intro}

Federated learning (FL) has emerged as a powerful paradigm for distributed learning~\cite{konecny2016federated,kairouz2021advances,mcmahan2017communication,yang2018applied}, enabling multiple users to collaboratively train models without sharing their raw data~\cite{Zhu2020DLG,geiping2020inverting}. Despite this advantage, privacy leakage remains a fundamental concern: even when only model updates are shared, an honest-but-curious aggregator or colluding users may still infer sensitive information about individual inputs. This tension between utility and security has made secure aggregation a central building block in modern distributed learning systems.

Secure aggregation~\cite{bonawitz2016practical,bonawitz2017practical} enables the computation of aggregate functions, such as the sum of user inputs, while provably concealing individual contributions. Existing solutions are largely based on cryptographic techniques, which rely on computational assumptions and often incur significant overhead. In contrast, information-theoretic approaches~\cite{zhao2023secure} aim to provide unconditional security guarantees and to characterize the fundamental trade-offs among communication, randomness, and security constraints.

Most existing information-theoretic formulations adopt homogeneous security requirements, where all users' inputs must be protected against any colluding set up to a given size. While analytically tractable, such models are inherently limited: they fail to capture practical scenarios in which different subsets of users require different levels of protection against different collusion patterns. Moreover, prior works are typically restricted to single-server or symmetric settings, which do not adequately reflect the heterogeneous and distributed nature of modern systems.

In this paper, we address these limitations by studying \tit{multi-server secure aggregation~\cite{li_zhang_MSSA} over a two-hop network with heterogeneous security requirements and arbitrary collusion}~\cite{li2025weakly,Li_Zhang_WeaklyDSA,Li_Zhang_WeaklyHSA}. The network consists of multiple servers, each connected to a disjoint subset of users, and all servers aim to compute the global sum of user inputs. We consider a general security model in which prescribed subsets of inputs must remain information-theoretically secure against any server, even when it colludes with an arbitrary subset of other users. This formulation enables fine-grained, asymmetric protection across user groups and strictly generalizes existing models as special cases. To the best of our knowledge, this is the first work that provides a unified information-theoretic treatment of multi-server secure aggregation with heterogeneous security constraints and arbitrary collusion.

From an information-theoretic perspective, the problem presents several fundamental challenges. Unlike single-server models, each server in a multi-server architecture must recover the global sum of all users' inputs, which imposes stringent consistency constraints across distributed observations. At the same time, each server observes heterogeneous information consisting of both locally received user messages and broadcasts messages from other servers, whose interaction creates intricate statistical dependencies that complicate the security analysis. Furthermore, the security requirements are inherently asymmetric, as both the protected input sets and the colluding user sets are arbitrary and need not exhibit any symmetry or size constraints, thereby precluding standard symmetry-based arguments. These challenges make it highly nontrivial to characterize how the interplay among network topology, heterogeneous security constraints, and arbitrary collusion structures jointly determines the minimum communication and randomness resources required for secure aggregation. Our results provide a near-complete information-theoretic characterization of these limits, with exact optimality established for all communication rates and most key-rate regimes, and a bounded-gap characterization for the remaining cases.

\subsection{Related Work}

Secure aggregation has been extensively studied in both cryptographic and information-theoretic settings. Cryptographic approaches, including secure multi-party computation and masking-based protocols, provide practical solutions under computational assumptions and have been widely adopted in federated learning systems~\cite{bonawitz2016practical,bonawitz2017practical}. However, these works typically do not aim to characterize the fundamental limits of communication and randomness.

Information-theoretic secure aggregation has recently attracted increasing attention. Early works focused on single-server settings or symmetric multi-user models with homogeneous security constraints, where all users' inputs are protected against colluding sets of bounded size~\cite{zhao2023optimal,zhao2022mds,zhao2024secure}. These works established fundamental trade-offs between communication and randomness, and developed coding-based schemes that are optimal in specific regimes.
More recent works have extended these results to distributed and multi-server architectures~\cite{li_zhang_MSSA}. In particular, multi-server secure aggregation has been investigated under various collusion models and network structures, revealing new trade-offs between communication efficiency and security guarantees. However, most existing formulations still assume homogeneous security requirements or impose symmetric structures across users and servers, which limit their applicability in heterogeneous systems.

In parallel, several works have considered practical extensions motivated by federated learning, including robustness to user dropout~\cite{9834981,so2022lightsecagg,jahani2022swiftagg,jahani2023swiftagg+,zhang2025secure}, user selection mechanisms~\cite{zhao2022mds,zhao2023optimal,zhao2024secure}, groupwise or structured key designs~\cite{zhao2023secure,wan2024information,wan2024capacity,Li_Zhang_GroupwiseDSA}, oblivious servers~\cite{9834981,sun2023secure}, malicious users~\cite{karakocc2021secure}, vector secure aggregation~\cite{yuan2025vectorlinearsecureaggregation,xu2026networkfunctioncomputationvector,hu2026capacityregionindividualkey}, decentralized secure aggregation~\cite{Zhang_Li_Wan_DSA,Li_Zhang_GroupwiseDSA,Li_Zhang_WeaklyDSA,zhang2026informationtheoreticsecureaggregationregular,weng2026resilient}, and hierarchical secure aggregation~\cite{Zhang_wan_HSA,zhang2025fundamental,xu2025hierarchicalsecureaggregationrelay}. These models capture additional system-level constraints but are typically studied under homogeneous or simplified security assumptions, and do not address the heterogeneous and asymmetric security requirements considered in this work.

We further clarify the distinction between our multi-server broadcast model and the closely related heterogeneous hierarchical secure aggregation framework. In the hierarchical setting, a three-layer user--relay--server topology is considered, where relays aggregate user messages and forward to a single server. A key feature of the hierarchical model is that each entity receives only one type of information: relays receive messages directly from users, and the server receives messages only from relays. Consequently, the security constraints are enforced separately for relays and for the server, each based on a single source of incoming information.

In contrast, our multi-server model features several key differences:

\begin{itemize}
    \item \textbf{Multiple decoders}: Every server independently decodes the global sum, imposing consistency constraints across all servers (see correctness condition (6)). In hierarchical settings, only a single server (or aggregator) decodes the final result.
    
    \item \textbf{Two types of information at each server}: Each server receives two distinct kinds of information: (i) direct messages from its associated users during the first hop, and (ii) broadcast messages from all other servers during the second hop. This creates intricate statistical dependencies that must be carefully managed in both the correctness and security analyses. Unlike the hierarchical setting where each entity processes only one type of incoming information, each server here must simultaneously handle both user-level and server-level messages.
    
    \item \textbf{Per-server secrecy with heterogeneous collusion}: Security must hold for each server when it colludes with arbitrary subsets of users. Because each server observes both user messages and inter-server broadcasts, the masking requirements become more complex. In particular, some servers may observe protected inputs directly, while others may only observe them indirectly through broadcast messages, leading to heterogeneous security conditions across servers.
\end{itemize}

While the hierarchical framework introduces certain structural parameters for heterogeneous security constraints, the multi-server broadcast architecture fundamentally alters the masking requirements. The presence of inter-server broadcasts and the fact that every server must decode the sum introduce new combinatorial conditions that determine the minimal key rate. The linear program in Theorem~2 captures the fractional key allocation needed to accommodate these inter-server dependencies and the two types of information at each server.

In contrast to the above works, our framework accommodates arbitrary security input sets and colluding user sets, enabling fine-grained and heterogeneous protection across users. This formulation strictly generalizes existing models and reveals the fundamental interplay between network topology, collusion patterns, and heterogeneous security requirements.

\subsection{Summary of Contributions}

The main contributions of this paper are summarized as follows:

\begin{itemize}
    \item We propose a general information-theoretic model for multi-server secure aggregation with heterogeneous security constraints and arbitrary collusion over a two-hop network. The model captures fine-grained security requirements across different user subsets under an honest-but-curious threat model.
    
    \item We characterize the optimal communication rates for all parameter regimes. In addition, we establish tight lower bounds on the minimum key rate for most cases and construct matching achievable schemes. In particular, we identify a critical parameter $e^*$ that precisely characterizes the minimum key rate in most regimes, capturing the maximum number of inputs that must be simultaneously protected under worst-case server--colluder configurations.
    
    \item For the remaining regime, we derive a general converse bound and propose an achievable scheme based on a linear program, yielding a bounded-gap characterization of the optimal key rate.
    
    \item Our results reveal that the interplay between network topology, heterogeneous security requirements, and collusion structures fundamentally determines the communication and randomness costs, and strictly generalize existing results on secure aggregation.
\end{itemize}

\tit{Paper Organization:}
The  rest of  this paper is organized as follows. Section~\ref{sec:problem description} introduces the system model and formulates the problem. Section~\ref{resu} presents the main results. Sections~\ref{pfthm1} and \ref{pfthm2} provide the proofs of Theorems~\ref{thm1} and \ref{thm2}, respectively, including converse, achievability, and lower/upper bound analyses. Section~\ref{pfsecurity} establishes the security of all proposed schemes. Finally, Section~\ref{conclu} concludes the paper.

\tit{Notation:} Throughout the paper, we use the following notation. 
For integers $m$ and $n$, let $[m:n] \eqdef \{m,m+1,\dots,n\}$ if $m\le n$, and $[m:n]=\emptyset$ if $m>n$. In particular, $[1:n]$ is abbreviated as $[n]$. 
Calligraphic letters (\eg, $\Ac,\Bc$) denote sets, and $\Ac \backslash \Bc \eqdef \{x\in \Ac: x\notin \Bc\}$ for sets $\Ac$ and $\Bc$. 
For a set of random variables $X_1,\dots,X_m$, we write $X_{\Sc} \eqdef \{X_i\}_{i\in \Sc}$ and $X_{1:m} \eqdef \{X_1,\dots,X_m\}$. 
The collection of all $n$-subsets of $\Ac$ is denoted by $\binom{\Ac}{n} \eqdef \{\Sc \subseteq \Ac: |\Sc|=n\}$, with $\binom{\Ac}{0}=\emptyset$. 
The power set of $\Ac$ is $\Pc(\Ac) \eqdef \cup_{n=0}^{|\Ac|} \binom{\Ac}{n}$. 
For two tuples $(i,j)$ and $(i',j')$, we say $(i,j)<(i',j')$ if $i<i'$, or if $i=i'$ and $j<j'$.

\section{Problem Formulation}
\label{sec:problem description}

We consider the \secagg problem in a two-hop communication network consisting of $U \geq 3$ servers. 
Each server $u \in [U]$ is connected to a unique set of $V_u$ users, disjoint from those of other servers, 
and the total number of users is $K = \sum_{u=1}^{U} V_u$. 
In the first hop, each user transmits its message to its associated server over an error-free link. 
In the second hop, each server broadcasts messages to all other servers through error-free broadcast channels, 
as illustrated in Fig.~\ref{fig:model}.
\begin{figure}[h]
    \centering
    \includegraphics[width=0.55\textwidth]{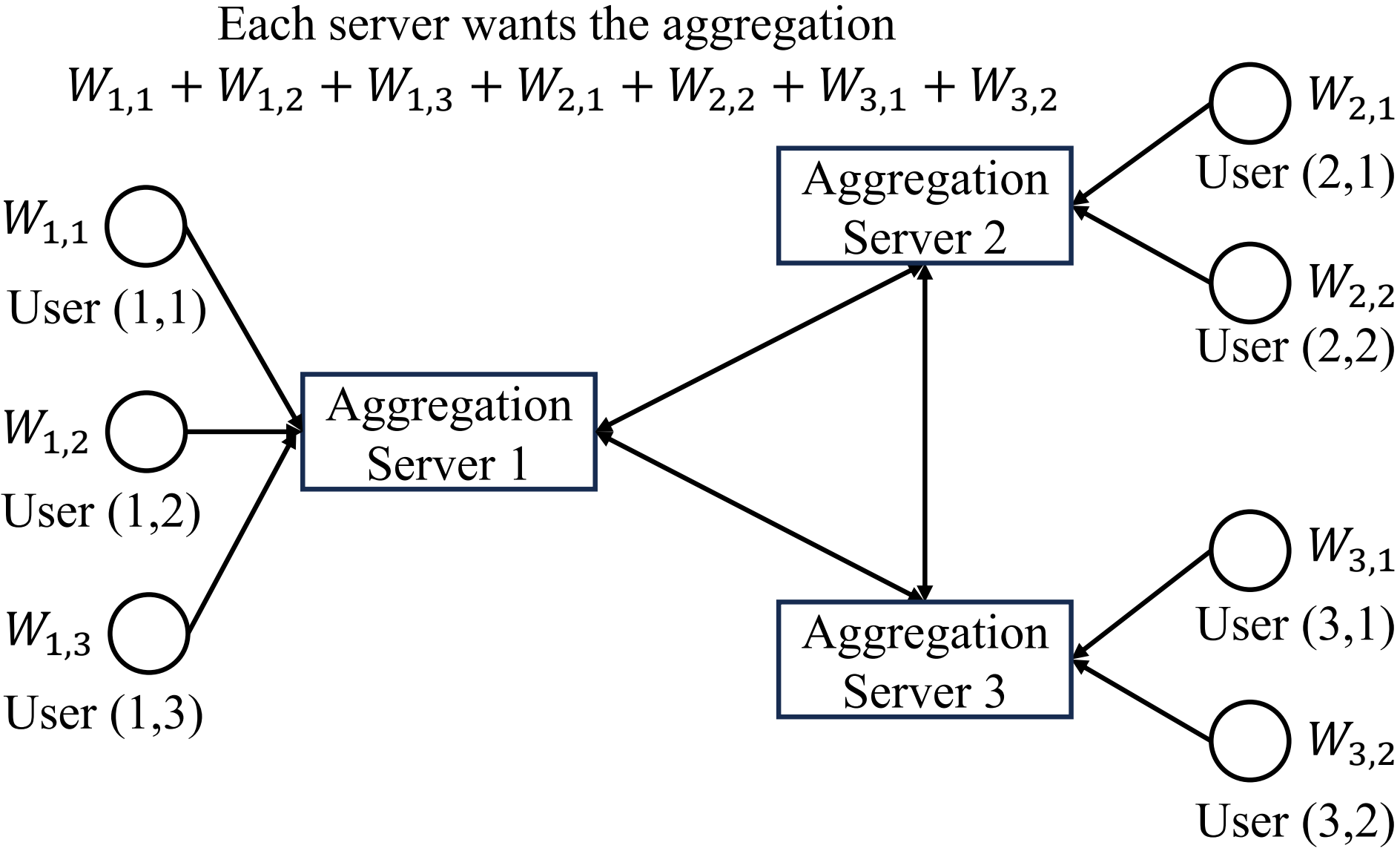}
\caption{\small Multi-server secure aggregation with $U=3$ servers and 
$V_1=3, V_2=V_3=2$ users per server. 
Each server computes the sum of all users' inputs 
$W_{1,1}+W_{1,2}+W_{1,3}+W_{2,1}+W_{2,2}+W_{3,1}+W_{3,2}$ 
without revealing the inputs in the security sets $\bm{\mathcal{S}}$, 
even if users in the colluding sets $\bm{\mathcal{T}}$ cooperate. 
The security sets are $\bm{\mathcal{S}} = (\emptyset, \{(1,1)\}, \{(2,1)\}, \{(1,1),(2,1)\})$, 
and the possible colluding sets $\bm{\mathcal{T}}$ include all combinations of 
$(1,2),(2,2),(3,1),(3,2)$, ranging from single users to all four together.}
    \label{fig:model}
\end{figure}
Let $\mathcal{K}_{\{u\}} \triangleq \{(u,v): v \in [V_u]\}$ denote the set of users associated with server $u$, 
where $(u,v)$ represents the $v^{\text{th}}$ user connected to server $u$. 
For any subset of servers $\mathcal{U} \subseteq [U]$, define 
$\mathcal{K}_{\mathcal{U}} \triangleq \bigcup_{u \in \mathcal{U}} \mathcal{K}_{\{u\}}$, 
which represents all users associated with servers in $\mathcal{U}$. 
In particular, the set of all users in the network is denoted by 
$\mathcal{K} \triangleq \mathcal{K}_{[U]} = \bigcup_{u \in [U]} \mathcal{K}_{\{u\}}$, 
where $|\mathcal{K}| = K = \sum_{u=1}^{U} V_u$.
Each user $(u,v)$ possesses an \emph{input} $W_{u,v}$ consisting of $L$ independent and identically distributed (i.i.d.) symbols uniformly drawn from a finite field $\mathbb{F}_q$. 
Additionally, each user holds a secret key $Z_{u,v}$, referred to as an \emph{individual key}, containing $L_Z$ symbols over $\mathbb{F}_q$. 
The inputs 
$\mathcal{W}_{\mathcal{K}} \triangleq \{W_{u,v}\}_{(u,v)\in\mathcal{K}}$
are mutually independent and independent of the keys 
$\mathcal{Z}_{\mathcal{K}} \triangleq \{Z_{u,v}\}_{(u,v)\in\mathcal{K}},$
i.e.,
\begin{align}
H\left(\mathcal{W}_{\mathcal{K}},\mathcal{Z}_{\mathcal{K}}\right)
&= H(\mathcal{Z}_{\mathcal{K}}) + \sum_{(u,v)\in\mathcal{K}} H(W_{u,v}), \label{ind} \\
H(W_{u,v}) &= L \quad \text{(in $q$-ary units)}, \quad \forall (u,v)\in\mathcal{K}. \label{h2}
\end{align}
This implies that all users' inputs are
mutually \indep, and also \indep of the secret keys.
The individual keys $\mathcal{Z}_{\mathcal{K}}$ may be arbitrarily correlated and are deterministically generated from a \emph{source key} $Z_{\Sigma}$, which consists of $L_{\Sigma}$ mutually independent symbols over the finite field $\mathbb{F}_q$, such that
\begin{align}
    H\left(\mathcal{Z}_{\mathcal{K}} \mid Z_{\Sigma}\right) = 0. \label{sourcekey}
\end{align}
A trusted  key server generates and 
assigns the keys offline, ensuring that each user $(u,v)$ possesses its designated key $Z_{u,v}$ prior to the aggregation task.

\tbf{\Comm protocol.}
The system follows a two-hop communication protocol. In the first hop, each user $(u,v)$ transmits an $L_X$-symbol message $X_{u,v}$ to its associated server $u$, which is deterministically generated from the user's input $W_{u,v}$ and individual key $Z_{u,v}$:
\begin{align}
    H(X_{u,v} \mid W_{u,v}, Z_{u,v}) = 0, \quad \forall (u,v) \in \mathcal{K}. \label{messageX}
\end{align}
In the second hop, each server $u$ generates an $L_Y$-symbol broadcast message $Y_u$ to all other servers, deterministically derived from the messages received from its associated users, $\{X_{u,v}\}_{v \in [V_u]}$:
\begin{align}
    H(Y_u \mid \{X_{u,v}\}_{v \in [V_u]}) = 0, \quad \forall u \in [U]. \label{messageY}
\end{align}
Each server $k \in [U]$, upon receiving the messages from its associated users and the broadcasts from the other servers, is required to recover the global sum of all users' inputs:
\begin{align}
\text{[Correctness]} \quad 
H\Bigg(\sum_{(u,v) \in \mathcal{K}} W_{u,v} \;\Big|\; 
\{Y_u\}_{u \in [U] \setminus \{k\}}, \{X_{k,v}\}_{v \in [V_k]} \Bigg) = 0, 
\quad \forall k \in [U]. \label{correctness}
\end{align}
Each server must recover the global sum of all users' inputs while preserving the privacy of individual inputs, which will be formalized in the next section.

\tbf{Security model.}
The security input sets are represented by a monotone\footnote{A set system is called monotone if, whenever a set belongs to the system, all its subsets also belong to the system.} set system 
$\bm{\mathcal{S}} = \{\mathcal{S}_1, \dots, \mathcal{S}_M\}$, 
and the colluding user sets are represented by another monotone\footnote{Assume $\bigcup_m \mathcal{S}_m \neq \emptyset$ and $|\mathcal{T}_n|\le K-2$, so that the security constraints are nontrivial.} set system 
$\bm{\mathcal{T}} = \{\mathcal{T}_1, \dots, \mathcal{T}_N\}$. 
The security constraint requires that, for any colluding user set $\mathcal{T}_n$, the users in $\mathcal{T}_n$ learn no additional information about the inputs in any $\mathcal{S}_m$ beyond what is already contained in their own inputs and keys if $\mathcal{S}_m \cap \mathcal{T}_n \neq \emptyset$:
\begin{align}
\text{[Security]} 
~I\Bigg(\{Y_u\}_{u \in [U] \setminus \{k\}}, \{X_{k,v}\}_{v \in [V_k]}; \{W_{i,j}\}_{(i,j)\in \mathcal{S}_m} 
\;\Bigg|\; \sum_{(i,j)\in \mathcal{K}} W_{i,j}, \{W_{i,j}, Z_{i,j}\}_{(i,j)\in \mathcal{T}_n} \Bigg) = 0, ~ \forall k \in [U]. \label{security}
\end{align}

\tbf{Performance metrics.}
The communication rates $R_X$ and $R_Y$ denote the number of  symbols per input symbol in $X_{u,v}$ and $Y_u$, respectively. 
Similarly, the source key rate $R_{Z_{\Sigma}}$ denotes the number of $q$-ary symbols per input symbol in the source key $Z_{\Sigma}$:
\begin{align}
\label{eq: def of rate}
R_X \eqdef \frac{L_X}{L}, \quad
R_Y \eqdef \frac{L_Y}{L}, \quad
R_{Z_{\Sigma}} \eqdef \frac{L_{Z_{\Sigma}}}{L}.
\end{align}
A rate tuple $(R_X, R_Y, R_{Z_{\Sigma}})$ is \emph{achievable} if there exists a secure coding scheme that implements the key generation and message construction for $X_{u,v}$, $Y_u$, and $Z_{\Sigma}$ such that the structural conditions (\ref{messageX}) and (\ref{messageY}) are satisfied, the communication and key rates $(R_X, R_Y, R_{Z_{\Sigma}})$ are attained, and both the correctness (\ref{correctness}) and security (\ref{security}) conditions hold. 
The \emph{optimal rate region} $\mathcal{R}^*$ is defined as the closure of all achievable rate tuples, representing the fundamental trade-offs among communication and key rates under the given correctness and security constraints.

\subsection{Auxiliary Definitions}
\label{subsec:auxiliary defs} 

In this section, we introduce several definitions that are useful for presenting the main results. 

Certain users require protection via key variables even if they do not explicitly belong to any $\mathcal{S}_m$. These implicitly protected users are characterized through the following sets.

\begin{definition}[Security sets $\mathcal{U}^{(m,n)}$ and $\mathcal{F}^{(m,n)}$]
\label{def:secrelay11}
For any pair of user sets $\mathcal{S}_m$ and $\mathcal{T}_n$, define
\begin{align}
\mathcal{U}^{(m,n)}
&\triangleq \Big\{ u \in [U] :
 |\mathcal{S}_m \cap \mathcal{K}_{\{u\}}| \neq 0,  
 |(\mathcal{S}_m  \cup \mathcal{T}_n)\cap \mathcal{K}_{\{u\}}| = V_u
\Big\}, \label{def:Uset} \\
\mathcal{F}^{(m,n)}
&\triangleq \Big\{ u \in [U] :
 |\mathcal{S}_m \cap \mathcal{K}_{\{u\}}| = 0, 
 |\mathcal{T}_n \cap \mathcal{K}_{\{u\}}| = V_u
\Big\}. \label{def:Fset}
\end{align}
\end{definition}

Intuitively, $\mathcal{U}^{(m,n)}$ and $\mathcal{F}^{(m,n)}$ classify servers based on their role in protecting the inputs of $\mathcal{S}_m$ against collusions in $\mathcal{T}_n$.  
A server $u \in \mathcal{U}^{(m,n)}$ is connected to at least one user in $\mathcal{S}_m$, and all its users belong to $\mathcal{S}_m \cup \mathcal{T}_n$. 
From the server's perspective, inputs in $\mathcal{K}_{\{u\}} \setminus \mathcal{S}_m$ are already known (they belong to $\mathcal{T}_n$). 
Hence, the inputs in $\mathcal{S}_m \cap \mathcal{K}_{\{u\}}$ require protection, which consumes at least $L$ symbols of independent key. 
Such servers directly observe protected inputs and are critical in the security analysis.
A server $u \in \mathcal{F}^{(m,n)}$ is connected only to users in $\mathcal{T}_n$ and observes no protected inputs. 
No additional key is needed for these servers.
For servers in $[U] \setminus (\mathcal{U}^{(m,n)} \cup \mathcal{F}^{(m,n)})$, 
any connection to users in $\mathcal{S}_m$ is accompanied by at least one connection to users outside $\mathcal{S}_m \cup \mathcal{T}_n$. 
These unknown inputs serve as implicit keys that mask the protected inputs. 
Consequently, no additional independent key is required. Servers without any connection to $\mathcal{S}_m$ also need no key.

We will use the following example to explain the definitions.
\begin{example}\label{ex1}
Consider $U=3$ and $V_1=3, V_2=V_3=2$, there are 7 users $(1,1),(1,2),(1,3),(2,1),(2,2),$ $(3,1),(3,2)$. The security input sets are $\bm{\mathcal{S}}=(\mathcal{S}_1,\cdots,\mathcal{S}_{4})=(\emptyset, \{(1,1)\}, \{(2,1)\}, \{(1,1),(2,1)\})$, and the colluding user sets are $\bm{\mathcal{T}}=(\mathcal{T}_1,\cdots,\mathcal{T}_{16})=(\emptyset, \{(1,2)\}, \{(2,2)\}, \{(3,1)\},  \{(3,2)\}, \{(1,2),(2,2)\},\{(1,2),(3,1)\},$ $\{(1,2),$ $(3,2)\},\{(2,2),(3,1)\}, \{(2,2),(3,2)\}, \{(3,1),(3,2)\},\{(1,2),(2,2),(3,1)\},\{(1,2),(2,2),(3,2)\},$ $\{(1,2),$ $(3,1),(3,2)\},\{(2,2),(3,1),(3,2)\},\{(1,2),(2,2),(3,1),(3,2)\})$.
\end{example}

Consider the pair $(\mathcal{S}_4=\{(1,1),(2,1)\}, \mathcal{T}_{16}=\{(1,2),(2,2),(3,1),(3,2)\})$. When $u=1$, the user set associated with server~1 is $\mathcal{K}_{\{1\}}=\{(1,1),(1,2),(1,3)\}$. Since $|\mathcal{S}_4\cap\mathcal{K}_{\{1\}}|=|\{(1,1),(2,1)\}\cap\{(1,1),(1,2),(1,3)\}|=1\neq 0$, and $|(\mathcal{S}_4\cup\mathcal{T}_{16})\cap\mathcal{K}_{\{1\}}|=|(\{(1,1),(2,1)\}\cup\{(1,2),(2,2),(3,1),(3,2)\})\cap\{(1,1),(1,2),(1,3)\}|=2\neq V_1$, hence $1\notin\mathcal{U}^{(4,16)}$. 
When $u=2$, we have $|\mathcal{S}_4\cap\mathcal{K}_{\{2\}}|=|\{(1,1),(2,1)\}\cap\{(2,1),(2,2)\}|=|\{(2,1)\}|\neq 0$, and $|(\mathcal{S}_4\cup\mathcal{T}_{16})\cap\mathcal{K}_{\{2\}}|=|(\{(1,1),(2,1)\}\cup\{(1,2),(2,2),(3,1),(3,2)\})\cap\{(2,1),(2,2)\}|=|\{(2,1),(2,2)\}|=2=V_2$, then $2\in \mathcal{U}^{(4,16)}$.
When $u=3$, we have $|\mathcal{S}_4\cap\mathcal{K}_{\{3\}}|=|\{(1,1),(2,1)\}\cap\{(3,1),(3,2)\}|= 0$, and $|\mathcal{T}_{16}\cap\mathcal{K}_{\{3\}}|=|\{(1,2),(2,2),(3,1),(3,2)\}\cap\{(3,1),(3,2)\}|=|\{(3,1),(3,2)\}|=2=V_3$, then $3\in \mathcal{F}^{(4,16)}$. So, $\mathcal{U}^{(4,16)}=\{2\}$, and $\mathcal{F}^{(4,16)}=\{3\}$. 
Intuitively, $u\in \mathcal{U}^{(m,n)}$ indicates that at least $L$ key symbols are needed to protect the inputs of users in server $u$. Consequently, a total of $|\mathcal{U}^{(m,n)}|$ key symbols are required to secure the inputs of all users in the server set $\mathcal{U}^{(m,n)}$.

\begin{definition}[Implicit Security Input Set $\mathcal{S}_I$]
\label{def:imp}
Let $\mathcal{K}$ denote the set of all $K$ users. 
Under the security constraint, the \emph{implicit security input set} is defined as
\begin{align}
\mathcal{S}_{I} 
\triangleq& 
\Big\{
\mathcal{K} \setminus 
\Big( (\mathcal{S}_m \cap \mathcal{K}_{\{u\}}) 
      \cup \mathcal{K}_{\mathcal{U}^{(m,n)}} 
      \cup \mathcal{T}_n \Big) 
\big| (\mathcal{S}_m \cap \mathcal{K}_{\{u\}}) 
       \cup \mathcal{K}_{\mathcal{U}^{(m,n)}} 
       \cup \mathcal{T}_n \big| = K-1,~\notag\\
&~~~~~~~~u\in [U], m\in[M], n\in[N] 
\Big\} 
\setminus \bigcup_{i\in[M]} \mathcal{S}_i.
\end{align}

Intuitively, an implicit security input set consists of users whose inputs are not explicitly listed in any $\mathcal{S}_m$, but whose values are uniquely determined once the other $K-1$ inputs are revealed. 
Although these inputs are not explicitly protected, any leakage would violate correctness, as it allows the remaining input to be inferred from the aggregation.

\medskip
\noindent
\textbf{Example (continued from Example~\ref{ex1}):} 
Consider $(u,m,n)=(1,4,16)$, where 
$ (\mathcal{S}_4 \cap \mathcal{K}_{\{1\}}) \cup \mathcal{K}_{\mathcal{U}^{(4,16)}} \cup \mathcal{T}_{16} = (\{(1,1),(2,1)\}\cap \{(1,1),(1,2),(1,3)\})  \cup \{(2,1),(2,2)\} \cup \{(1,2),(2,2),(3,1),(3,2)\} = 6 = K-1.$
Hence, the corresponding implicit security input set is $\mathcal{S}_{I} = \{(1,3)\}$.
This shows that although User $(1,3)$ does not belong to any explicit security input set, the implicit security constraints still require it to be protected by key variables. 
Intuitively, when Server~1 colludes with Users $(1,2)$, $(2,2)$, $(3,1)$, and $(3,2)$, the requirement of keeping Users $(1,1)$ and $(2,1)$ secure forces User $(1,3)$ to remain secure as well. 
More formally, this will be captured through entropy arguments; see~(\ref{corollary1_eqX}) in Corollary~\ref{corol1}.
\end{definition}

\begin{definition}[Total Security Input Set $\overline{\mathcal{S}}$] 
\label{def:tot} 
The \emph{total security input set} is defined as the union of all explicit and implicit security input sets:
\begin{align}
\overline{\mathcal{S}} \triangleq \Big(\bigcup_{m\in[M]} \mathcal{S}_m\Big) \cup \mathcal{S}_I.
\end{align}
\end{definition}

\noindent
\textbf{Example (continued from Example~\ref{ex1}):} 
In this case, the total security input set is 
$\overline{\mathcal{S}} = \{(1,1),(1,3),(2,1)\}.$



\begin{definition}[$a^*$ and $e^*$] 
\label{def:totset1}
For any $u\in [U], m\in [M], n\in [N]$, define
\begin{equation}
\mathcal{A}_{u,m,n} \triangleq 
\Big((\mathcal{S}_m \cap \mathcal{K}_{\{u\}}) 
\cup \mathcal{K}_{\mathcal{U}^{(m,n)}} 
\cup \mathcal{T}_n\Big) \cap \overline{\mathcal{S}} .
\end{equation}
The maximum cardinality of $\mathcal{A}_{u,m,n}$ is denoted as
\begin{equation}
\label{eq:def a*}
a^* \triangleq \max_{u\in [U],\,m\in[M],\,n\in[N]} 
|\mathcal{A}_{u,m,n}|.
\end{equation}
Furthermore, define
\begin{equation}
\label{eq:def e*}
e^* \triangleq 
\max_{u\in [U],\,m\in[M],\,n\in[N]}
\Big(
\big|(\mathcal{S}_m \cap \mathcal{K}_{\{u\}})\cup \mathcal{T}_n\big|\cap\overline{\mathcal{S}}
+|\mathcal{U}^{(m,n)}\setminus\{u\}|
\Big).
\end{equation}
\end{definition}

\noindent
\textbf{Example (continued from Example~\ref{ex1}):} 
We have:
$\mathcal{A}_{1,4,16}
=((\{(1,1),(2,1)\}\cap\{(1,1),(1,2)\})
\cup\{(2,1),(2,2)\}
\cup\{(1,2),(2,2),(3,1),(3,2)\})
\cap\{(1,1),(1,3),(2,1)\}
=\{(1,1),(2,1)\}$.
Hence, $a^*=2$.

Moreover,
$\big|(\mathcal{S}_4 \cap \mathcal{K}_{\{1\}}) \cup \mathcal{T}_{16} \big| \cap \overline{\mathcal{S}} = 1, 
|\mathcal{U}^{(4,16)}\setminus\{1\}| = 1,$
so that $e^* = 2$.




\begin{definition}[Union of Sets $\mathcal{Q}_1$ and $\mathcal{Q}_2$]
\label{def:uni2}
Find all $u,m,n$ such that $|\mathcal{A}_{u,m,n}|=|\overline{\mathcal{S}}|$, 
and denote the union of the corresponding 
$(\mathcal{S}_m\cap\mathcal{K}_{\{u\}})\cup \mathcal{K}_{\mathcal{U}^{(m,n)}} \cup \mathcal{T}_n$
as $\mathcal{Q}_1$, i.e.,
\begin{equation}
\mathcal{Q}_1 \triangleq 
\bigcup_{u,m,n:|\mathcal{A}_{u,m,n}|=|\overline{\mathcal{S}}|}
\Big(
(\mathcal{S}_m \cap \mathcal{K}_{\{u\}})
\cup \mathcal{K}_{\mathcal{U}^{(m,n)}}
\cup \mathcal{T}_n
\Big).
\end{equation}
Next, find all $u,m,n$ such that
$|((\mathcal{S}_m \cap \mathcal{K}_{\{u\}})\cup \mathcal{T}_n)\cap\overline{\mathcal{S}}|
+|\mathcal{U}^{(m,n)}\setminus\{u\}|=e^*,$
and denote the union of the corresponding sets as $\mathcal{Q}_2$, i.e.,
\begin{equation}
\mathcal{Q}_2 \triangleq
\bigcup_{u,m,n:|((\mathcal{S}_m \cap \mathcal{K}_{\{u\}})\cup \mathcal{T}_n)\cap\overline{\mathcal{S}}|
+|\mathcal{U}^{(m,n)}\setminus\{u\}|=e^*}
\Big(
(\mathcal{S}_m \cap \mathcal{K}_{\{u\}})
\cup \mathcal{K}_{\mathcal{U}^{(m,n)}}
\cup \mathcal{T}_n
\Big).
\end{equation}
Finally, if $|\mathcal{Q}_1|\neq0$, define
\begin{equation}
\mathcal{Q}\triangleq \mathcal{Q}_1\cup\mathcal{Q}_2 .
\end{equation}
Otherwise, $\mathcal{Q}=\emptyset$.

\end{definition}

In Example~\ref{ex1}, we have $|\overline{\mathcal{S}}| = 3$. 
Since there exist no $u,m,n$ such that $|\mathcal{A}_{u,m,n}| = 3$, it follows that $\mathcal{Q}_1 = \emptyset$. 
Therefore, $\mathcal{Q} = \emptyset$.
On the other hand, $e^* = 2$. 
By enumerating all $u,m,n$ such that 
$|((\mathcal{S}_m \cap \mathcal{K}_{\{u\}})\cup \mathcal{T}_n)\cap\overline{\mathcal{S}}|
+|\mathcal{U}^{(m,n)}\setminus \{u\}| = e^* = 2$, 
we obtain 
$\mathcal{Q}_2=\{(1,1),(1,2),(2,1),(2,2),(3,1),(3,2)\}$.
For example, when $(u,m,n)=(1,4,16)$,
$|((\mathcal{S}_4 \cap \mathcal{K}_{\{1\}})\cup \mathcal{T}_{16})\cap\overline{\mathcal{S}}|
+|\mathcal{U}^{(4,16)}\setminus \{1\}| = 2$.

\section{Main Result} \label{resu}

\begin{theorem}\label{thm1}
For the multi-server secure aggregation problem with $U \ge 3$ servers, where each server connects to $V_u$ users ($u \in [U]$), security input sets $\{\mathcal{S}_m\}_{m \in [M]}$, and colluding user sets $\{\mathcal{T}_n\}_{n \in [N]}$, the optimal rate region is given by
\begin{equation}
\mathcal{R}^* = \{(\rx,\ry,\rzsigma) \mid R_X \ge 1, R_Y \ge 1, R_{Z_\Sigma} \ge R_{Z_\Sigma}^* \},
\end{equation}
if one of the following mutually exclusive conditions holds:
\begin{enumerate}
\item $e^* = K$;
\item $e^* \le K-1$ and $a^* \le |\overline{\mathcal{S}}|-1$;
\item $e^* \le K-1$, $a^* = |\overline{\mathcal{S}}|$, and $|\mathcal{Q}| \le K-1$.
\end{enumerate}
Under these conditions,
\begin{equation}
R_{Z_\Sigma}^* =
\begin{cases}
K-1, & \trm{if } e^* = K,\\
e^*, & \text{otherwise}.
\end{cases}
\end{equation}
\end{theorem}

\begin{remark}[Intuitive Explanation of Theorem~\ref{thm1}]
Theorem~\ref{thm1} characterizes the optimal communication and source key rates for multi-server secure aggregation under arbitrary collusion constraints, as specified by the three mutually exclusive conditions above. Intuitively, these conditions indicate which users' inputs require protection and explain how the corresponding key requirements arise:

\begin{itemize}
    \item \textbf{Condition 1 ($e^* = K$):} Every user participates in at least one critical combination of server connections and colluding sets. 
    Since a server can observe the sum of all inputs except one, all $K$ users must be protected. 
    However, because the sum of inputs is known from the aggregation, only $K-1$ independent key symbols are required. The remaining input is implicitly secured as it is determined by the sum of the other $K-1$ keys.

    \item \textbf{Condition 2 ($e^* \le K-1$ and $a^* \le |\overline{\mathcal{S}}|-1$):} 
    Here, only users in the total security input set $\overline{\mathcal{S}}$ require key protection; users outside $\overline{\mathcal{S}}$ do not, as their values do not affect security. 
    The maximal number of keys simultaneously needed by any server for a given colluding set is $e^*$. 
    Intuitively, in the worst-case scenario, a server observes up to $e^*$ inputs that must be protected at the same time, so $e^*$ independent key symbols suffice to guarantee security.

    \item \textbf{Condition 3 ($e^* \le K-1$, $a^* = |\overline{\mathcal{S}}|$, and $|\mathcal{Q}| \le K-1$):} 
    In this case, some server-colluder combinations cover all users in $\overline{\mathcal{S}}$ but not the entire set $\mathcal{K}$. 
    Therefore, keys are needed for all users in $\overline{\mathcal{S}}$ and potentially at most one additional user outside $\mathcal{Q}$ who participates in the combination that reaches the maximum critical overlap. 
    Consequently, the total number of keys required is again $e^*$, representing the largest set of inputs that must be simultaneously protected in any critical configuration.
\end{itemize}

In all cases, the communication rates $R_X$ and $R_Y$ are at least $1$ (to transmit each user's input), while the source key rate $R_{Z_\Sigma}^*$ captures the minimal number of independent keys necessary to ensure security under the given server connectivity and collusion constraints. 
Thus, $R_{Z_\Sigma}^* = K-1$ when all users are critical, and $R_{Z_\Sigma}^* = e^*$ when only a subset of users is critical.
\end{remark}

\begin{theorem}\label{thm2}
For the multi-server secure aggregation problem with $U \ge 3$ servers, where each server connects to $V_u$ users $(u \in [U])$, security input sets $\{\mathcal{S}_m\}_{m \in [M]}$, and colluding user sets $\{\mathcal{T}_n\}_{n \in [N]}$,
consider the remaining case where
$e^* \le K-1, 
a^* = |\overline{\mathcal{S}}|,$ and $|\mathcal{Q}| = K.$
Then the rate region
\begin{eqnarray}
    \mathcal{R} = \{(\rx,\ry,\rzsigma) \mid R_X \ge 1, R_Y \ge 1, R_{Z_\Sigma} \ge e^* + b^*  \}
\end{eqnarray}
is achievable,
where $b^*$ is the optimal value of the following linear program.
 \begin{eqnarray}
&& \min  \sum_{(u,v)\in \mathcal{K}\setminus\overline{\mathcal{S}}} b_{u,v} \notag
\\
 \text{subject to} 
    && \sum_{(u,v)\in \mathcal{K}\setminus ((\mathcal{S}_m\cap\mathcal{K}_{\{u\}})\cup \mathcal{K}_{\mathcal{U}^{(m,n)}} \cup \mathcal{T}_n)} b_{u,v} 
    \geq 1,\forall u,m,n ~\mbox{such that}~ |\mathcal{A}_{u,m,n}| = a^*, 
    \\
    && ~~~~~~~~~~~~~~b_{u,v}\geq 0,~~~~~~~~~~~~~~~~~~~~~ \forall (u,v)\in \mathcal{K}\setminus\overline{\mathcal{S}}. 
    \label{lp}
 \end{eqnarray}
\end{theorem}

\begin{proposition}
\tit{For any multi-server secure aggregation scheme with arbitrary collusion patterns,}
\begin{eqnarray}
    R_{Z_\Sigma} \ge e^*.
\end{eqnarray}
\end{proposition}

\begin{corollary}
\tit{For the remaining case, the optimal rate satisfies}
\begin{eqnarray}
    e^* \le R_{Z_\Sigma}^* \le e^* + b^* .
\end{eqnarray}
\end{corollary}

\begin{remark}[Intuitive Explanation of Theorem~\ref{thm2}]
Theorem~\ref{thm2} addresses the remaining case in multi-server secure aggregation, 
where $e^* \le K-1$, $a^* = |\overline{\mathcal{S}}|$, and $|\mathcal{Q}| = K$. 
In this scenario, the optimal key rate cannot be determined exactly as in Theorem~\ref{thm1}, 
but an achievable upper bound is provided via a linear program.
Intuitively, the reasoning is as follows:
\begin{itemize}
    \item All users in the total security input set $\overline{\mathcal{S}}$ must be protected with key symbols, 
    since they are directly involved in the critical server-colluder combinations. This ensures the security of these explicitly critical inputs.
    \item Users outside $\overline{\mathcal{S}}$ (i.e., in $\mathcal{K} \setminus \overline{\mathcal{S}}$) may also require key symbols. 
    According to Lemma~2, for each critical configuration $(u,m,n)$ achieving the maximum overlap $a^*$, 
    at least $L$ key symbols must be associated with users not in $\mathcal{A}_{u,m,n}$ to preserve both correctness and security. 
    This is because these outside users may participate indirectly in combinations that could reveal the critical inputs if not masked.
    \item In the worst-case scenario, all users in $\mathcal{K} \setminus \overline{\mathcal{S}}$ might appear in some colluding set $\mathcal{T}_n$ (or in $\mathcal{T}_n \setminus \overline{\mathcal{S}}$). 
    Consequently, to satisfy both correctness and security, these users must contribute key symbols as determined by the linear program. 
    Since different critical configurations may involve different subsets of $\mathcal{K} \setminus \overline{\mathcal{S}}$, 
    we cannot, in general, reduce the key allocation below the value $b^*$.
    \item Therefore, the total number of key symbols required is upper-bounded by $e^* + b^*$, 
    where $e^*$ accounts for the keys for users in $\overline{\mathcal{S}}$ and $b^*$ accounts for the necessary keys for certain users outside $\overline{\mathcal{S}}$.
    The linear program ensures that every critical server-colluder combination receives sufficient masking to satisfy security and correctness.
\end{itemize}

In summary, Theorem~\ref{thm2} quantifies an achievable upper bound on the source key rate 
for the remaining case, capturing both the directly critical users and the additional keys required for correctness in the worst-case collusion scenario. 
\end{remark}

\section{Proof of Theorem \ref{thm1}}
\label{pfthm1}

In this section, we establish Theorem~\ref{thm1}, which characterizes the optimal rate region for most parameter regimes. 
The proof is divided into two parts: the converse and the achievability. 
We begin with simple examples to highlight the key ideas, and then proceed to the general scheme and its analysis.

\subsection{Converse Proof of Example~\ref{ex1}}

In general, the converse proof proceeds by lower bounding the entropy of carefully constructed subsets of source keys. Specifically, for any $k,m,n$, it suffices to show that
\begin{align}
H(\{Z_{u,v}\}_{(u,v)\in\mathcal{A}_{k,m,n}})
\ge \Big(
|((\mathcal{S}_m \cap \mathcal{K}_{\{k\}})\cup \mathcal{T}_n)\cap\overline{\mathcal{S}}|
+|\mathcal{U}^{(m,n)}\setminus\{k\}|
\Big)L.
\end{align}

We illustrate the argument using Example~\ref{ex1} (see Definition~\ref{def:secrelay11}). 
In this example, $e^*=2 \le |\overline{\mathcal{S}}|-1=2$, and hence the condition in Theorem~\ref{thm1} is satisfied.
Consider $k=1$, $m=4$, and $n=16$. In this case, $\mathcal{U}^{(4,16)}=\{2\}$ and
$((\mathcal{S}_4 \cap \mathcal{K}_{\{1\}})\cup \mathcal{T}_{16})\cap\overline{\mathcal{S}}
= \{(1,1)\}.$
By Definition~\ref{def:totset1}, we have
$\mathcal{A}_{k,m,n}
= \big((\mathcal{S}_m \cap \mathcal{K}_{\{k\}}) 
\cup \mathcal{K}_{\mathcal{U}^{(m,n)}} 
\cup \mathcal{T}_n\big)\cap\overline{\mathcal{S}}.$
Applying the chain rule of entropy, we obtain
\begin{align}
&H(\{Z_{u,v}\}_{(u,v)\in\mathcal{A}_{1,4,16}})\notag\\
\ge\;& H(Z_{1,1},Z_{1,2},Z_{2,1},Z_{2,2}, Z_{3,1},Z_{3,2})\notag\\
=\;& H(Z_{1,2},Z_{2,2}, Z_{3,1},Z_{3,2})  + H(Z_{1,1},Z_{2,1},Z_{2,2}\mid Z_{1,2},Z_{2,2}, Z_{3,1},Z_{3,2}) \notag\\
\ge\;& (1+|\mathcal{U}^{(4,16)}|)L = 2L. 
\label{eq:ex1-entropy}
\end{align}
Therefore, it remains to show that $H(Z_{1,1},Z_{2,1},Z_{2,2}\mid Z_{1,2},Z_{2,2}, Z_{3,1},Z_{3,2}) 
\ge (1+|\mathcal{U}^{(4,16)}|)L = 2L,$
which we establish next.

Intuitively, one key symbol is required to protect the input message $W_{1,1}$, and one additional key symbol is needed to protect the message associated with the surviving user in $\mathcal{U}^{(4,16)}$. We next formalize this intuition using entropy arguments.
\begin{align}
    &H(Z_{1,1},Z_{2,1},Z_{2,2}|Z_{1,2},Z_{2,2}, Z_{3,1},Z_{3,2})\notag\\
    \geq&H(Z_{1,1},Z_{2,1},Z_{2,2}|Z_{1,2},Z_{2,2}, Z_{3,1},Z_{3,2},W_{1,1},W_{2,1},W_{1,2},W_{2,2}, W_{3,1},W_{3,2})\\
    \geq&I(Z_{1,1},Z_{2,1},Z_{2,2};X_{1,1},Y_2|Z_{1,2},Z_{2,2}, Z_{3,1},Z_{3,2},W_{1,1},W_{2,1},W_{1,2},W_{2,2}, W_{3,1},W_{3,2})\\
    =&H(X_{1,1},Y_2|Z_{1,2},Z_{2,2}, Z_{3,1},Z_{3,2},W_{1,1},W_{2,1},W_{1,2},W_{2,2}, W_{3,1},W_{3,2})\\
    &-\underbrace{H(X_{1,1},Y_2|Z_{1,1},Z_{2,1},Z_{2,2},Z_{1,2},Z_{2,2}, Z_{3,1},Z_{3,2},W_{1,1},W_{2,1},W_{1,2},W_{2,2}, W_{3,1},W_{3,2})}_{\overset{(\ref{messageX})(\ref{messageY})}{=}0}\label{pfex1t1}\\
    =&H(X_{1,1},Y_2|Z_{1,2},Z_{2,2}, Z_{3,1},Z_{3,2},W_{1,2},W_{2,2}, W_{3,1},W_{3,2})\notag\\
    &-I(X_{1,1},Y_2;W_{1,1},W_{2,1}|Z_{1,2},Z_{2,2}, Z_{3,1},Z_{3,2},W_{1,2},W_{2,2}, W_{3,1},W_{3,2})\\
    =&H(X_{1,1}|Z_{1,2},Z_{2,2}, Z_{3,1},Z_{3,2},W_{1,2},W_{2,2}, W_{3,1},W_{3,2})\notag\\
    &+H(Y_2|X_{1,1},Z_{1,2},Z_{2,2}, Z_{3,1},Z_{3,2},W_{1,2},W_{2,2}, W_{3,1},W_{3,2})\notag\\
    &-I(X_{1,1},Y_2;W_{1,1},W_{2,1}|Z_{1,2},Z_{2,2}, Z_{3,1},Z_{3,2},W_{1,2},W_{2,2}, W_{3,1},W_{3,2})\\
    \geq&H(X_{1,1}|Z_{1,2},Z_{2,2}, Z_{3,1},Z_{3,2},W_{1,2},W_{2,2}, W_{3,1},W_{3,2})\notag\\
    &+H(Y_2|X_{1,1},W_{1,1},Z_{1,1},Z_{1,2},Z_{2,2}, Z_{3,1},Z_{3,2},W_{1,2},W_{2,2}, W_{3,1},W_{3,2})\notag\\
    &-I\Big(X_{1,1},Y_2,\sum_{(u,v)\in \mathcal{K}}W_{u,v};W_{1,1},W_{2,1}\Big|Z_{1,2},Z_{2,2}, Z_{3,1},Z_{3,2},W_{1,2},W_{2,2}, W_{3,1},W_{3,2}\Big)\\
    =&H(X_{1,1}|Z_{1,2},Z_{2,2}, Z_{3,1},Z_{3,2},W_{1,2},W_{2,2}, W_{3,1},W_{3,2})\notag\\
    &+H(Y_2|W_{1,1},Z_{1,1},Z_{1,2},Z_{2,2}, Z_{3,1},Z_{3,2},W_{1,2},W_{2,2}, W_{3,1},W_{3,2})\notag\\
    &-\underbrace{I\Big(\sum_{(u,v)\in \mathcal{K}}W_{u,v};W_{1,1},W_{2,1}\Big|Z_{1,2},Z_{2,2}, Z_{3,1},Z_{3,2},W_{1,2},W_{2,2}, W_{3,1},W_{3,2}\Big)}_{\overset{(\ref{ind})}{=}0}\notag\\
    &-\underbrace{I\Big(X_{1,1},Y_2;W_{1,1},W_{2,1}\Big|\sum_{(u,v)\in \mathcal{K}}W_{u,v},Z_{1,2},Z_{2,2}, Z_{3,1},Z_{3,2},W_{1,2},W_{2,2}, W_{3,1},W_{3,2}\Big)}_{\overset{(\ref{security})}{=}0}\label{pfex1t2}\\
    \geq&2L.\label{pfex1t3}
\end{align}
The equality in \eqref{pfex1t1} follows since $X_{1,1}$ is a deterministic function of $(Z_{1,1},W_{1,1})$, and $Y_2$ is a deterministic function of $(Z_{2,1},W_{2,1},Z_{2,2},W_{2,2})$. Hence, conditioned on these variables, $(X_{1,1},Y_2)$ is fully determined.
The first equality in \eqref{pfex1t2} follows from the fact that $X_{1,1}$ is determined by $(Z_{1,1},W_{1,1})$, and thus conditioning on $X_{1,1}$ can be replaced by conditioning on $(Z_{1,1},W_{1,1})$.
The third term in \eqref{pfex1t2} is zero due to the independence between the input messages and the keys (cf. \eqref{ind}). 
The last term in \eqref{pfex1t2} is zero by the security constraint associated with $\mathcal{S}_4$ and $\mathcal{T}_{16}$ (cf. \eqref{security}), which ensures that $(X_{1,1},Y_2)$ reveals no information about $(W_{1,1},W_{2,1})$ even when the aggregate $\sum_{(u,v)\in\mathcal{K}} W_{u,v}$ and the side information are given.
The first and second terms in \eqref{pfex1t2} are each lower bounded by $L$, which follows from Lemma~\ref{lemma1}.

\subsection{Achievability Proof of Example \ref{ex1}}

Before presenting the general achievable scheme, we first revisit Example~\ref{ex1} and provide an explicit construction over a small finite field.

We show that $R_{Z_\Sigma} = 2$ is achievable. Let $q = 2$, i.e., all symbols are defined over $\mathbb{F}_2$.
Consider two i.i.d. uniform random variables $(N_1,N_2) = Z_\Sigma$ over $\mathbb{F}_5$, so that $L_\Sigma = 2$, and set $L=1$. 
The key variables are defined as
\begin{align}
    Z_{1,1}=N_1,~Z_{1,3}=N_2,~Z_{2,1}=-(N_1+N_2),~Z_{1,2}=Z_{2,2}=Z_{3,1}=Z_{3,2}=0.
\end{align}
Intuitively, each key variable $Z_{u,v}$ acts as a one-time pad for the corresponding user input $W_{u,v}$ in the transmitted messages $X_{u,v}$ and $Y_u$, ensuring that even if a server observes all colluding users' messages, it learns nothing about the protected inputs beyond the sum. In other words, $X_{u,v}$ and $Y_u$ carry the masked user inputs, while $Z_{u,v}$ provides the randomness needed to hide these inputs from any unauthorized view.
The transmitted messages are given by
\begin{align}
&X_{1,1}=W_{1,1}+N_1,~X_{1,2}=W_{1,2},~X_{1,3}=W_{1,3}+N_2,\notag\\
    &X_{2,1}=W_{2,1}-(N_1+N_2),~X_{2,2}=W_{2,2},~X_{3,1}=W_{3,1},~X_{3,2}=W_{3,2},\\
    &Y_1=W_{1,1}+W_{1,2}+W_{1,3}+N_1+N_2,~Y_2=W_{2,1}+W_{2,2}-(N_1+N_2),~Y_3=W_{3,1}+W_{3,2}.
\end{align}

Correctness follows since
$\sum_{u\in [3]\setminus\{k\}} Y_u+\sum_{(u,v)\in \mathcal{K}_{\{k\}}}X_{u,v} 
= W_{1,1}+W_{1,2}+W_{1,3}+W_{2,1}+W_{2,2}+W_{3,1}+W_{3,2}.$
For security, we verify one representative case with $k=1$, $\mathcal{S}_4 = \{(1,1),(2,1)\}$, and $\mathcal{T}_{16} = \{(1,2),(2,2),(3,1),(3,2)\}$. 
The remaining cases follow similarly and are deferred to the general proof in Section~\ref{pfsecurity}.
\begin{align}
    &I(W_{1,1},W_{2,1};X_{1,1},X_{1,2},X_{1,3},Y_2,Y_3|W_{1,1}+\cdots+W_{3,2},W_{1,2},W_{2,2},W_{3,1},W_{3,2},Z_{1,2},Z_{2,2},Z_{3,1},Z_{3,2})\notag\\
    =&H(X_{1,1},X_{1,2},X_{1,3},Y_2,Y_3|W_{1,1}+\cdots+W_{3,2},W_{1,2},W_{2,2},W_{3,1},W_{3,2},Z_{1,2},Z_{2,2},Z_{3,1},Z_{3,2})\notag\\
    &-H(X_{1,1},X_{1,2},X_{1,3},Y_2,Y_3|W_{1,1}+\cdots+W_{3,2},W_{1,2},W_{2,2},W_{3,1},W_{3,2},Z_{1,2},Z_{2,2},Z_{3,1},Z_{3,2},W_{1,1},W_{2,1})\\
    \leq &H(W_{1,1}+N_{1},W_{1,3}+N_2,W_{2,1}+W_{2,2}-(N_1+N_2)|W_{1,1}+W_{1,3}+W_{2,1}+W_{2,2})\notag\\
    &-H(N_1,N_2,-(N_1+N_2)|W_{1,1},W_{1,2},W_{1,3},W_{2,1},W_{2,2},W_{3,1},W_{3,2})\label{pfexsect1}\\
    =&H(W_{1,1}+N_{1},W_{1,3}+N_2,W_{2,1}+W_{2,2}-(N_1+N_2),W_{1,1}+W_{1,3}+W_{2,1}+W_{2,2})\notag\\
    &-H(W_{1,1}+W_{1,3}+W_{2,1}+W_{2,2})-H(N_1,N_2)\label{pfexsect2}\\
    =&3-1-2=0.
\end{align}
Here, the equality in \eqref{pfexsect1} follows since $X_{1,2}$ and $Y_3$ are deterministic given the conditioning variables and hence do not contribute to the entropy. Moreover, $-(N_1+N_2)$ is a deterministic function of $(N_1,N_2)$.
The equality in \eqref{pfexsect2} follows by the chain rule of entropy. Since all variables are defined over $\mathbb{F}_5$ and are independent unless linearly constrained, each independent symbol contributes one unit of entropy. The four linear combinations in \eqref{pfexsect2} have joint entropy $3$, while $H(W_{1,1}+W_{1,3}+W_{2,1}+W_{2,2})=1$ and $H(N_1,N_2)=2$.
Therefore, the mutual information is zero, and the security constraint in \eqref{security} is satisfied.

\subsection{General Converse Proof of Theorem~\ref{thm1}}
\label{convpfthm1}

We now extend the preceding argument to general parameter regimes. 
To facilitate the analysis, we first establish several auxiliary lemmas. 

As a basic observation, each user message $X_{u,v}$ must convey at least $L$ symbols, 
which corresponds to the size of the input $W_{u,v}$, 
even if all other inputs and key variables are revealed. 
Similarly, each server message $Y_u$ must contain at least $L$ symbols 
whenever there exists at least one user $(u,v)$ whose input $W_{u,v}$ 
is not available from the side information of the other servers. 
This observation provides a simple but fundamental lower bound on the communication rates.


\begin{lemma} \label{lemma1}
\emph{For any $u\in [U]$, $(u,v) \in \mathcal{K}$, it holds that}
\begin{align}
& H\left( X_{u,v}|\{W_{i,j},Z_{i,j} \}_{(i,j)\in \mathcal{K}\backslash \{(u,v)\}}  \right)
\ge L,\label{lemma1X>=L}\\
& H\left( Y_u|\{W_{i,j},Z_{i,j} \}_{(i,j)\in \mathcal{K}\backslash \{(u,v)\}}  \right)
\ge L\label{lemma1Y>=L}.
\end{align}
\end{lemma}

\begin{IEEEproof}
A server can recover the aggregate $\sum_{(u,v)\in \mathcal{K}} W_{u,v}$ only when sufficient information about each input $W_{u,v}$ is conveyed through the user-to-server and server-to-server communications. Hence, any message carrying $W_{u,v}$ must have entropy at least $H(W_{u,v}) = L$.
Formally, for the user message $X_{u,v}$, we have
\begin{align}
& H\left( X_{u,v}|\{W_{i,j},Z_{i,j} \}_{(i,j)\in \mathcal{K}\backslash \{(u,v)\}}  \right)\notag\\
 \geq& I\left(X_{u,v};  \sum_{(i,j)\in \mathcal{K}}W_{i,j}     \bigg |    \{W_{i,j},Z_{i,j} \}_{(i,j)\in \mathcal{K}\backslash \{(u,v)\}}          \right)\\ 
  =& H\left(\sum_{(i,j)\in \mathcal{K}}W_{i,j}  \bigg |    \{W_{i,j},Z_{i,j} \}_{(i,j)\in \mathcal{K}\backslash \{(u,v)\}}   \right)  - H\left(\sum_{(i,j)\in \mathcal{K}}W_{i,j}  \bigg | X_{u,v},  \{W_{i,j},Z_{i,j} \}_{(i,j)\in \mathcal{K}\backslash \{(u,v)\}}   \right)\\
\overset{(\ref{ind})(\ref{messageX})(\ref{messageY})}{=}&    H\left(W_{u,v}   \right) - \underbrace{ H\left(\sum_{(i,j)\in \mathcal{K}}W_{i,j}  \bigg | \{X_{u,v'}\}_{v'\in [V_u]},  \{W_{i,j},Z_{i,j} \}_{(i,j)\in \mathcal{K}\backslash \{(u,v)\}} , \{Y_k\}_{k\in{[U]\setminus\{u\}}} \right)}_{ \overset{(\ref{correctness})}{=}    0}\label{pf_lemma1_t1}\\
    \overset{(\ref{h2})}{=}& L.\label{pf_lemma1_t2}
\end{align}
The equality in (\ref{pf_lemma1_t1}) holds as follows. 
For the first term, $W_{u,v}$ is independent of 
$\{W_{i,j},Z_{i,j}\}_{(i,j)\in \mathcal{K}\setminus \{(u,v)\}}$ by (\ref{ind}), 
and hence it equals $H(W_{u,v})$. 
For the second term, note that conditioning reduces entropy. 
Given $\{W_{i,j},Z_{i,j}\}_{(i,j)\in \mathcal{K}\setminus \{(u,v)\}}$, 
the remaining user messages $\{X_{u,v'}\}_{v'\in [V_u]}$ are determined by $X_{u,v}$, 
and the server messages $\{Y_k\}_{k\in [U]\setminus\{u\}}$ are also determined. 
Therefore, we can include these variables in the conditioning without increasing entropy. 
Under this conditioning, the sum $\sum_{(i,j)\in\mathcal{K}} W_{i,j}$ is fully determined by the correctness condition (\ref{correctness}), 
and hence the entropy term is zero. 
Finally, (\ref{pf_lemma1_t2}) follows since $W_{u,v}$ consists of $L$ independent and uniformly distributed symbols, 
which implies $H(W_{u,v}) = L$ as stated in (\ref{h2}).

Having obtained the entropy lower bound for each user message $X_{u,v}$ in (\ref{lemma1X>=L}), we next analyze the server messages $Y_u$. The proof of (\ref{lemma1Y>=L}) proceeds along the same line of reasoning, by relating the entropy of $Y_u$ to the information required to recover the desired sum.
\begin{eqnarray}
&& H\left( Y_u|\{W_{i,j},Z_{i,j} \}_{(i,j)\in \mathcal{K}\backslash \{(u,v)\}}  \right)\notag\\
 \geq&& I\left( Y_u ; \sum_{(i,j)\in \mathcal{K}}W_{i,j}    \bigg|\{W_{i,j},Z_{i,j}\}_{(i,j)\in \mathcal{K}\backslash \{(u,v)\}}  \right)\\
 =&& H\left( \sum_{(i,j)\in \mathcal{K}}W_{i,j}    \bigg|\{W_{i,j},Z_{i,j}\}_{(i,j)\in \mathcal{K}\backslash \{(u,v)\}}  \right)  
 - H\left( \sum_{(i,j)\in \mathcal{K}}W_{i,j}    \bigg|Y_u,\{W_{i,j},Z_{i,j}\}_{(i,j)\in \mathcal{K}\backslash \{(u,v)\}}  \right)\\
 =&& H(W_{u,v})- \underbrace{H\left( \sum_{(i,j)\in \mathcal{K}}W_{i,j}    \bigg|\{Y_k\}_{k\in [U]\setminus\{u\}},\{X_{u,v'}\}_{v'\in[V_{u}]},\{W_{i,j},Z_{i,j}\}_{(i,j)\in \mathcal{K}\backslash \{(u,v)\}}  \right)}_{\overset{(\ref{messageX}) (\ref{messageY}) (\ref{correctness})}{=}0}\label{pf_lemma1_tt2}\\
 =&&L.
\end{eqnarray}
For the second term in (\ref{pf_lemma1_tt2}), note that 
$\{X_{i,j}\}_{(i,j)\in \mathcal{K}\setminus\{(u,v)\}}$ 
is a deterministic function of 
$\{W_{i,j},Z_{i,j}\}_{(i,j)\in \mathcal{K}\setminus\{(u,v)\}}$ 
according to (\ref{messageX}). 
In particular, for any $k\neq u$, the set 
$\{X_{k,v'}\}_{v'\in [V_k]}$ is fully determined by 
$\{W_{k,j},Z_{k,j}\}_{j\in [V_k]}$, which is already included in the conditioning variables. 
Moreover, by (\ref{messageY}), the server messages 
$\{Y_k\}_{k\in [U]\setminus\{u\}}$ are functions of the corresponding user messages, 
and hence are also determined.
Therefore, given 
$\{Y_k\}_{k\in [U]\setminus\{u\}}$, 
$\{X_{u,v'}\}_{v'\in [V_u]}$, 
and $\{W_{i,j},Z_{i,j}\}_{(i,j)\in \mathcal{K}\setminus\{(u,v)\}}$, 
the sum $\sum_{(i,j)\in \mathcal{K}} W_{i,j}$ can be reconstructed by the correctness condition (\ref{correctness}). 
This implies that the conditional entropy term in (\ref{pf_lemma1_tt2}) is zero.
\end{IEEEproof}

Next, we show that, under the security and correctness constraints, the keys used by users outside the sets 
$(\mathcal{S}_m \cap \mathcal{K}_{\{u'\}}) \cup \mathcal{K}_{\mathcal{U}^{(m,n)}} \cup \mathcal{T}_n$  
must have entropy at least $L$, conditioned on the information available to the colluding users.

\begin{lemma}\label{lemma2}
\tit{For any $u' \in [U]$, $m \in [M]$, and $n \in [N]$, consider the sets 
$\mathcal{K}_{\{u'\}}$, $\mathcal{K}_{\mathcal{U}^{(m,n)}}$, $\mathcal{S}_m$, and $\mathcal{T}_n$ 
such that $(\mathcal{S}_m \cap \mathcal{K}_{\{u'\}}) \cap \mathcal{T}_n = \emptyset,
|(\mathcal{S}_m \cap \mathcal{K}_{\{u'\}}) \cup 
\mathcal{K}_{\mathcal{U}^{(m,n)}} \cup 
\mathcal{T}_n| \le K-1.$ 
Then}
\begin{align}
H\Big(
\{Z_{u,v}\}_{(u,v) \in \mathcal{K} \setminus
((\mathcal{S}_m \cap \mathcal{K}_{\{u'\}})
\cup \mathcal{K}_{\mathcal{U}^{(m,n)}}
\cup \mathcal{T}_n)}
\,\Big|\,
\{Z_{u,v}\}_{(u,v)\in\mathcal{T}_n}
\Big)
\ge L.
\label{lemma2_eqX}
\end{align}
\end{lemma}

\begin{IEEEproof}
We have
\begin{align}
&H\Big(\{Z_{u,v}\}_{(u,v) \in \mathcal{K} \setminus((\mathcal{S}_m \cap \mathcal{K}_{\{u'\}}) \cup \mathcal{K}_{\mathcal{U}^{(m,n)}} \cup \mathcal{T}_n)}\;\big|\;\{Z_{u,v}\}_{(u,v)\in\mathcal{T}_n}\Big) \notag\\
\geq&I\left(\{Z_{u,v}\}_{(u,v) \in \mathcal{K} \setminus((\mathcal{S}_m \cap \mathcal{K}_{\{u'\}}) \cup \mathcal{K}_{\mathcal{U}^{(m,n)}} \cup \mathcal{T}_n)};\{Z_{u,v}\}_{(u,v) \in(\mathcal{S}_m\cap\mathcal{K}_{\{u'\}})\cup \mathcal{K}_{\mathcal{U}^{(m,n)}}}\big|\{Z_{u,v}\}_{(u,v)\in\mathcal{T}_n}\right)\\
\overset{(\ref{ind})}{=}&I\left(\{Z_{u,v},W_{u,v}\}_{(u,v) \in \mathcal{K} \setminus((\mathcal{S}_m \cap \mathcal{K}_{\{u'\}}) \cup \mathcal{K}_{\mathcal{U}^{(m,n)}} \cup \mathcal{T}_n)};\{Z_{u,v},W_{u,v}\}_{(u,v) \in(\mathcal{S}_m\cap\mathcal{K}_{\{u'\}})\cup \mathcal{K}_{\mathcal{U}^{(m,n)}}}\big|\{Z_{u,v},W_{u,v}\}_{(u,v)\in\mathcal{T}_n}\right) \label{eq:tx12}\\
\overset{\substack{(\ref{messageX})\\(\ref{messageY})}}{\ge}&I\Big(\{X_{u,v},W_{u,v}\}_{(u,v) \in \mathcal{K} \setminus((\mathcal{S}_m \cap \mathcal{K}_{\{u'\}}) \cup \mathcal{K}_{\mathcal{U}^{(m,n)}} \cup \mathcal{T}_n)};\{W_{u,v}\}_{(u,v)\in(\mathcal{S}_m\cap\mathcal{K}_{\{u'\}})\cup \mathcal{K}_{\mathcal{U}^{(m,n)}}},\{X_{u,v}\}_{(u,v)\in(\mathcal{S}_m\cap\mathcal{K}_{\{u'\}})},\notag\\
&\{Y_u\}_{u\in \mathcal{U}^{(m,n)}}\big|\{Z_{u,v},W_{u,v}\}_{(u,v)\in\mathcal{T}_n}\Big)\label{lemma2_pf_tt1}\\
\geq&I\left(\{X_{u,v},W_{u,v}\}_{(u,v)\in \mathcal{K} \setminus((\mathcal{S}_m \cap \mathcal{K}_{\{u'\}}) \cup \mathcal{K}_{\mathcal{U}^{(m,n)}} \cup \mathcal{T}_n)};\{W_{u,v}\}_{(u,v)\in(\mathcal{S}_m\cap\mathcal{K}_{\{u'\}})\cup \mathcal{K}_{\mathcal{U}^{(m,n)}}}\big|\right.\notag\\
&\left.\{Z_{u,v},W_{u,v}\}_{(u,v)\in\mathcal{T}_n},\{X_{u,v}\}_{(u,v)\in(\mathcal{S}_m\cap\mathcal{K}_{\{u'\}})},\{Y_u\}_{u\in \mathcal{U}^{(m,n)}}\right)\\
=&H\left(\{W_{u,v}\}_{(u,v)\in(\mathcal{S}_m \cap \mathcal{K}_{\{u'\}}) \cup \mathcal{K}_{\mathcal{U}^{(m,n)}}}\big|\{Z_{u,v},W_{u,v}\}_{(u,v)\in\mathcal{T}_n},\{X_{u,v}\}_{(u,v)\in(\mathcal{S}_m\cap\mathcal{K}_{\{u'\}})},\{Y_u\}_{u\in \mathcal{U}^{(m,n)}}\right)\notag \\
&-~H\left(\{W_{u,v}\}_{(u,v)\in(\mathcal{S}_m \cap \mathcal{K}_{\{u'\}}) \cup \mathcal{K}_{\mathcal{U}^{(m,n)}}}\big|\{Z_{u,v},W_{u,v}\}_{(u,v)\in\mathcal{T}_n},\{X_{u,v}\}_{(u,v)\in(\mathcal{S}_m\cap\mathcal{K}_{\{u'\}})},\{Y_u\}_{u\in \mathcal{U}^{(m,n)}},\right.\notag\\
&\left.\{X_{u,v},W_{u,v}\}_{(u,v)\in \mathcal{K} \setminus((\mathcal{S}_m \cap \mathcal{K}_{\{u'\}}) \cup \mathcal{K}_{\mathcal{U}^{(m,n)}} \cup \mathcal{T}_n)}\right)\\
\geq&H\left(\{W_{u,v}\}_{(u,v)\in(\mathcal{S}_m \cap \mathcal{K}_{\{u'\}}) \cup \mathcal{K}_{\mathcal{U}^{(m,n)}}}\Big|\sum_{(u,v)\in \mathcal{K}}W_{u,v},\{Z_{u,v},W_{u,v}\}_{(u,v)\in\mathcal{T}_n},\{X_{u,v}\}_{(u,v)\in(\mathcal{S}_m\cap\mathcal{K}_{\{u'\}})},\{Y_u\}_{u\in \mathcal{U}^{(m,n)}}\right)\notag \\
&-~H\left(\{W_{u,v}\}_{(u,v)\in((\mathcal{S}_m \cap \mathcal{K}_{\{u'\}}) \cup \mathcal{K}_{\mathcal{U}^{(m,n)}})\setminus\mathcal{T}_n}\Big|\sum_{(u,v)\in ((\mathcal{S}_m \cap \mathcal{K}_{\{u'\}}) \cup \mathcal{K}_{\mathcal{U}^{(m,n)}})\setminus\mathcal{T}_n}W_{u,v}\right)\\
=&~~~~H\left(\{W_{u,v}\}_{(u,v)\in(\mathcal{S}_m \cap \mathcal{K}_{\{u'\}}) \cup \mathcal{K}_{\mathcal{U}^{(m,n)}}}\Big|\sum_{(u,v)\in \mathcal{K}}W_{u,v},\{Z_{u,v},W_{u,v}\}_{(u,v)\in\mathcal{T}_n}\right)-\notag \\
&\underbrace{I\left(\{W_{u,v}\}_{(u,v)\in(\mathcal{S}_m \cap \mathcal{K}_{\{u'\}}) \cup \mathcal{K}_{\mathcal{U}^{(m,n)}}};\{X_{u,v}\}_{(u,v)\in(\mathcal{S}_m\cap\mathcal{K}_{\{u'\}})},\{Y_u\}_{u\in \mathcal{U}^{(m,n)}}\Big|\sum_{(u,v)\in \mathcal{K}}W_{u,v},\{Z_{u,v},W_{u,v}\}_{(u,v)\in\mathcal{T}_n}\right)}_{\overset{(\ref{security})}{=}0}\notag \\
&-~H\left(\{W_{u,v}\}_{(u,v)\in((\mathcal{S}_m \cap \mathcal{K}_{\{u'\}}) \cup \mathcal{K}_{\mathcal{U}^{(m,n)}})\setminus\mathcal{T}_n}\Big|\sum_{(u,v)\in ((\mathcal{S}_m \cap \mathcal{K}_{\{u'\}}) \cup \mathcal{K}_{\mathcal{U}^{(m,n)}})\setminus\mathcal{T}_n}W_{u,v}\right)\label{pf_lemma2_t1}\\
=&H\left(\{W_{u,v}\}_{(u,v)\in((\mathcal{S}_m \cap \mathcal{K}_{\{u'\}}) \cup \mathcal{K}_{\mathcal{U}^{(m,n)}})\setminus\mathcal{T}_n}\Big|\sum_{(u,v)\in \mathcal{K}}W_{u,v},\{Z_{u,v},W_{u,v}\}_{(u,v)\in\mathcal{T}_n}\right)\notag \\
&-~H\left(\{W_{u,v}\}_{(u,v)\in((\mathcal{S}_m \cap \mathcal{K}_{\{u'\}}) \cup \mathcal{K}_{\mathcal{U}^{(m,n)}})\setminus\mathcal{T}_n}\Big|\sum_{(u,v)\in ((\mathcal{S}_m \cap \mathcal{K}_{\{u'\}}) \cup \mathcal{K}_{\mathcal{U}^{(m,n)}})\setminus\mathcal{T}_n}W_{u,v}\right)\label{pf_lemma2_t3}\\
\overset{(\ref{h2})}{=}&|((\mathcal{S}_m \cap \mathcal{K}_{\{u'\}}) \cup \mathcal{K}_{\mathcal{U}^{(m,n)}})\setminus\mathcal{T}_n|L-(|((\mathcal{S}_m \cap \mathcal{K}_{\{u'\}}) \cup \mathcal{K}_{\mathcal{U}^{(m,n)}})\setminus\mathcal{T}_n|-1)L\\
=&L.
\end{align}
The inequality in (\ref{lemma2_pf_tt1}) follows from the functional dependence of each message on its local input and key, together with the data processing inequality: for $(u,v)\in(\mathcal{S}_m\cap\mathcal{K}_{\{u'\}})$, each $X_{u,v}$ is a deterministic function of $(W_{u,v},Z_{u,v})$ by (\ref{messageX}), and for $(u,v)\in\mathcal{K}_{\mathcal{U}^{(m,n)}}$, the user messages are determined by their inputs and keys. Each server message $Y_u$, $u\in\mathcal{U}^{(m,n)}$, is generated from the associated user messages by (\ref{messageY}). Replacing $(Z_{u,v},W_{u,v})$ with $(X_{u,v},W_{u,v})$ and including $\{Y_u\}_{u\in\mathcal{U}^{(m,n)}}$ cannot increase the mutual information.  
The mutual information term in (\ref{pf_lemma2_t1}) vanishes due to the security requirement in (\ref{security}): under the admissible side information $\mathcal{T}_n$, the messages $\{X_{u,v}\}_{(u,v)\in(\mathcal{S}_m\cap\mathcal{K}_{\{u'\}})}$ and $\{Y_u\}_{u\in \mathcal{U}^{(m,n)}}$ reveal no additional information about the protected inputs beyond the aggregate sum.  
Therefore, the conditional entropy in (\ref{lemma2_eqX}) is lower bounded by $L$, as claimed.
\end{IEEEproof}

By specializing Lemma~\ref{lemma2} to the case where 
$\left|(\mathcal{S}_m \cap \mathcal{K}_{\{u'\}}) \cup \mathcal{K}_{\mathcal{U}^{(m,n)}} \cup \mathcal{T}_n\right| = K-1$, 
we obtain a direct consequence for the indices belonging to the implicit security input set $\mathcal{S}_{I}$ (see Definition~\ref{def:imp}).

\begin{corollary} \label{corol1}
\tit{Fix $u' \in [U]$, $m \in [M]$, and $n \in [N]$. 
Consider the sets $\mathcal{K}_{\{u'\}}$, $\mathcal{K}_{\mathcal{U}^{(m,n)}}$, $\mathcal{S}_m$, and $\mathcal{T}_n$ such that
$\left|(\mathcal{S}_m \cap \mathcal{K}_{\{u'\}}) \cup \mathcal{K}_{\mathcal{U}^{(m,n)}} \cup \mathcal{T}_n\right| = K-1.$
Let the unique remaining index be
$\{(u,v)\} = \mathcal{K} \setminus \big((\mathcal{S}_m \cap \mathcal{K}_{\{u'\}}) \cup \mathcal{K}_{\mathcal{U}^{(m,n)}} \cup \mathcal{T}_n\big)$.
Then the entropy of the corresponding key satisfies}
\begin{eqnarray}
    H\left(Z_{u,v}\,\big|\,\{Z_{i,j}\}_{(i,j)\in\mathcal{T}_n}\right) 
    \overset{(\ref{lemma2_eqX})}{\geq} L. \label{corollary1_eqX}
\end{eqnarray}
\end{corollary}



Lemma~\ref{lemma1} provides a lower bound on the entropy of individual messages. 
We now extend this result to sets of messages, taking into account the inputs and keys known to colluding users.

\begin{lemma}\label{lemma3}
\tit{For any $u'\in[U]$, $\mathcal{U}\subseteq[U]$, and $\mathcal{V}\subseteq[V_{u'}]$, 
such that
$(\mathcal{S}_m \cap \mathcal{K}_{\{u'\}}) \cap \mathcal{T}_n = \emptyset$ and
$|(\mathcal{S}_m \cap \mathcal{K}_{\{u'\}}) 
\cup \mathcal{K}_{\mathcal{U}^{(m,n)}} 
\cup \mathcal{T}_n| \le K-1$, we have}
\begin{align}
&H\left(\{X_{u',v}\}_{v\in\mathcal V}\,\middle|\,
\{W_{i,j},Z_{i,j}\}_{(i,j)\in \mathcal{T}_n}\right)\ge |\mathcal V|L,
\label{lemma3xvwz}\\
&H\left(\{Y_u\}_{u\in \mathcal{U}} \,\middle|\,
\{X_{u',v}\}_{v\in [V_{u'}]},
\{W_{i,j},Z_{i,j}\}_{(i,j)\in \mathcal{T}_n}\right)
\ge |\mathcal{U}\setminus \{u'\}|L.
\label{lemma3yuwz}
\end{align}
\end{lemma}

\begin{IEEEproof}
For (\ref{lemma3xvwz}), we have
    \begin{align}
        &H\left(\{X_{u',v}\}_{v\in\Vc}\,\middle|\,\{W_{i,j}\}_{(i,j)\in \mathcal{T}_n},\{Z_{i,j}\}_{(i,j)\in \mathcal{T}_n}\right)\notag\\
        =& \sum_{v\in \Vc  }H\left( X_{u',v}| \{X_{u',l}\}_{l\in \Vc\backslash \{v\},l<v }    ,\{W_{i,j},Z_{i,j}\}_{(i,j)\in \mathcal{T}_n}\right)\\
 \ge& \sum_{v\in \Vc  }H\left( X_{u',v}| \{X_{u',l}\}_{l\in \Vc\backslash \{v\} }    ,\{W_{i,j},Z_{i,j}\}_{(i,j)\in \mathcal{T}_n}\right)\\
\ge &\sum_{v\in \Vc  }H\left( X_{u',v}|\{W_{u',l}, Z_{u',l}\}_{l\in \Vc\backslash \{v\} }, \{X_{u',l}\}_{l\in \Vc\backslash \{v\} }    ,\{W_{i,j},Z_{i,j}\}_{(i,j)\in \mathcal{T}_n}\right)\\
 \overset{(\ref{messageX})}{=}& \sum_{v\in \Vc  }H\left( X_{u',v}|\{W_{u',l}, Z_{u',l}\}_{l\in \Vc\backslash \{v\} }   ,\{W_{i,j},Z_{i,j}\}_{(i,j)\in \mathcal{T}_n}\right)\label{lemma2ttt0}\\
\geq& \sum_{v\in \Vc  }H\left( X_{u',v}|\{W_{u',l}, Z_{u',l}\}_{(u',l)\in \mathcal{K}\backslash \{(u',v)\} } \right)\label{lemma2ttt1}\\
 \overset{(\ref{lemma1X>=L})}{\ge}&|\Vc|L,
    \end{align}
where (\ref{lemma2ttt1}) follows from the fact that conditioning reduces entropy. For any $v\in\Vc$, we have $\{W_{u',l}, Z_{u',l}\}_{l\in \Vc\setminus \{v\}}$ $ \subseteq \{W_{u',l}, Z_{u',l}\}_{l\in [V]\setminus \{v\}} \subseteq \{W_{i,j}, Z_{i,j}\}_{(i,j)\in \mathcal{K}\setminus \{(u',v)\}}$, and since $(u',v)\notin \mathcal{T}_n$, it also holds that $\{W_{i,j},Z_{i,j}\}_{(i,j)\in \mathcal{T}_n}$ $ \subseteq \{W_{i,j}, Z_{i,j}\}_{(i,j)\in \mathcal{K}\setminus \{(u',v)\}}$. Therefore, the conditioning set $\{W_{u',l}, Z_{u',l}\}_{l\in \Vc\backslash \{v\} }   \cup\{W_{i,j},Z_{i,j}\}_{(i,j)\in \mathcal{T}_n}$ in (\ref{lemma2ttt0}) is a subset of $\{W_{i,j}, Z_{i,j}\}_{(i,j)\in \mathcal{K}\setminus \{(u',v)\}}$, which yields (\ref{lemma2ttt1}).

For (\ref{lemma3yuwz}), we have
\begin{align}
    &H\left(\{Y_u\}_{u\in \mathcal{U}} \middle|\{X_{u',v}\}_{v\in [V_{u'}]},\{Z_{i,j}\}_{(i,j)\in \mathcal{T}_n}\right)\notag\\
    =&H\left(\{Y_u\}_{u\in \mathcal{U}\setminus \{u'\}} \middle|\{X_{u',v}\}_{v\in [V_{u'}]},Y_{u'},\{Z_{i,j}\}_{(i,j)\in \mathcal{T}_n}\right)\\
    = & \sum_{u\in \mathcal{U}\setminus \{u'\}}H\left( Y_{u}| \{Y_{k}\}_{k\in \mathcal{U}\setminus \{u\},k<u } ,\{X_{u',v}\}_{v\in [V_{u'}]},Y_{u'},\{W_{i,j},Z_{i,j}\}_{(i,j)\in \mathcal{T}_n}\right) \\
    \geq & \sum_{u\in \mathcal{U}\setminus \{u'\}}H\left( Y_{u}| \{Y_{k}\}_{k\in \mathcal{U}\setminus\{u\} } ,\{X_{u',v}\}_{v\in [V_{u'}]},Y_{u'},\{W_{i,j},Z_{i,j}\}_{(i,j)\in \mathcal{T}_n}\right) \\
    \geq & \sum_{u\in \mathcal{U}\setminus \{u'\}}H\big( Y_{u}|\{W_{k,v},Z_{k,v}\}_{(k,v)\in \mathcal{K}_{\mathcal{U}\setminus \{u\}}},  \{Y_{k}\}_{k\in \mathcal{U}\setminus\{u\} } ,\{W_{u',v},Z_{u',v}\}_{v\in [V_{u'}]},\{X_{u',v}\}_{v\in [V_{u'}]},Y_{u'},\notag\\
    &\{W_{i,j},Z_{i,j}\}_{(i,j)\in \mathcal{T}_n}\big) \\
    \overset{(\ref{messageX})(\ref{messageY})}{=}& \sum_{u\in \mathcal{U}\setminus \{u'\}}H\left( Y_{u}|\{W_{k,v},Z_{k,v}\}_{(k,v)\in \mathcal{K}_{\mathcal{U}\setminus \{u\}}},\{W_{u',v},Z_{u',v}\}_{v\in [V_{u'}]},\{W_{i,j},Z_{i,j}\}_{(i,j)\in \mathcal{T}_n}\right) \label{lemma2ttt22}\\
    \geq &\sum_{u\in \mathcal{U}\setminus \{u'\}}H\left( Y_{u}|\{W_{k,v},Z_{k,v}\}_{(k,v)\in \mathcal{K}\setminus \{(u,v')\}} \right)\label{lemma2ttt2}\\
    \overset{(\ref{lemma1Y>=L})}{\geq}&|\mathcal{U}\setminus\{u'\}|L,
\end{align}
where (\ref{lemma2ttt2}) follows from the fact that conditioning reduces entropy. For any $u\in\mathcal{U}\setminus\{u'\}$ and $k\in \mathcal{U}\setminus\{u\}$, the variables $\{W_{k,v},Z_{k,v}\}_{(k,v)\in \mathcal{K}_{\mathcal{U}\setminus \{u\}}} \subseteq \{W_{k,v}, Z_{k,v}\}_{(k,v)\in \mathcal{K}\setminus \{(u,v')\}}$. Moreover, since $u'\neq u$, it holds that $\{W_{u',v}, Z_{u',v}\}_{v\in [V_{u'}]} \subseteq \{W_{k,v}, Z_{k,v}\}_{(k,v)\in \mathcal{K}\setminus \{(u,v')\}}$. Finally, as $(u,v')\notin \mathcal{T}_n$, we also have $\{W_{i,j}, Z_{i,j}\}_{(i,j)\in \mathcal{T}_n} \subseteq \{W_{i,j}, Z_{i,j}\}_{(i,j)\in \mathcal{K}\setminus \{(u,v')\}}$. Therefore, the conditioning set in (\ref{lemma2ttt22}) is a subset of $\{W_{k,v}, Z_{k,v}\}_{(k,v)\in \mathcal{K}\setminus \{(u,v')\}}$, which yields (\ref{lemma2ttt2}).
\end{IEEEproof}

Next, we extend the entropy lower bounds to the elements in the set $\mathcal{T}_n \cap \overline{\mathcal{S}}$, both individually and collectively.

\begin{lemma}
\label{lemma4ZV>=VL}
\tit{For each element $(u',v') \in \mathcal{T}_n \cap \overline{\mathcal{S}}$, we have}
\begin{align}
H\left(Z_{u',v'}\big|\{Z_{u,v}\}_{(u,v)\in \mathcal{T}_n\setminus\{(u',v')\}}\right)&\geq L, \label{lemma4zv1} \\
H\left(\{Z_{u,v}\}_{(u,v) \in \mathcal{T}_n\cap\overline{\mathcal{S}}}\big|\{Z_{u,v}\}_{(u,v)\in\mathcal{T}_n\setminus \overline{\mathcal{S}}}\right)&\ge |\mathcal{T}_n\cap\overline{\mathcal{S}}| L. \label{lemma4zv} 
\end{align}
\end{lemma}

\begin{IEEEproof}
We now prove (\ref{lemma4zv1}) by bounding the entropy of each individual element in $\mathcal{T}_n \cap \overline{\mathcal{S}}$:
\begin{align}
    &H\left(Z_{u',v'}\big|\{Z_{u,v}\}_{(u,v)\in \mathcal{T}_n\setminus\{(u',v')\}}\right)\\
    \geq&H\left(Z_{u',v'}\big|\{Z_{u,v}\}_{(u,v)\in \mathcal{T}_n\setminus\{(u',v')\}},\{W_{u,v}\}_{(u,v)\in \mathcal{T}_n}\right)\\
    \geq&I\left(Z_{u',v'};X_{u',v'}\big|\{Z_{u,v}\}_{(u,v)\in \mathcal{T}_n\setminus\{(u',v')\}},\{W_{u,v}\}_{(u,v)\in \mathcal{T}_n}\right)\\
    =&H\left(X_{u',v'}\big|\{Z_{u,v}\}_{(u,v)\in \mathcal{T}_n\setminus\{(u',v')\}},\{W_{u,v}\}_{(u,v)\in \mathcal{T}_n}\right)\\
    =&H\left(X_{u',v'}\big|\{Z_{u,v},W_{u,v}\}_{(u,v)\in \mathcal{T}_n\setminus\{(u',v')\}}\right)\notag\\
    &-I\left(X_{u',v'};W_{u',v'}\big|\{Z_{u,v},W_{u,v}\}_{(u,v)\in \mathcal{T}_n\setminus\{(u',v')\}}\right)\\
    \overset{(\ref{lemma1X>=L})}{\geq}&L-I\left(X_{u',v'},\sum_{(u,v)\in \mathcal{K}}W_{u,v};W_{u',v'}\Bigg|\{Z_{u,v},W_{u,v}\}_{(u,v)\in \mathcal{T}_n\setminus\{(u',v')\}}\right)\\
    =&L-\underbrace{I\left(\sum_{(u,v)\in \mathcal{K}}W_{u,v};W_{u',v'}\Bigg|\{Z_{u,v},W_{u,v}\}_{(u,v)\in \mathcal{T}_n\setminus\{(u',v')\}}\right)}_{\overset{(\ref{ind})}{=}0}\notag\\
    &-\underbrace{I\left(X_{u',v'};W_{u',v'}\Bigg|\sum_{(u,v)\in \mathcal{K}}W_{u,v},\{Z_{u,v},W_{u,v}\}_{(u,v)\in \mathcal{T}_n\setminus\{(u',v')\}}\right)}_{\overset{(\ref{security})}{=}0}\label{lem4t1}\\
    =&L.
\end{align} 
Where the third term in (\ref{lem4t1}) is zero since $(u',v') \in \mathcal{T}_n \cap \overline{\mathcal{S}}$. 
By definition of $\overline{\mathcal{S}}$, there exists $m'$ such that $(u',v') \in \mathcal{S}_{m'}$, 
and the security constraint (\ref{security}) implies that the above mutual information is zero.

For (\ref{lemma4zv}), we have
\begin{eqnarray}
&&H\left(\{Z_{u,v}\}_{(u,v) \in \mathcal{T}_n\cap\overline{\mathcal{S}}}\big|\{Z_{u,v}\}_{(u,v)\in\mathcal{T}_n\setminus \overline{\mathcal{S}}}\right)\notag\\
 &\geq&\sum_{u',v': (u',v')\in \mathcal{T}_n\cap\overline{\mathcal{S}}}H\left(Z_{u',v'}\big|\{Z_{u,v}\}_{(u,v)\in\mathcal{T}_n\setminus \overline{\mathcal{S}}},\{Z_{u,v}\}_{(u,v)\in (\mathcal{T}_n\cap\overline{\mathcal{S}})\setminus\{(u',v')\}}\right) \\
&\geq&\sum_{u',v': (u',v')\in \mathcal{T}_n\cap\overline{\mathcal{S}}}H\left(Z_{u',v'}\big|\{Z_{u,v}\}_{(u,v)\in (\mathcal{T}_n)\setminus\{(u',v')\}}\right) \\
&\overset{(\ref{lemma4zv1})}{\geq}&|\mathcal{T}_n\cap\overline{\mathcal{S}}|L.\label{pf_lemma4_t4}
\end{eqnarray}
The last step follows by applying (\ref{lemma4zv1}) to each $(u',v')\in \mathcal{T}_n\cap\overline{\mathcal{S}}$.
\end{IEEEproof}

We now consider arbitrary sets $\mathcal{S}_m$ and $\mathcal{T}_n$. 
Intuitively, the total entropy of the keys associated with the explicit 
security input set $\mathcal{S}_m \cap \mathcal{K}_{\{u'\}}$, given the 
information available to colluding users, must be at least 
$|\mathcal{S}_m \cap \mathcal{K}_{\{u'\}}|L$. 
Similarly, for the explicit security input set 
$\mathcal{K}_{\mathcal{U}^{(m,n)}}$, the required total entropy of the keys 
is at least $|\mathcal{U}^{(m,n)}\setminus \{u'\}|L$.

\begin{lemma} \label{lemma5} 
\tit{Let $u' \in [U]$, $m \in [M]$, and $n \in [N]$. 
Consider the sets $\mathcal{K}_{\{u'\}}, \mathcal{K}_{\mathcal{U}^{(m,n)}}, 
\mathcal{S}_m$, and $\mathcal{T}_n$ such that 
$(\mathcal{S}_m \cap \mathcal{K}_{\{u'\}}) \cap \mathcal{T}_n = \emptyset$ 
and 
$|(\mathcal{S}_m \cap \mathcal{K}_{\{u'\}}) 
\cup \mathcal{K}_{\mathcal{U}^{(m,n)}} 
\cup \mathcal{T}_n| \le K-1$. 
Then}
\begin{align}
H\big(\{Z_{u,v}\}_{(u,v) \in (\mathcal{S}_m \cap \mathcal{K}_{\{u'\}}) 
\cup \mathcal{K}_{\mathcal{U}^{(m,n)}}} 
\big| 
\{Z_{u,v}\}_{(u,v)\in \mathcal{T}_n}\big) 
\ge 
\left(|\mathcal{S}_m \cap \mathcal{K}_{\{u'\}}|
+|\mathcal{U}^{(m,n)}\setminus \{u'\}|\right)L.
\label{lemma5_eqY}
\end{align}
\end{lemma}
\begin{IEEEproof}
We have
\begin{align}
& H\Big(\{Z_{u,v}\}_{(u,v) \in (\mathcal{S}_m \cap \mathcal{K}_{\{u'\}}) \cup \mathcal{K}_{\mathcal{U}^{(m,n)}}} \big| \{Z_{u,v}\}_{(u,v)\in \mathcal{T}_n}\Big)\\
\geq&I\Big(\{Z_{u,v}\}_{(u,v) \in (\mathcal{S}_m \cap \mathcal{K}_{\{u'\}}) \cup \mathcal{K}_{\mathcal{U}^{(m,n)}}};\{X_{u',v}\}_{(u',v) \in (\mathcal{S}_m \cap \mathcal{K}_{\{u'\}})},\{Y_u\}_{u\in \mathcal{U}^{(m,n)}} \big|\notag\\
&\{Z_{u,v}\}_{(u,v)\in\mathcal{T}_n},\{W_{u,v}\}_{(u,v) \in (\mathcal{S}_m \cap \mathcal{K}_{\{u'\}}) \cup \mathcal{K}_{\mathcal{U}^{(m,n)}} \cup \mathcal{T}_n}\Big) ~~~~\\
\overset{(\ref{messageX})}{=}&H\left(\{X_{u',v}\}_{(u',v) \in (\mathcal{S}_m \cap \mathcal{K}_{\{u'\}})},\{Y_u\}_{u\in \mathcal{U}^{(m,n)}}\big|\{Z_{u,v}\}_{(u,v)\in\mathcal{T}_n},\{W_{u,v}\}_{(u,v) \in (\mathcal{S}_m \cap \mathcal{K}_{\{u'\}}) \cup \mathcal{K}_{\mathcal{U}^{(m,n)}} \cup \mathcal{T}_n}\right) \label{pf_lemma5_t1}\\
=&H\Big(\{X_{u',v}\}_{(u',v) \in (\mathcal{S}_m \cap \mathcal{K}_{\{u'\}})},\{Y_u\}_{u\in \mathcal{U}^{(m,n)}}\big|\{Z_{u,v},W_{u,v}\}_{(u,v)\in\mathcal{T}_n}\Big)\notag\\
&-I\Big(\{X_{u',v}\}_{(u',v) \in (\mathcal{S}_m \cap \mathcal{K}_{\{u'\}})},\{Y_u\}_{u\in \mathcal{U}^{(m,n)}};\{W_{u,v}\}_{(u,v) \in (\mathcal{S}_m \cap \mathcal{K}_{\{u'\}}) \cup \mathcal{K}_{\mathcal{U}^{(m,n)}}}\Big|\{Z_{u,v},W_{u,v}\}_{(u,v)\in\mathcal{T}_n}\Big)\\
\geq&H\Big(\{X_{u',v}\}_{(u',v) \in (\mathcal{S}_m \cap \mathcal{K}_{\{u'\}})}\big|\{Z_{u,v},W_{u,v}\}_{(u,v)\in\mathcal{T}_n}\Big)\notag\\
&+H\Big(\{Y_u\}_{u\in \mathcal{U}^{(m,n)}}\big|\{X_{u',v}\}_{(u',v) \in (\mathcal{S}_m \cap \mathcal{K}_{\{u'\}})},\{Z_{u,v},W_{u,v}\}_{(u,v)\in\mathcal{T}_n}\Big)-\notag\\
&I\Bigg(\{X_{u',v}\}_{(u',v) \in (\mathcal{S}_m \cap \mathcal{K}_{\{u'\}})},\{Y_u\}_{u\in \mathcal{U}^{(m,n)}},\sum_{(u,v)\in \mathcal{K}}W_{u,v};\{W_{u,v}\}_{(u,v) \in (\mathcal{S}_m \cap \mathcal{K}_{\{u'\}}) \cup \mathcal{K}_{\mathcal{U}^{(m,n)}}}\Bigg|\{Z_{u,v},W_{u,v}\}_{(u,v)\in\mathcal{T}_n}\Bigg)\\
\geq&H\Big(\{X_{u',v}\}_{(u',v) \in (\mathcal{S}_m \cap \mathcal{K}_{\{u'\}})}\big|\{Z_{u,v},W_{u,v}\}_{(u,v)\in\mathcal{T}_n}\Big)\notag\\
&+H\Big(\{Y_u\}_{u\in \mathcal{U}^{(m,n)}}\big|\{X_{u',v}\}_{v \in [V_{u'}]},\{Z_{u,v},W_{u,v}\}_{(u,v)\in\mathcal{T}_n}\Big)\notag\\
&-\underbrace{I\Bigg(\sum_{(u,v)\in \mathcal{K}}W_{u,v};\{W_{u,v}\}_{(u,v) \in (\mathcal{S}_m \cap \mathcal{K}_{\{u'\}}) \cup \mathcal{K}_{\mathcal{U}^{(m,n)}}}\Bigg|\{Z_{u,v},W_{u,v}\}_{(u,v)\in\mathcal{T}_n}\Bigg)}_{\overset{(\ref{ind})}{=}0}-\notag\\
&\underbrace{I\Bigg(\{X_{u',v}\}_{(u',v) \in (\mathcal{S}_m \cap \mathcal{K}_{\{u'\}})},\{Y_u\}_{u\in \mathcal{U}^{(m,n)}};\{W_{u,v}\}_{(u,v) \in (\mathcal{S}_m \cap \mathcal{K}_{\{u'\}}) \cup \mathcal{K}_{\mathcal{U}^{(m,n)}}}\Bigg|\sum_{(u,v)\in \mathcal{K}}W_{u,v},\{Z_{u,v},W_{u,v}\}_{(u,v)\in\mathcal{T}_n}\Bigg)}_{\overset{(\ref{security})}{=}0}\label{pf_lemma5_t2}\\
\geq&\left(|\mathcal{S}_m \cap \mathcal{K}_{\{u'\}}|
+|\mathcal{U}^{(m,n)}\setminus \{u'\}|\right)L.
\end{align}
(\ref{pf_lemma5_t1}) holds because the messages 
$\{X_{u',v}\}_{(u',v)\in (\mathcal{S}_m \cap \mathcal{K}_{\{u'\}})}$ 
are deterministic functions of the corresponding inputs and keys, i.e., 
$\{W_{u',v}, Z_{u',v}\}_{(u',v)\in (\mathcal{S}_m \cap \mathcal{K}_{\{u'\}})}$, 
and the server messages $\{Y_u\}_{u\in \mathcal{U}^{(m,n)}}$ are deterministic 
functions of $\{X_{u,v}\}_{(u,v)\in \mathcal{K}_{\mathcal{U}^{(m,n)}}}$.
In (\ref{pf_lemma5_t2}), the first two terms follow from the chain rule of 
entropy together with the lower bounds established in Lemma~\ref{lemma3} 
((\ref{lemma3xvwz}) and (\ref{lemma3yuwz})), respectively. 
The third term is zero due to the independence between inputs and keys 
(see (\ref{ind})), while the fourth term is zero due to the security 
constraint associated with $\mathcal{S}_m$ and $\mathcal{T}_n$ 
(see (\ref{security})).
\end{IEEEproof}

Equipped with the Lemmas, the converse bounds on the communication rates $R_X$, $R_Y$, and the source key rate $R_{Z_\Sigma}$ follow immediately.

\subsubsection{Proof of $R_X \ge 1$}
For any $(u,v) \in \mathcal{K}$, we have
\begin{align}
L_X &= H(X_{u,v}) \ge H\Big(X_{u,v} \,\big|\, \{W_{i,j},Z_{i,j}\}_{(i,j)\in \mathcal{K}\backslash \{(u,v)\}} \Big)
\overset{(\ref{lemma1X>=L})}{\ge} L,\\
\Rightarrow \quad R_X &= \frac{L_X}{L} \ge 1.
\end{align}

\subsubsection{Proof of $R_Y \ge 1$}
For any $u \in [U]$, we similarly have
\begin{align}
L_Y &= H(Y_u) \ge H\Big(Y_u \,\big|\, \{W_{i,j},Z_{i,j}\}_{(i,j)\in \mathcal{K}\backslash \{(u,v)\}} \Big)
\overset{(\ref{lemma1Y>=L})}{\ge} L,\\
\Rightarrow \quad R_Y &= \frac{L_Y}{L} \ge 1.
\end{align}

Note that these communication rate lower bounds hold regardless of the security constraints. Since the server must recover the sum of all users' inputs, the information corresponding 
each user's input of size $L$ must be transmitted through the corresponding communication links
Consequently, $R_X \ge 1$ represents the minimum required rate on the user-to-server links, whereas $R_Y \ge 1$ represents the minimum required rate on the server-to-server links.

\subsubsection{Proof of $R_{Z_{\Sigma}}\ge \min\{e^*, K-1\}$}

We are now ready to show that $R_{Z_\Sigma} \geq  K-1$ when $e^*=K$, and $R_{Z_\Sigma} \geq e^*$ when $e^* \le K-1$ and $a^* \le |\overline{\mathcal{S}}|-1$ or $e^* \le K-1$, $a^* = |\overline{\mathcal{S}}|$, and $|\mathcal{Q}| \le K-1$. 
First, we show that $R_{Z_\Sigma} \geq e^*$.
Consider arbitrary $\mathcal{K}_{\{u'\}}$ and $\mathcal{U}^{(m,n)}$ 
whose corresponding sets $\mathcal{S}_m$ and $\mathcal{T}_n$ satisfy 
the conditions in Lemmas~\ref{lemma4ZV>=VL} and \ref{lemma5} (note that the set systems are monotone so that we may assume without loss that $(\mathcal{S}_m \cap \mathcal{K}_{\{u'\}}) \cap \mathcal{T}_n = \emptyset$).
\begin{eqnarray}
    H(Z_{\Sigma})&\overset{(\ref{sourcekey})}{\geq}&H\left(\{Z_{u,v}\}_{(u,v)\in \mathcal{K}}\right)\\
    &=&H\left(\{Z_{u,v}\}_{(u,v)\in \mathcal{T}_n}\right)+H\left(\{Z_{u,v}\}_{(u,v)\in \mathcal{K}}|\{Z_{u,v}\}_{(u,v)\in \mathcal{T}_n}\right)\\
    &\geq&H\left(\{Z_{u,v}\}_{(u,v) \in \mathcal{T}_n\cap\overline{\mathcal{S}}}\big|\{Z_{u,v}\}_{(u,v)\in\mathcal{T}_n\setminus \overline{\mathcal{S}}}\right)\notag\\
    &&+H\big(\{Z_{u,v}\}_{(u,v) \in (\mathcal{S}_m \cap \mathcal{K}_{\{u'\}}) \cup \mathcal{K}_{\mathcal{U}^{(m,n)}}} \big| \{Z_{u,v}\}_{(u,v)\in \mathcal{T}_n}\big)~~\\
    &\overset{(\ref{lemma4zv})(\ref{lemma5_eqY})}{\geq}& \big|\mathcal{T}_{n}\cap\overline{\mathcal{S}}\big| L+ |\mathcal{S}_m \cap \mathcal{K}_{\{u'\}}|L+|\mathcal{U}^{(m,n)}\setminus \{u'\}|L\\
    &=&\Big( | ((\mathcal{S}_m \cap \mathcal{K}_{\{u'\}})\cup \mathcal{T}_n)\cap\overline{\mathcal{S}}|+|\mathcal{U}^{(m,n)}\setminus \{u'\}|\Big)L \label{totalkey_pf_t1}\\
\Rightarrow &&H(Z_{\Sigma}) \geq \max_{u',m,n} \Big( | ((\mathcal{S}_m \cap \mathcal{K}_{\{u'\}})\cup \mathcal{T}_n)\cap\overline{\mathcal{S}}|+|\mathcal{U}^{(m,n)}\setminus \{u'\}|\Big)L =e^*L,\label{totalrand_eq1}\\
\Rightarrow &&R_{Z_{\Sigma}} = \frac{b_{Z_{\Sigma}}}{L} \geq \frac{H(Z_{\Sigma})}{L} \geq e^*,
\end{eqnarray}
where (\ref{totalkey_pf_t1}) follows from the facts that $(\mathcal{S}_m \cap \mathcal{K}_{\{u'\}})\cap \mathcal{T}_n = \emptyset$ and $(\mathcal{S}_m \cap \mathcal{K}_{\{u'\}})\subseteq \overline{\mathcal{S}}$.

Second, a slight twist shows that $R_{Z_\Sigma} \ge K-1$ when $e^*=K$. 
Specifically, there exist sets $\mathcal{K}_{\{u\}}, \mathcal{S}_m, \mathcal{T}_n$ such that
$|((\mathcal{S}_m \cap \mathcal{K}_{\{u\}}) \cup \mathcal{T}_n) \cap \overline{\mathcal{S}}|
+ |\mathcal{U}^{(m,n)} \setminus \{u\}| = K$, 
and
$|(\mathcal{S}_{m} \cap \mathcal{K}_{\{u\}}) \cup \mathcal{K}_{\mathcal{U}^{(m,n)}} \cup \mathcal{T}_{n}| = K$.
Since the set systems $\bm{\mathcal{S}}$ and $\bm{\mathcal{T}}$ are monotone, there then exist 
$\mathcal{K}_{\{u'\}}, \mathcal{S}_{m'}, \mathcal{T}_{n'}$ such that
$|(\mathcal{S}_{m'} \cap \mathcal{K}_{\{u'\}}) \cup \mathcal{K}_{\mathcal{U}^{(m',n')}} \cup \mathcal{T}_{n'}| = K-1$ 
and 
$|((\mathcal{S}_{m'} \cap \mathcal{K}_{\{u'\}}) \cup \mathcal{T}_{n'}) \cap \overline{\mathcal{S}}| + |\mathcal{U}^{(m',n')} \setminus \{u'\}| = K-1$.
Applying Lemmas~\ref{lemma4ZV>=VL} and \ref{lemma5}, we then obtain the desired lower bound.
\begin{eqnarray}
H(Z_{\Sigma})&\overset{(\ref{sourcekey})}{\geq}&H\left(\{Z_{u,v}\}_{(u,v)\in \mathcal{K}}\right)\\
    &=&H\left(\{Z_{u,v}\}_{(u,v)\in \mathcal{T}_n}\right)+H\left(\{Z_{u,v}\}_{(u,v)\in \mathcal{K}}|\{Z_{u,v}\}_{(u,v)\in \mathcal{T}_{n'}}\right)\\
    &\geq&H\left(\{Z_{u,v}\}_{(u,v) \in \mathcal{T}_{n'}\cap\overline{\mathcal{S}}}\big|\{Z_{u,v}\}_{(u,v)\in\mathcal{T}_{n'}\setminus \overline{\mathcal{S}}}\right)\notag\\
    &&+H\big(\{Z_{u,v}\}_{(u,v) \in (\mathcal{S}_{m'} \cap \mathcal{K}_{\{u'\}}) \cup \mathcal{K}_{\mathcal{U}^{(m',n')}}} \big| \{Z_{u,v}\}_{(u,v)\in \mathcal{T}_{n'}}\big)~~\\
    &\overset{(\ref{lemma4zv})(\ref{lemma5_eqY})}{\geq}& \big|\mathcal{T}_{n'}\cap\overline{\mathcal{S}}\big| L+ |\mathcal{S}_{m'} \cap \mathcal{K}_{\{u'\}}|L+|\mathcal{U}^{(m',n')}\setminus \{u'\}|L\\
    &=&\Big( | ((\mathcal{S}_{m'} \cap \mathcal{K}_{\{u'\}})\cup \mathcal{T}_{n'})\cap\overline{\mathcal{S}}|+|\mathcal{U}^{(m',n')}\setminus \{u'\}|\Big)L  = (K-1)L,\\
\Rightarrow ~~~ R_{Z_{\Sigma}} &=& \frac{b_{Z_{\Sigma}}}{L} \geq \frac{H(Z_{\Sigma})}{L} \geq K-1.
\end{eqnarray}

The proof of the converse is now complete. Note that the established bounds on the communication rates $R_X$, $R_Y$, and the source key rate $R_{Z_\Sigma}$ also apply to the setting of Theorem~\ref{thm2}, so it remains to construct an achievable scheme to complete the proof.

\subsection{General Achievability Proof of Theorem \ref{thm1}}
\label{subsec:achievability}

Before presenting the achievable schemes, we outline the structure of this section. 
Each subsubsection constructs a coding scheme corresponding to a 
different condition on the source key rate $R_{Z_\Sigma}$. 
These schemes together establish the general achievability of Theorem~\ref{thm1}, 
ensuring both correctness and security under the respective conditions.

\subsubsection{Achievable Scheme for the Case $R_{Z_\Sigma}=K-1$ When $e^*=K$}
\label{sec:ach1}

Each input $(u,v)$ is assigned a key variable $Z_{u,v}$. 
Let $\{N_{u,v}\}_{(u,v)\in \mathcal{K}\setminus\{(U,V)\}}$ be $K-1$ independent and identically distributed (i.i.d.) uniform random variables over $\mathbb{F}_q$. 
We set
\begin{eqnarray}
Z_{u,v} &=& N_{u,v}, \forall (u,v)\in\mathcal{K}\setminus\{(U,V)\} \notag\\
Z_{U,V} &=& -\sum_{(u,v)\in\mathcal{K}\setminus\{(U,V)\}} N_{u,v}.
\label{eq:t1t1}
\end{eqnarray}
We have
\begin{eqnarray}
\sum_{(u,v)\in\mathcal{K}} Z_{u,v} =0 .
\label{eq:t1t2}
\end{eqnarray}
The transmitted messages are 
\begin{eqnarray}
X_{u,v} &=& W_{u,v}+Z_{u,v}, \quad (u,v)\in\mathcal{K}, 
\label{eq:t1t3}\\
Y_u &=& \sum_{v\in[V_u]} X_{u,v}, \quad u\in[U].
\label{eq:t1t4}
\end{eqnarray}
Therefore, the achieved communication rates are 
$R_X=L_X/L=1$ and $R_Y=L_Y/L=1$, 
while the source key rate is 
$R_{Z_\Sigma}=b_{Z_\Sigma}/L=K-1$.
Correctness follows since for any server $k\in[U]$,
\begin{eqnarray}
    \sum_{v\in[V_k]} X_{k,v}
+\sum_{u\in[U]\setminus\{k\}} Y_u &\overset{(\ref{eq:t1t4})}{=}&\sum_{v\in[V_k]} X_{k,v}
+\sum_{u\in[U]\setminus\{k\}} \sum_{v\in[V_u]} X_{u,v} \notag\\
&=&\sum_{(u,v)\in\mathcal{K}} X_{u,v} \notag\\
&\overset{(\ref{eq:t1t3})(\ref{eq:t1t2})}{=}&
\sum_{(u,v)\in\mathcal{K}} W_{u,v}.
\end{eqnarray}
The proof of security is deferred to Section \ref{pfsecurity}.

\subsubsection{Achievable Scheme of $R_{Z_\Sigma} = e^*$ for the `Otherwise' Case where $e^* \leq K-1$ and $a^* \leq |\overline{\mathcal{S}}|-1$}
\label{sec:ach2}

Each input $(u,v)$ in the (implicit and explicit) security input set $\overline{\mathcal{S}}$ is assigned a key variable. 
We operate over the finite field $\mathbb{F}_q$. To guarantee the existence of coefficient vectors satisfying the required linear independence conditions, we choose the field size such that $q > e^* \binom{|\overline{S}|}{e^*}$. 

\begin{remark}
The existence of such coefficient vectors follows from the Schwartz--Zippel lemma: the determinant of any $e^* \times e^*$ submatrix formed by the selected vectors is a non-zero polynomial of total degree at most $e^*$ in the entries of $\mathbf{h}_{u,v}$. With $q$ exceeding the maximum number of roots of this polynomial, a non-vanishing assignment exists. A sufficient condition is $q > e^* \binom{|\overline{S}|}{e^*}$.
\end{remark}

Let $N_1,\ldots,N_{e^*}$ be $e^*$ i.i.d. uniform random variables over $\mathbb{F}_{q}$, and define the column vector
$Z_\Sigma = {\bf N} = (N_1,\ldots,N_{e^*})^\top \in \mathbb{F}_{q}^{e^* \times 1}$.
The  \indiv keys are given by
\begin{eqnarray}
Z_{u,v} &=& {\bf h}_{u,v}\,{\bf N}, \quad (u,v)\in \overline{\mathcal{S}} \notag\\
Z_{u,v} &=& 0, \quad (u,v)\in \mathcal{K}\setminus \overline{\mathcal{S}},
\label{eq:c111}
\end{eqnarray}
where ${\bf h}_{u,v} \in \mathbb{F}_{q}^{1\times e^*}$ are coefficient vectors constructed as follows. 
Suppose
$\overline{\mathcal{S}} = \{(u_1,v_1),\ldots,(u_{|\overline{\mathcal{S}}|},v_{|\overline{\mathcal{S}}|})\}$.
We choose
${\bf h}_{u_1,v_1},\ldots,{\bf h}_{u_{|\overline{\mathcal{S}}|-1},v_{|\overline{\mathcal{S}}|-1}}$
independently and uniformly from $\mathbb{F}_{q}^{1\times e^*}$, and define
\begin{eqnarray}
{\bf h}_{u_{|\overline{\mathcal{S}}|},v_{|\overline{\mathcal{S}}|}}
=
-\left({\bf h}_{u_1,v_1}+\cdots+{\bf h}_{u_{|\overline{\mathcal{S}}|-1},v_{|\overline{\mathcal{S}}|-1}}\right).
\label{eq:c121}
\end{eqnarray}
There must exist a realization\footnote{
Since $e^* < |\overline{\mathcal{S}}|$, the determinant associated with any selection of $e^*$ distinct vectors ${\bf h}_{u,v}$ defines a non-zero polynomial. Applying the Schwartz--Zippel lemma, the product of all such determinant polynomials has degree $e^* \binom{|\overline{\mathcal{S}}|}{e^*}$. By choosing a field $\mathbb{F}_q$ with $q$ larger than this degree, the probability that the product polynomial evaluates to a non-zero value is strictly positive, which ensures the validity of (\ref{eq:sz111}).
}
of $\{{\bf h}_{u,v}:(u,v)\in\overline{\mathcal{S}}\}$ such that
\begin{eqnarray}
\mbox{\tit{Any set of at most $e^*$ vectors chosen from} } 
\{{\bf h}_{u,v}:(u,v)\in\overline{\mathcal{S}}\} \notag\\
\mbox{\tit{and}} 
\left\{\sum_{v\in[V_u]}{\bf h}_{u,v}:u\in\mathcal{U}^{(m,n)}\right\}
\mbox{\tit{is linearly independent.}}
\label{eq:sz111}
\end{eqnarray}
Moreover,
\begin{eqnarray}
\sum_{(u,v)\in\mathcal{K}} Z_{u,v}
\overset{(\ref{eq:c111})}{=}
\sum_{(u,v)\in\overline{\mathcal{S}}} Z_{u,v}
\overset{(\ref{eq:c121})}{=} 0 .
\label{eq:t2t1}
\end{eqnarray}
The transmitted messages are defined as
\begin{eqnarray}
X_{u,v} &=& W_{u,v} + Z_{u,v}, \quad (u,v)\in\mathcal{K} \label{eq:achm1X1}\\
Y_u &=& \sum_{v\in[V_u]} X_{u,v}, \quad u\in[U]. 
\label{eq:achm1Y11}
\end{eqnarray}
Thus, the achieved communication rates are
$R_X=L_X/L=1$ and $R_Y=L_Y/L=1$, while the source key rate is
$R_{Z_\Sigma}=b_{Z_\Sigma}/L=e^*$.
Correctness follows since for any server $k\in[U]$
\begin{eqnarray}
\sum_{v\in[V_k]}X_{k,v}+\sum_{u\in[U]\setminus\{k\}}Y_u &\overset{(\ref{eq:achm1Y11})}{=}&\sum_{v\in[V_k]}X_{k,v}+\sum_{u\in[U]\setminus\{k\}}\sum_{v\in[V_u]}X_{u,v} \notag\\
&=&\sum_{(u,v)\in\mathcal{K}}X_{u,v} \notag\\
&\overset{(\ref{eq:achm1X1})(\ref{eq:t2t1})}{=}
&\sum_{(u,v)\in\mathcal{K}}W_{u,v}.
\end{eqnarray}
The security proof is deferred to Section~\ref{pfsecurity}, which relies on proving the existence of a realization of the randomly generated vectors ${\bf h}_{u,v}$ satisfying the required full-rank property.

\subsubsection{Achievable Scheme of $R_{Z_\Sigma} = e^*$ for the `Otherwise' Case where $e^* \leq K-1$, $a^* = |\overline{\mathcal{S}}|$, and $|\mathcal{Q}| \leq K-1$}
\label{sec:ach3}

In this case, each input in the total security input set $\overline{\mathcal{S}}$, together with one additional input outside $\mathcal{Q}$, is assigned a key variable. 
We operate over the finite field $\mathbb{F}_q$ with field size chosen such that
$q > e^* \binom{|\overline{\mathcal{S}}|+1}{e^*}$.

Let $N_1,\ldots,N_{e^*}$ be $e^*$ independent and identically distributed (i.i.d.) uniform random variables over $\mathbb{F}_q$, and define
$Z_\Sigma = {\bf N} = (N_1,\ldots,N_{e^*})^\top \in \mathbb{F}_q^{e^* \times 1}$.
Since $|\mathcal{Q}| < K$, there exists $(u',v') \in \mathcal{K}\setminus\mathcal{Q}$. 
Moreover, $(u',v') \notin \overline{\mathcal{S}}$ because $e^* = |\overline{\mathcal{S}}|$ and $\overline{\mathcal{S}} \subset \mathcal{Q}$.
The \indiv keys are given by
\begin{eqnarray}
Z_{u,v} &=& {\bf h}_{u,v}{\bf N}, \quad (u,v)\in (\overline{\mathcal{S}}\cup\{(u',v')\}) \notag\\
Z_{u,v} &=& 0, \quad (u,v)\in \mathcal{K}\setminus(\overline{\mathcal{S}}\cup\{(u',v')\})
\label{eq:c21}
\end{eqnarray}
where ${\bf h}_{u,v}\in\mathbb{F}_q^{1\times e^*}$ are chosen as follows. 
Suppose
$\overline{\mathcal{S}}=\{(u_1,v_1),\ldots,(u_{|\overline{\mathcal{S}}|},v_{|\overline{\mathcal{S}}|})\}$.
We choose
${\bf h}_{u_1,v_1},\ldots,{\bf h}_{u_{|\overline{\mathcal{S}}|},v_{|\overline{\mathcal{S}}|}}$
independently and uniformly from $\mathbb{F}_q^{1\times e^*}$ and define
\begin{eqnarray}
{\bf h}_{u',v'} =
-\left({\bf h}_{u_1,v_1}+\cdots+{\bf h}_{u_{|\overline{\mathcal{S}}|},v_{|\overline{\mathcal{S}}|}}\right).
\label{eq:c22}
\end{eqnarray}
Using the same argument based on Schwartz-Zippel lemma \cite{Demillo_Lipton,Schwartz,Zippel}, there exists a realization of $\{{\bf h}_{u,v}\}$ such that
\begin{eqnarray}
\mbox{\tit{Any set of at most $e^*$ vectors chosen from}}
\{{\bf h}_{u,v}:(u,v)\in\overline{\mathcal{S}}\cup\{(u',v')\}\} \notag\\
\mbox{\tit{and}}
\left\{\sum_{v\in[V_u]}{\bf h}_{u,v}:u\in\mathcal{U}^{(m,n)}\right\}
\mbox{\tit{are linearly independent.}}
\label{eq:sz2}
\end{eqnarray}
Moreover,
\begin{eqnarray}
\sum_{(u,v)\in\mathcal{K}} Z_{u,v}
\overset{(\ref{eq:c21})}{=}
\sum_{(u,v)\in(\overline{\mathcal{S}}\cup\{(u',v')\})} Z_{u,v}
\overset{(\ref{eq:c22})}{=}0.
\label{eq:c23}
\end{eqnarray}
The transmitted messages are 
\begin{eqnarray}
X_{u,v} &=& W_{u,v}+Z_{u,v}, \quad (u,v)\in\mathcal{K}
\label{eq:achm2X}\\
Y_u &=& \sum_{v\in[V_u]}X_{u,v}, \quad u\in[U].
\label{eq:achm2Y}
\end{eqnarray}
The achieved communication rates satisfy
$R_X=L_X/L=1$ and $R_Y=L_Y/L=1$.
Since each symbol is from $\mathbb{F}_q$, we have $b_{Z_\Sigma}=e^*L$, which gives the key rate
$R_{Z_\Sigma}=e^*$.
To verify correctness, for any server $k\in[U]$ we have
\begin{eqnarray}
    \sum_{v\in[V_k]}X_{k,v}+\sum_{u\in[U]\setminus\{k\}}Y_u 
&\overset{(\ref{eq:achm2Y})}{=}&\sum_{v\in[V_k]}X_{k,v}
+\sum_{u\in[U]\setminus\{k\}}\sum_{v\in[V_u]}X_{u,v} \notag\\
&=&\sum_{(u,v)\in\mathcal{K}}X_{u,v} \notag\\
&\overset{(\ref{eq:achm2X})(\ref{eq:c23})}{=}
&\sum_{(u,v)\in\mathcal{K}}W_{u,v}.
\end{eqnarray}
The security proof is deferred to Section~\ref{pfsecurity}, where we show that the randomly generated vectors ${\bf h}_{u,v}$ satisfy the required full-rank property.

\section{Proof of Theorem \ref{thm2}} \label{pfthm2}

This section focuses on constructing achievable schemes for Theorem~\ref{thm2}. 
Since the converse bounds for the communication and source key rates have already been established in Section~\ref{convpfthm1}, 
we only need to demonstrate achievable schemes that attain the upper bound for the source key rate. 
We start with a concrete example to illustrate the main ideas before generalizing.

\subsection{Upper bound Proof of Example \ref{ex2}}
\begin{example}
\label{ex2}
Consider a system with $U=3$ servers and $K=8$ users, where $V_1=4$, $V_2=3$, and $V_3=1$. 
The security input sets are given by
$(\mathcal{S}_1,$ $\cdots,$ $\mathcal{S}_{8})=(\emptyset,\{(1,1)\},\{(2,1)\},\{(2,2)\},\{(1,1),(2,1)\},$ $\{(1,1),(2,2)\},\{(2,1),$ $(2,2)\},\{(1,1),(2,1),(2,2)\})$ and collusion sets
$(\mathcal{T}_1, \cdots, \mathcal{T}_{28}) =(~\emptyset, \{(1,2)\},\{(1,3)\},\{(1,4)\},\{(2,1)\},\{(2,3)\},$ $\{(3,1)\},\{(1,2),(1,3)\},\{(1,2),$ $(1,4)\},~\{(1,2),(2,3)\},~\{(1,2),(3,1)\},~\{(1,3),(1,4)\},~\{(1,3),(2,1)\},~\{(1,3),$ $(2,3)\},\{(1,3),(3,1)\},\{(1,4),(2,1)\},\{(1,4),(2,3)\},\{(1,4),(3,1)\},\{(2,1),(3,1)\},\{(2,3),(3,1)\},~\{(1,2),(1,3),$ $(2,3)\},~\{(1,2),(1,4),(2,3)\},~\{(1,2),(2,3),(3,1)\},~\{(1,3),(1,4),(2,1)\},~\{(1,3),(1,4),(3,1)\},\{(1,3),(2,1),$ $(3,1)\},\{(1,4),(2,1),(3,1)\},\{(1,3),(1,4),(2,1),(3,1)\}).$ 
In this example, we have $a^* = 3$, $e^* = 2$, and $|\overline{\mathcal{S}}| = 3$. 
Furthermore, the set $\mathcal{Q}$ can be computed as
$|\mathcal{Q}|=|\mathcal{Q}_1|=|\cup_{u,m,n:|\mathcal{A}_{(u,m,n)}|=|\overline{\mathcal{S}}|}(\mathcal{S}_m \cap \mathcal{K}_{\{u\}})
\cup \mathcal{K}_{\mathcal{U}^{(m,n)}}
\cup \mathcal{T}_n|=|((\mathcal{S}_8 \cap \mathcal{K}_{\{1\}})
\cup \mathcal{K}_{\mathcal{U}^{(8,21)}}
\cup \mathcal{T}_{21})\cup((\mathcal{S}_8 \cap \mathcal{K}_{\{1\}})
\cup \mathcal{K}_{\mathcal{U}^{(8,22)}}
\cup \mathcal{T}_{22})\cup((\mathcal{S}_8 \cap \mathcal{K}_{\{1\}})
\cup \mathcal{K}_{\mathcal{U}^{(8,23)}}
\cup \mathcal{T}_{23})|=|\{(1,1),(1,2),(1,3),(1,4),(2,1),(2,2),(2,3),(3,1)\}|=K=8$.
Hence, the conditions specified in Theorem~\ref{thm2} are satisfied, since
$e^* = 2 \leq K-1 = 7$, 
$a^* = |\overline{\mathcal{S}}| = 3$, and 
$|\mathcal{Q}| = K = 8$.

For the lower bound, the inequality $R_{Z_{\Sigma}} \geq e^*$ has already been established in Theorem~\ref{thm1}. 
It remains to show the achievability of the region $R_{Z_{\Sigma}} \geq e^* + b^*$ by constructing a scheme with rate $e^* + b^*$ that satisfies both the correctness constraint~(\ref{correctness}) and the security constraint~(\ref{security}).



Before specifying the scheme, we first characterize the required key sizes based on the constraints induced by the linear program. 
For any $(u,v)\in \overline{\mathcal{S}}$, the corresponding key must satisfy $H(Z_{u,v}) \ge L$. 
For $(u,v)\in \mathcal{K}\setminus \overline{\mathcal{S}}$, the key sizes are determined by the constraints in Lemma~\ref{lemma2}.

Consider all sets $\mathcal{A}_{u,m,n}$. When $|\mathcal{A}_{u,m,n}| \le |\overline{\mathcal{S}}|-1$, the constraints in Lemma~\ref{lemma2} are automatically satisfied for the keys in $\mathcal{K}\setminus \overline{\mathcal{S}}$ due to the cardinality gap. 
Hence, it suffices to focus on the cases where $|\mathcal{A}_{u,m,n}| = |\overline{\mathcal{S}}|$.
In Example~\ref{ex2}, there are exactly three such sets, namely $\mathcal{A}_{1,8,21}$, $\mathcal{A}_{1,8,22}$, and $\mathcal{A}_{1,8,23}$. 
Applying Lemma~\ref{lemma2} to these sets yields the only nontrivial constraints:
\begin{eqnarray}
    \mathcal{A}_{1,8,21}: && H(Z_{1,4},Z_{3,1}\mid Z_{1,2},Z_{1,3},Z_{2,3})\geq L, \notag \\
    \mathcal{A}_{1,8,22}: && H(Z_{1,3},Z_{3,1}\mid Z_{1,2},Z_{1,4},Z_{2,3})\geq L, \notag \\
    \mathcal{A}_{1,8,23}: && H(Z_{1,3},Z_{1,4}\mid Z_{1,2},Z_{2,3},Z_{3,1})\geq L.  \label{eq:tt3}
\end{eqnarray}
These constraints completely characterize the feasible region of key assignments for this example.

Next, we reformulate the constraints in (\ref{eq:tt3}) as a linear program with variables $b_{u,v}$. 
To this end, we express the joint entropy in terms of incremental entropies under a fixed (lexicographic) ordering as follows:
\begin{align}
    &H(Z_{1,2}) = b_{1,2}L,\;
H(Z_{1,3}\mid Z_{1,2}) = b_{1,3}L,\;
H(Z_{1,4}\mid Z_{1,3},Z_{1,2}) = b_{1,4}L,\notag\\
&H(Z_{2,3}\mid Z_{1,2},Z_{1,3},Z_{1,4}) = b_{2,3}L,\;
H(Z_{3,1}\mid Z_{1,2},Z_{1,3},Z_{1,4},Z_{2,3}) = b_{3,1}L.
\end{align}
By normalizing the total entropy $H(Z_{1,2},Z_{1,3},Z_{1,4},Z_{2,3},Z_{3,1})$ by $L$ and applying the chain rule to (\ref{eq:tt3}), we obtain the following linear program:
 \begin{eqnarray}
&& \min~~  b_{1,2}+b_{1,3}+b_{1,4}+b_{2,3}+b_{3,1}   
\\
 \text{subject to} 
     &&  b_{1,4}+b_{3,1}\geq (H(Z_{1,4}|Z_{1,2},Z_{1,3},Z_{2,3})+H(Z_{3,1}|Z_{1,2},Z_{1,3},Z_{1,4},Z_{2,3}))/L\geq 1, \notag \\
      &&  b_{1,3}+b_{3,1}\geq (H(Z_{1,3}|Z_{1,2},Z_{1,4},Z_{2,3})+H(Z_{3,1}|Z_{1,2},Z_{1,3},Z_{1,4},Z_{2,3}))/L\geq 1, \notag \\
     &&  b_{1,3}+b_{1,4}\geq (H(Z_{1,3}|Z_{1,2},Z_{2,3},Z_{3,1})+H(Z_{1,4}|Z_{1,2},Z_{1,3},Z_{2,3},Z_{3,1}))/L\geq 1, \notag \\
    && b_{1,2}\geq 0, b_{1,3}\geq 0, b_{1,4}\geq 0,b_{2,3}\geq 0, b_{3,1}\geq 0.
    \label{eq:upplp222}
 \end{eqnarray}
It can be verified that one optimal solution is given by
$b^*_{1,2}=b^*_{2,3}=0$ and 
$b^*_{1,3}=b^*_{1,4}=b^*_{3,1}=1/2$, 
which yields $b^* = 3/2$.
Therefore, the source key rate $R_{Z_\Sigma} = e^* + b^* = 2 + 3/2 = 7/2$ is achievable.
The optimal LP solution specifies the required entropy allocation for each key variable in $\mathcal{K}_{[U]}\setminus\overline{\mathcal{S}}$, which guides the construction of an achievable scheme satisfying (\ref{correctness}) and (\ref{security}).

Let $q=5$ and operate over $\mathbb{F}_5$. 
We generate $7$ i.i.d. uniform random variables $N_1,\ldots,N_7$ and let $Z_\Sigma = (N_1,\ldots,N_7)$, so that the total key size is $L_\Sigma=7$. 
Set $L=2$, i.e., each message consists of two symbols $W_{u,v}=(W^{(1)}_{u,v},W^{(2)}_{u,v})$ for $(u,v)\in\mathcal{K}$.
Based on the optimal solution of (\ref{eq:upplp222}), the key entropies are given by
$H(Z_{1,1})=2$, $H(Z_{1,2})=0$, $H(Z_{1,3})=1$, $H(Z_{1,4})=1$, 
$H(Z_{2,1})=2$, $H(Z_{2,2})=2$, $H(Z_{2,3})=0$, and $H(Z_{3,1})=1$.
We now construct the individual keys as linear combinations of $(N_1,\ldots,N_7)$:
\begin{align}
 &Z_{1,1} = (- (N_1 + N_2 + N_3 + N_5 + N_7),- (N_1 + 2N_2 + N_4 + N_6 + 3N_7 )), \notag\\
 &Z_{1,2} = 0,\quad Z_{1,3} = N_1,\quad Z_{1,4} = N_2, \notag\\
 &Z_{2,1} = (N_3,N_4),\quad Z_{2,2} = (N_5,N_6),\quad Z_{2,3} = 0,\quad Z_{3,1} = N_7.\label{eq:ex3key}
\end{align}

Each User $(u,v)$ sends a message $X_{u,v}=W_{u,v}+Z_{u,v}$ to Server $u$ \af:
\begin{align}
& X_{1,1} = \left[
\begin{array}{c}
W^{(1)}_{1,1} - (N_1 + N_2 + N_3 + N_5 + N_7)\\
W^{(2)}_{1,1} - (N_1 + 2N_2 + N_4 + N_6 + 3N_7 )
\end{array}
\right],
X_{1,2} = \left[
\begin{array}{c}
W^{(1)}_{1,2}  \\
W^{(2)}_{1,2}  
\end{array}
\right],~
X_{1,3} = \left[
\begin{array}{c}
W^{(1)}_{1,3} + N_1 \\
W^{(2)}_{1,3} + N_1 
\end{array}
\right],\notag\\
& 
X_{1,4} = \left[
\begin{array}{c}
W^{(1)}_{1,4}+ N_2  \\
W^{(2)}_{1,4}+ 2N_2 
\end{array}
\right], ~
X_{2,1} = \left[
\begin{array}{c}
W^{(1)}_{2,1} + N_3 \\
W^{(2)}_{2,1} + N_4 
\end{array}
\right],~
X_{2,2} = \left[
\begin{array}{c}
W^{(1)}_{2,2} + N_5 \\
W^{(2)}_{2,2} + N_6 
\end{array}
\right],\notag\\
& 
X_{2,3} = \left[
\begin{array}{c}
W^{(1)}_{2,3}  \\
W^{(2)}_{2,3}  
\end{array}
\right],~
X_{3,1} = \left[
\begin{array}{c}
W^{(1)}_{3,1} + N_7  \\
W^{(2)}_{3,1} + 3N_7
\end{array}
\right].\label{eq:ex3messageX}
\end{align}
Upon receiving the messages from its associated users, Server $u$ computes 
$Y_u=\sum_{v\in [V_u]}X_{u,v}$ and sends it to the other servers. We have
\begin{align}
& Y_1=X_{1,1}+X_{1,2}+X_{1,3}+X_{1,4} = \left[
\begin{array}{c}
W^{(1)}_{1,1}+W^{(1)}_{1,2} +W^{(1)}_{1,3} +W^{(1)}_{1,4} -N_3 - N_5 - N_7 \\
W^{(2)}_{1,1}+ W^{(2)}_{1,2}+ W^{(2)}_{1,3}+ W^{(2)}_{1,4} - N_4 - N_6 - 3N_7 
\end{array}
\right],\notag
\\
& 
Y_2=X_{2,1}+X_{2,2}+X_{2,3} = \left[
\begin{array}{c}
W^{(1)}_{2,1}+W^{(1)}_{2,2}+W^{(1)}_{2,3} + N_3 +N_5\\
W^{(2)}_{2,1}+ W^{(2)}_{2,2}+W^{(2)}_{2,3} + N_4 + N_6 
\end{array}
\right],\notag\\
& 
Y_3=X_{3,1} = \left[
\begin{array}{c}
W^{(1)}_{3,1}  + N_7 \\
W^{(2)}_{3,1}+ 3N_7 
\end{array}
\right].\label{eq:ex3messageY}
\end{align}

\textbf{Correctness:} 
For any server $k \in [3]$, the total sum of all inputs can be recovered as 
$\sum_{v \in [V_k]} X_{k,v} + \sum_{u \in [3] \setminus \{k\}} Y_u
= \sum_{(u,v)\in\mathcal{K}} W_{u,v}.$
This follows from the fact that the aggregate of all transmitted messages equals the sum of the inputs plus the sum of all key variables, and the key construction in (\ref{eq:ex3key}) ensures that the sum of all key variables evaluates to zero.

\textbf{Security:}
To verify security, we consider the specific case $u=1, m=8, n=28$; other cases follow similarly and are deferred to the general proof in Section~\ref{pfsecurity}. We compute the conditional mutual information: 
\begin{align}
&I\Big(W_{1,1},W_{2,1},W_{2,2} ; X_{1,1},X_{1,2},X_{1,3},X_{1,4},Y_2,Y_3 \Big| \sum_{(u,v)\in \mathcal{K}}W_{u,v},  W_{1,3}, Z_{1,3}, W_{1,4}, Z_{1,4},W_{2,1}, Z_{2,1},W_{3,1}, Z_{3,1} \Big) \notag\\
=& H\Big( X_{1,1},X_{1,2},X_{1,3},X_{1,4},Y_2,Y_3 \Big| \sum_{(u,v)\in \mathcal{K}}W_{u,v},  W_{1,3}, Z_{1,3}, W_{1,4}, Z_{1,4},W_{2,1}, Z_{2,1},W_{3,1}, Z_{3,1} \Big)- \notag\\
&H\Big(X_{1,1},X_{1,2},X_{1,3},X_{1,4},Y_2,Y_3 \Big| \sum_{(u,v)\in \mathcal{K}}W_{u,v},  W_{1,1},W_{2,1},W_{2,2} , W_{1,3}, Z_{1,3}, W_{1,4}, Z_{1,4},W_{2,1}, Z_{2,1},W_{3,1}, Z_{3,1} \Big) \notag\\
=& H\Big( X_{1,1},X_{1,2},Y_2 \Big| \sum_{(u,v)\in \mathcal{K}}W_{u,v},  W_{1,3}, Z_{1,3}, W_{1,4}, Z_{1,4},W_{2,1}, Z_{2,1},W_{3,1}, Z_{3,1} \Big) \notag\\
&-H\Big(X_{1,1},X_{1,2},Y_2 \Big| \sum_{(u,v)\in \mathcal{K}}W_{u,v},  W_{1,1},W_{2,1},W_{2,2} , W_{1,3}, Z_{1,3}, W_{1,4}, Z_{1,4},W_{2,1}, Z_{2,1},W_{3,1}, Z_{3,1} \Big)\label{pfex2t1}\\
=& H\Big(W^{(1)}_{1,1}+W^{(1)}_{1,2} - N_5, W^{(2)}_{1,1}+ W^{(2)}_{1,2}  - N_6, W^{(1)}_{1,2}, W^{(2)}_{1,2}, W^{(1)}_{2,2}+W^{(1)}_{2,3} +N_5, W^{(2)}_{2,2}+W^{(2)}_{2,3} + N_6\notag\\
&\Big| \sum_{(u,v)\in \mathcal{K}}W_{u,v},  W_{1,3}, N_{1}, W_{1,4}, N_{2},W_{2,1}, N_{3},N_4,W_{3,1}, N_{7} \Big) \notag\\
&-H\Big(W^{(1)}_{1,2} - N_5,  W^{(2)}_{1,2}  - N_6, W^{(1)}_{1,2}, W^{(2)}_{1,2}, W^{(1)}_{2,3} +N_5, W^{(2)}_{2,3} + N_6\notag\\
&\Big| \sum_{(u,v)\in \mathcal{K}}W_{u,v},W_{1,1},W_{2,1},W_{2,2} ,  W_{1,3}, N_{1}, W_{1,4}, N_{2},W_{2,1}, N_{3},N_4,W_{3,1}, N_{7} \Big)\label{pfex2t2}\\
=& H\Big(W^{(1)}_{1,1} - N_5, W^{(2)}_{1,1}  - N_6, W^{(1)}_{1,2}, W^{(2)}_{1,2}, W^{(1)}_{2,2}+W^{(1)}_{2,3} +N_5, W^{(2)}_{2,2}+W^{(2)}_{2,3} + N_6\notag\\
&\Big| W^{(1)}_{1,1} + W^{(1)}_{1,2}+ W^{(1)}_{2,2}+W^{(1)}_{2,3},W^{(2)}_{1,1} + W^{(2)}_{1,2}+ W^{(2)}_{2,2}+W^{(2)}_{2,3} \Big) \notag\\
&-H\Big( N_5,  N_6, W^{(1)}_{1,2}, W^{(2)}_{1,2}, W^{(1)}_{2,3} , W^{(2)}_{2,3} \Big|  W^{(1)}_{1,2}+ W^{(1)}_{2,3}, W^{(2)}_{1,2}+ W^{(2)}_{2,3} \Big)\label{pfex2t3}\\
=& H\Big(W^{(1)}_{1,1} - N_5, W^{(2)}_{1,1}  - N_6, W^{(1)}_{1,2}, W^{(2)}_{1,2}, W^{(1)}_{2,2}+W^{(1)}_{2,3} +N_5, W^{(2)}_{2,2}+W^{(2)}_{2,3} + N_6,W^{(1)}_{1,1} + W^{(1)}_{1,2}+ W^{(1)}_{2,2}+W^{(1)}_{2,3},\notag\\
& W^{(2)}_{1,1} + W^{(2)}_{1,2}+ W^{(2)}_{2,2}+W^{(2)}_{2,3} \Big) - H\Big(W^{(1)}_{1,1} + W^{(1)}_{1,2}+ W^{(1)}_{2,2}+W^{(1)}_{2,3},W^{(2)}_{1,1} + W^{(2)}_{1,2}- W^{(2)}_{2,2}+W^{(2)}_{2,3} \Big)- \notag\\
&H\Big( N_5,  N_6, W^{(1)}_{1,2}, W^{(2)}_{1,2}, W^{(1)}_{2,3} , W^{(2)}_{2,3} ,  W^{(1)}_{1,2}+ W^{(1)}_{2,3}, W^{(2)}_{1,2}+ W^{(2)}_{2,3} \Big)+H\Big(W^{(1)}_{1,2}+ W^{(1)}_{2,3}, W^{(2)}_{1,2}+ W^{(2)}_{2,3}\Big)\\
=&6-2-6+2=0.
\end{align}
(\ref{pfex2t1}) holds since $X_{1,3}, X_{1,4}, Y_3$ are deterministic functions of $(W_{1,3}, Z_{1,3}, W_{1,4}, Z_{1,4}, W_{3,1}, Z_{3,1})$, which are already included in the conditioning set, and hence can be removed without affecting the conditional entropy. In~(\ref{pfex2t2}), we substitute the encoding relations $X_{u,v} = W_{u,v} + Z_{u,v}$ and $Y_u = \sum_{v \in [V_u]} X_{u,v}$, and rewrite all signals in terms of message symbols and independent noise variables, while removing redundant conditioned variables. Equation~(\ref{pfex2t3}) follows by observing that the aggregate sums $W^{(1)}_{1,1} + W^{(1)}_{1,2} + W^{(1)}_{2,2} + W^{(1)}_{2,3}$ and $W^{(2)}_{1,1} + W^{(2)}_{1,2} + W^{(2)}_{2,2} + W^{(2)}_{2,3}$ are determined by the conditioning set, and similarly for $W^{(1)}_{1,2}+ W^{(1)}_{2,3}$ and $W^{(2)}_{1,2}+ W^{(2)}_{2,3}$ after revealing additional messages. Using these relations, the entropy terms can be expressed via invertible linear transformations of independent components. Finally, since all message symbols $\{W_{u,v}\}$ are mutually independent and independent of the keys $\{Z_{u,v}\}$, each independent component contributes one unit of entropy, and a direct dimension counting yields $6 - 2 - 6 + 2 = 0$, which establishes the desired result.

The achieved \comm rates are $R_X = L_X / L=2/2 = 1$ and $R_Y = L_Y / L=2/2 = 1$. The achieved key rate is $R_{Z_\Sigma} = b_{Z_\Sigma}/L = 7/2 =2+3/2= e^* + b^*$, as desired.
\end{example}

We now extend the insights from Example~\ref{ex2} to the general setting. 
The key assignment and message aggregation principles illustrated in the example can be applied systematically to any number of servers and users, leading to an achievable scheme that meets the correctness and security requirements for arbitrary parameters.

\subsection{General Proof of Theorem \ref{thm2}}\label{thm2subsec2}

The lower bound has already been established in Subsection~\ref{convpfthm1}. 
Here, we focus on proving the achievability of the rate region $R_X \ge 1, \; R_Y \ge 1, \; R_{Z_\Sigma} \ge e^* + b^*$,
under the assumptions $e^* \le K-1$, $a^* = |\overline{\mathcal{S}}|$, and $|\mathcal{Q}| = K$.
To compute the optimal $b^*$, we first allocate a key of size $L$ to each input in the total security input set $\overline{\mathcal{S}}$. 
For the remaining inputs in $\mathcal{K}\setminus \overline{\mathcal{S}}$, additional keys are assigned in such a way that the constraints in Lemma~\ref{lemma2} are satisfied. 
In particular, for any triple $(u',m,n)$ with $|\mathcal{A}_{u',m,n}| = |\overline{\mathcal{S}}|$, the key assignments must satisfy
\begin{align}
H\Big(
\{Z_{u,v}\}_{(u,v) \in \mathcal{K} \setminus ((\mathcal{S}_m \cap \mathcal{K}_{\{u'\}}) \cup \mathcal{K}_{\mathcal{U}^{(m,n)}} \cup \mathcal{T}_n)}
\,\Big|\,
\{Z_{u,v}\}_{(u,v)\in\mathcal{T}_n}
\Big)
\ge L.
\label{eq:upp1}
\end{align}
This condition guarantees that each critical set $\mathcal{A}_{u',m,n}$ of size $|\overline{\mathcal{S}}|$ has sufficient residual key entropy to meet the correctness and security requirements.

Next, we reformulate the constraints in (\ref{eq:upp1}) as a linear program over variables $b_{u,v}$, defined by
\begin{equation}
b_{u,v} \triangleq \frac{1}{L} H\Big(Z_{u,v} \,\big|\, (Z_{i,j})_{(i,j)\in\mathcal{K}\setminus\overline{\mathcal{S}},\,(i,j)<(u,v)}\Big),
\quad \forall (u,v)\in \mathcal{K}\setminus \overline{\mathcal{S}}.
\label{assume_eq1}
\end{equation}
For any $u', m, n$, all inputs in $\mathcal{K}\setminus\overline{\mathcal{S}}$ may potentially belong to the colluding set $\mathcal{T}_n$. 
By normalizing the total entropy $H(\{Z_{u,v}\}_{(u,v)\in\mathcal{K}\setminus\overline{\mathcal{S}}})$ by $L$ and applying the chain rule in lexicographic order, we arrive at the following linear program:
 \begin{eqnarray}
&& \min  \sum_{(u,v)\in \mathcal{K}\setminus\overline{\mathcal{S}}} b_{u,v} \label{eq:uppmin}
\\
 \text{subject to} 
    && \sum_{(u,v)\in \mathcal{K}\setminus ((\mathcal{S}_m \cap \mathcal{K}_{\{u'\}}) \cup \mathcal{K}_{\mathcal{U}^{(m,n)}} \cup \mathcal{T}_n)} b_{u,v} 
    \geq 1,~~~\forall u',m,n ~\mbox{such that}~ |\mathcal{A}_{u',m,n}| = |\overline{\mathcal{S}}| \label{eq:upplp20}
    \\
    && ~~~~~~~~~~~b_{u,v}\geq 0, ~~~~~~~~~~~~~~~~~~~~~\forall (u,v)\in \mathcal{K}\setminus\overline{\mathcal{S}}.
    \label{eq:upplp22}
 \end{eqnarray}
Returning to the key assignment for inputs in $\mathcal{K}\setminus\overline{\mathcal{S}}$, 
the size of the key allocated to each input is determined by the optimal solution of the linear program (\ref{eq:uppmin})--(\ref{eq:upplp22}). 
Let the optimal solution be
\begin{eqnarray}
b_{u,v} = b_{u,v}^* = \frac{p_{u,v}}{\overline{q}}, 
\quad \forall (u,v)\in \mathcal{K}\setminus\overline{\mathcal{S}},
\end{eqnarray}
where $p_{u,v}$ and $\overline{q}$ are non-negative integers. 
Then, the total key contribution from $\mathcal{K}\setminus\overline{\mathcal{S}}$ is
\begin{eqnarray}
b^* = \sum_{(u,v)\in \mathcal{K}\setminus\overline{\mathcal{S}}} b_{u,v}^*
= \frac{\overline{p}}{\overline{q}},
\label{eq:pq111}
\end{eqnarray}
where $\overline{p}=\sum_{(u,v)\in \mathcal{K}\setminus\overline{\mathcal{S}}} p_{u,v}$.



Choose a finite field $\mathbb{F}_q$ such that
$q > (e^*+ b^*)\overline{q} \binom{K \overline{q}}{(e^*+ b^*)\overline{q}}$. 
Consider $\overline{p}+e^*\overline{q}$ i.i.d. uniform variables
$Z_\Sigma={\bf N}=(N_1,\ldots,N_{\overline{p}+e^*\overline{q}})^\top \in \mathbb{F}_q^{(\overline{p}+e^*\overline{q})\times 1}$. To satisfy the correctness constraint, we set the keys as random linear combinations of ${\bf N}$, and then adjust a single matrix ${\bf H}_{i,j}$ so that the overall sum of keys is zero.
Set the key variables as
\begin{align}
Z_{u,v} &= {\bf F}_{u,v} {\bf G}_{u,v} {\bf N}, \quad (u,v) \in \mathcal{K}\setminus \overline{\mathcal{S}},\notag\\
Z_{u,v} &= {\bf H}_{u,v} {\bf N}, \quad (u,v) \in \overline{\mathcal{S}}, \label{schpft1}
\end{align}
where the entries of ${\bf F}_{u,v}\in \mathbb{F}_q^{\overline{q}\times p_{u,v}}$, 
${\bf G}_{u,v}\in \mathbb{F}_q^{p_{u,v}\times (\overline{p}+e^*\overline{q})}$, 
${\bf H}_{u,v}\in \mathbb{F}_q^{\overline{q}\times (\overline{p}+e^*\overline{q})}$ 
are drawn uniformly and independently.\footnote{The existence of matrices $\mathbf{F}_{u,v}, \mathbf{G}_{u,v}, \mathbf{H}_{u,v}$ satisfying the required linear independence properties follows from the Schwartz--Zippel lemma, provided $q$ is larger than the total degree of the determinant polynomial of any relevant square submatrix. The sufficient condition above ensures such an instantiation exists.}
For a chosen $(i,j)\in \overline{\mathcal{S}}$, let
\begin{eqnarray}
{\bf H}_{i,j} = - \Bigg( \sum_{(u,v)\in \mathcal{K}\setminus \overline{\mathcal{S}}} {\bf F}_{u,v} {\bf G}_{u,v} 
+ \sum_{(u,v)\in \overline{\mathcal{S}}\setminus \{(i,j)\}} {\bf H}_{u,v} \Bigg). \label{eq:uppttt22}
\end{eqnarray}
Consequently,
\begin{eqnarray}
\sum_{(u,v)\in \mathcal{K}\setminus \overline{\mathcal{S}}} {\bf F}_{u,v} {\bf G}_{u,v} + \sum_{(u,v)\in \overline{\mathcal{S}}} {\bf H}_{u,v} = {\bf 0} \quad \Rightarrow \quad \sum_{(u,v)\in\mathcal{K}} Z_{u,v} = 0. \label{eq:achthm3t33311}
\end{eqnarray}

By a similar argument using the Schwartz-Zippel lemma, we can guarantee the existence of an instantiation of the matrices 
${\bf H}_{u,v}, {\bf F}_{u,v}, {\bf G}_{u,v}$ such that the following linear independence property holds (here, ${\bf F}_{u,v}^1$ denotes the first $p_{u,v}$ rows of ${\bf F}_{u,v}$):
\begin{align}
    &\mbox{\tit{The rows of}} {\bf F}_{u,v}^1 {\bf G}_{u,v}, {\bf H}_{u,v}, 
\mbox{ \tit{and}  $\sum_{v\in[V_u]} ({\bf F}_{u,v}^1 {\bf G}_{u,v} + {\bf H}_{u,v}),~ u \in \mathcal{U}^{(m,n)},$} ~\mbox{\tit{are linearly independent}} \notag\\
&\mbox{\tit{whenever their total number of rows does not exceed} } (e^* + b^* )\overline{q}, \label{eq:sz32} 
\end{align}
which ensures that the key contributions cannot be linearly combined to leak additional information to colluding servers.
Finally, setting $L=\overline{q}$, we represent each message as a vector
${W}_{u,v}=(W^{(1)}_{u,v},\ldots,W^{(\overline{q})}_{u,v})^\top \in \mathbb{F}_{q}^{\overline{q} \times 1}$, 
and encode the transmitted messages using the assigned keys:
\begin{eqnarray}
     {X}_{u,v} &=& {W}_{u,v} + {Z}_{u,v}, \quad \forall (u,v) \in \mathcal{K}_{[U]},\label{eq:uppm3X}\\
     {Y}_{u}&=&\sum_{v\in [V_u]}X_{u,v}, \quad \forall u \in [U].\label{eq:uppm3Y}
\end{eqnarray}



\textbf{Correctness:} 
For any server $k\in[U]$, correctness follows from
\begin{eqnarray}
    \sum_{v\in[V_k]} X_{k,v} + \sum_{u\in[U]\setminus\{k\}} Y_u &\overset{(\ref{eq:uppm3Y})}{=}& \sum_{v\in[V_k]} X_{k,v} + \sum_{u\in[U]\setminus\{k\}} \sum_{v\in[V_u]} X_{u,v} \\
&=& \sum_{(u,v)\in\mathcal{K}} X_{u,v} \notag\\
&\overset{(\ref{eq:uppm3X})(\ref{eq:achthm3t33311})}{=}
& \sum_{(u,v)\in\mathcal{K}} W_{u,v}.
\end{eqnarray}

The achieved communication rates are
$R_X = L_X/L = \overline{q}/\overline{q} = 1$ and
$R_Y = L_Y/L = \overline{q}/\overline{q} = 1$.
The achieved key rate is
$R_{Z_\Sigma}
= b_{Z_\Sigma}/L 
= (\overline{p}+e^*\overline{q})/\overline{q} 
\overset{(\ref{eq:pq111})}{=} 
e^* + \sum_{(u,v)\in\mathcal{K}\setminus\overline{\mathcal{S}}} b_{u,v}^* 
= e^* + b^* .$
The security proof is presented in a unified manner in Section~\ref{pfsecurity}.


\section{Proof of Security} \label{pfsecurity}
In this section, we show that the schemes in Subsections~\ref{sec:ach1}, \ref{sec:ach2}, \ref{sec:ach3}, and Section~\ref{thm2subsec2} satisfy the security constraint.

We first consider the scheme in Subsection~\ref{sec:ach1} for the case $e^*=K$. 
For any $k,m,n$ satisfying
$|((\mathcal{S}_m \cap \mathcal{K}_{\{k\}})\cup \mathcal{T}_n)\cap\overline{\mathcal{S}}|
+|\mathcal{U}^{(m,n)}\setminus\{k\}| = K$ and 
$|((\mathcal{S}_m \cap \mathcal{K}_{\{k\}}) 
\cup \mathcal{K}_{\mathcal{U}^{(m,n)}} 
\cup \mathcal{T}_n)\cap\overline{\mathcal{S}}| = K$, 
we have
\begin{align}
    &I\Big(\{Y_u\}_{u\in[U]\setminus \{k\}},\{X_{k,v}\}_{v\in [V_k]};\{W_{u,v} \}_{(u,v )\in \mathcal{S}_m} \Big|\sum_{(u,v)\in\mathcal{K}}W_{u,v}, \{W_{u,v}, Z_{u,v } \}_{(u,v )\in \mathcal{T}_n} \Big)\notag\\
=&H\Big(\{Y_u\}_{u\in[U]\setminus \{k\}},\{X_{k,v}\}_{v\in [V_k]}
\Big|\sum_{(u,v)\in\mathcal{K}}W_{u,v}, \{W_{u,v}, Z_{u,v } \}_{(u,v )\in \mathcal{T}_n} \Big)\notag\\
&-H\Big(\{Y_u\}_{u\in[U]\setminus \{k\}},\{X_{k,v}\}_{v\in [V_k]}
\Big| \{W_{u,v} \}_{(u,v )\in \mathcal{S}_m},\sum_{(u,v)\in\mathcal{K}}W_{u,v}, \{W_{u,v}, Z_{u,v } \}_{(u,v )\in \mathcal{T}_n} \Big)\\
 \overset{\substack{(\ref{eq:t1t3})\\(\ref{eq:t1t4})}}{=}&H\Big(\Big\{\sum_{v\in [V_u]}\{W_{u,v} + Z_{u,v}\}\Big\}_{u\in[U]\setminus \{k\}},\{W_{k,v} + Z_{k,v}\}_{v\in [V_u]}
\Big|\sum_{(u,v)\in\mathcal{K}}W_{u,v}, \{W_{u,v}, Z_{u,v } \}_{(u,v )\in \mathcal{T}_n} \Big)\notag\\
&-H\Big(\Big\{\sum_{v\in [V_u]}\{W_{u,v} + Z_{u,v}\}\Big\}_{u\in[U]\setminus \{k\}},\{W_{k,v} + Z_{k,v}\}_{v\in [V_u]}
\Big| \{W_{u,v} \}_{(u,v )\in \mathcal{K}}, \{Z_{u,v } \}_{(u,v )\in \mathcal{T}_n} \Big)\label{secpfs1t1}\\
=&H\Big(\Big\{\sum_{v\in [V_u]}\{W_{u,v} + Z_{u,v}\}\Big\}_{u\in[U]\setminus (\{k\}\cup \mathcal{F}^{(m,n)})},\{W_{u,v} + Z_{u,v}\}_{(u,v)\in \mathcal{K}_{\{k\}}\setminus \mathcal{T}_n}
\Big|\notag\\
& \sum_{(u,v)\in\mathcal{K}\setminus (\mathcal{K}_{\mathcal{F}^{(m,n)}}\cup (\mathcal{K}_{\{k\}}\cap \mathcal{T}_n))}W_{u,v}, \{W_{u,v},Z_{u,v } \}_{(u,v )\in \mathcal{T}_n} \Big)\notag\\
& -H\Big(\Big\{\sum_{v\in [V_u]}Z_{u,v}\Big\}_{u\in[U]\setminus (\{k\}\cup \mathcal{F}^{(m,n)})},\{ Z_{u,v}\}_{(u,v)\in \mathcal{K}_{\{k\}}\setminus \mathcal{T}_n}
\Big| \{W_{u,v} \}_{(u,v )\in \mathcal{K}}, \{Z_{u,v } \}_{(u,v )\in \mathcal{T}_n} \Big)\label{secpfs1t2}\\
\leq&H\Big(\Big\{\sum_{v\in [V_u]}\{W_{u,v} + Z_{u,v}\}\Big\}_{u\in[U]\setminus (\{k\}\cup \mathcal{F}^{(m,n)})},\{W_{u,v} + Z_{u,v}\}_{(u,v)\in \mathcal{K}_{\{k\}}\setminus \mathcal{T}_n}
\Big|\sum_{(u,v)\in\mathcal{K}\setminus (\mathcal{K}_{\mathcal{F}^{(m,n)}}\cup (\mathcal{K}_{\{k\}}\cap \mathcal{T}_n))}W_{u,v}\Big)\notag\\
& -H\Big(\Big\{\sum_{v\in [V_u]}Z_{u,v}\Big\}_{u\in[U]\setminus (\{k\}\cup \mathcal{F}^{(m,n)})},\{ Z_{u,v}\}_{(u,v)\in \mathcal{K}_{\{k\}}\setminus \mathcal{T}_n}, \{Z_{u,v } \}_{(u,v )\in \mathcal{T}_n} \Big)+H\Big(\{Z_{u,v } \}_{(u,v )\in \mathcal{T}_n} \Big)\label{secpfs1t3}\\
=&(|[U]\setminus (\{k\}\cup \mathcal{F}^{(m,n)})|+|\mathcal{K}_{\{k\}}\setminus \mathcal{T}_n|-1)-(K-1)+|\mathcal{T}_n|\\
=&|\mathcal{U}^{(m,n)}\setminus \{k\}|+|\mathcal{S}_m\cap\mathcal{K}_{\{k\}}|+|\mathcal{T}_n\cap \overline{\mathcal{S}}|-K\label{secpfs1t4}\\
=&|\mathcal{U}^{(m,n)}\setminus \{k\}|+|(\mathcal{S}_m\cap\mathcal{K}_{\{k\}})\cup\mathcal{T}_n\cap \overline{\mathcal{S}}|-K\\
=&0.\label{secpfs1t5}
\end{align}

We next justify each step in the above derivation.
In (\ref{secpfs1t1}), the second entropy term follows since 
$\mathcal{S}_m \cup \mathcal{T}_n = \mathcal{K}$, 
which implies that all message variables are revealed in the conditioning.
In (\ref{secpfs1t2}), the reduction of the first term follows from the fact that
$\mathcal{K}_{\mathcal{F}^{(m,n)}} 
\cup (\mathcal{K}_{\{k\}} \cap \mathcal{T}_n) \subseteq \mathcal{T}_n$,
and hence the corresponding message variables are already known.
In (\ref{secpfs1t3}), the second entropy term equals $K-1$ since the total number of involved key variables is
$|[U]\setminus (\{k\}\cup \mathcal{F}^{(m,n)})|
+|\mathcal{K}_{\{k\}}\setminus \mathcal{T}_n|
+|\mathcal{T}_n| = K,$
and one linear constraint reduces the entropy by one.
In (\ref{secpfs1t4}), we use
$|\mathcal{K}_{\{k\}}\setminus \mathcal{T}_n|
=|\mathcal{S}_m\cap\mathcal{K}_{\{k\}}|$,
which follows from
$|(\mathcal{S}_m \cup \mathcal{T}_n)\cap \mathcal{K}_{\{k\}}|
=|\mathcal{K}_{\{k\}}|$.
Moreover,
$|\mathcal{T}_n| = |\mathcal{T}_n \cap \overline{\mathcal{S}}|$
since $|\overline{\mathcal{S}}|=K$.
Finally, (\ref{secpfs1t5}) follows from the condition
$|\mathcal{U}^{(m,n)}\setminus \{k\}|
+|(\mathcal{S}_m\cap\mathcal{K}_{\{k\}})
\cup(\mathcal{T}_n\cap \overline{\mathcal{S}})| = e^* = K,$
which implies that the mutual information equals zero.




Second, we consider the schemes in Subsections~\ref{sec:ach2}, \ref{sec:ach3}, and Section~\ref{thm2subsec2} for the case $e^* \leq K-1$, i.e., $|\mathcal{S}_m \cup \mathcal{T}_n| \leq K-1$. 
For any $k,m,n$ satisfying
$|((\mathcal{S}_m \cap \mathcal{K}_{\{k\}})\cup \mathcal{T}_n)\cap\overline{\mathcal{S}}|
+|\mathcal{U}^{(m,n)}\setminus\{k\}| \leq K-1,\\
|((\mathcal{S}_m \cap \mathcal{K}_{\{k\}})
\cup \mathcal{K}_{\mathcal{U}^{(m,n)}}
\cup \mathcal{T}_n)\cap\overline{\mathcal{S}}| \leq K-1,$
we follow an argument identical to that in Subsection~\ref{sec:ach1} and conclude that the mutual information is zero.
\begin{align}
    &I\Big(\{Y_u\}_{u\in[U]\setminus \{k\}},\{X_{k,v}\}_{v\in [V_k]};\{W_{u,v} \}_{(u,v )\in \mathcal{S}_m} \Big|\sum_{(u,v)\in\mathcal{K}}W_{u,v}, \{W_{u,v}, Z_{u,v } \}_{(u,v )\in \mathcal{T}_n} \Big)\notag\\
=&H\Big(\{Y_u\}_{u\in[U]\setminus \{k\}},\{X_{k,v}\}_{v\in [V_k]}
\Big|\sum_{(u,v)\in\mathcal{K}}W_{u,v}, \{W_{u,v}, Z_{u,v } \}_{(u,v )\in \mathcal{T}_n} \Big)\notag\\
&-H\Big(\{Y_u\}_{u\in[U]\setminus \{k\}},\{X_{k,v}\}_{v\in [V_k]}
\Big| \{W_{u,v} \}_{(u,v )\in \mathcal{S}_m},\sum_{(u,v)\in\mathcal{K}}W_{u,v}, \{W_{u,v}, Z_{u,v } \}_{(u,v )\in \mathcal{T}_n} \Big)\\
 =&H\Big(\Big\{\sum_{v\in [V_u]}\{W_{u,v} + Z_{u,v}\}\Big\}_{u\in[U]\setminus \{k\}},\{W_{k,v} + Z_{k,v}\}_{v\in [V_u]}
\Big|\sum_{(u,v)\in\mathcal{K}}W_{u,v}, \{W_{u,v}, Z_{u,v } \}_{(u,v )\in \mathcal{T}_n} \Big)-\notag\\
&H\Big(\Big\{\sum_{v\in [V_u]}\{W_{u,v} + Z_{u,v}\}\Big\}_{u\in[U]\setminus \{k\}},\{W_{k,v} + Z_{k,v}\}_{v\in [V_u]}
\Big| \{W_{u,v} \}_{(u,v )\in \mathcal{S}_m},\sum_{(u,v)\in\mathcal{K}}W_{u,v}, \{W_{u,v}, Z_{u,v } \}_{(u,v )\in \mathcal{T}_n} \Big)\label{secpfs2t1}\\
=&H\Big(\Big\{\sum_{v\in [V_u]}\{W_{u,v} + Z_{u,v}\}\Big\}_{u\in \mathcal{U}^{(m,n)}\setminus \{k\}},\{W_{u,v} + Z_{u,v}\}_{(u,v)\in (\mathcal{S}_m\cap\mathcal{K}_{\{k\}})\setminus \mathcal{T}_n}
\Big| \sum_{(u,v)\in\mathcal{K}}W_{u,v}, \{W_{u,v}, Z_{u,v } \}_{(u,v )\in \mathcal{T}_n} \Big)\notag\\
&+H\Big(\Big\{\sum_{v\in [V_u]}\{W_{u,v} + Z_{u,v}\}\Big\}_{u\in[U]\setminus (\mathcal{U}^{(m,n)}\cup\{k\})},\{W_{u,v} + Z_{u,v}\}_{(u,v)\in \mathcal{K}_{\{k\}}\setminus( \mathcal{S}_m\cup\mathcal{T}_n)}
\Big|\{W_{u,v}, Z_{u,v } \}_{(u,v )\in \mathcal{T}_n},\notag\\
&\sum_{(u,v)\in\mathcal{K}}W_{u,v}, \Big\{\sum_{v\in [V_u]}\{W_{u,v} + Z_{u,v}\}\Big\}_{u\in \mathcal{U}^{(m,n)}\setminus \{k\}},\{W_{u,v} + Z_{u,v}\}_{(u,v)\in (\mathcal{S}_m\cap\mathcal{K}_{\{k\}})\setminus \mathcal{T}_n} \Big)\notag\\
&-H\Big(\Big\{\sum_{v\in [V_u]}\{W_{u,v} + Z_{u,v}\}\Big\}_{u\in \mathcal{U}^{(m,n)}\setminus \{k\}},\{W_{u,v} + Z_{u,v}\}_{(u,v)\in (\mathcal{S}_m\cap\mathcal{K}_{\{k\}})\setminus \mathcal{T}_n}
\Big|\notag\\
& \{W_{u,v} \}_{(u,v )\in \mathcal{S}_m},\sum_{(u,v)\in\mathcal{K}}W_{u,v}, \{W_{u,v}, Z_{u,v } \}_{(u,v )\in \mathcal{T}_n} \Big)\notag\\
&-H\Big(\Big\{\sum_{v\in [V_u]}\{W_{u,v} + Z_{u,v}\}\Big\}_{u\in[U]\setminus (\mathcal{U}^{(m,n)}\cup\{k\})},\{W_{u,v} + Z_{u,v}\}_{(u,v)\in \mathcal{K}_{\{k\}}\setminus( \mathcal{S}_m\cup\mathcal{T}_n)}
\Big|\{W_{u,v}, Z_{u,v } \}_{(u,v )\in \mathcal{T}_n},\notag\\
&  \{W_{u,v} \}_{(u,v )\in \mathcal{S}_m},\sum_{(u,v)\in\mathcal{K}}W_{u,v}, \Big\{\sum_{v\in [V_u]}\{W_{u,v} + Z_{u,v}\}\Big\}_{u\in \mathcal{U}^{(m,n)}\setminus \{k\}},\{W_{u,v} + Z_{u,v}\}_{(u,v)\in (\mathcal{S}_m\cap\mathcal{K}_{\{k\}})\setminus \mathcal{T}_n} \Big)\\
\leq &H\Big(\Big\{\sum_{v\in [V_u]}\{W_{u,v} + Z_{u,v}\}\Big\}_{u\in \mathcal{U}^{(m,n)}\setminus \{k\}},\{W_{u,v} + Z_{u,v}\}_{(u,v)\in (\mathcal{S}_m\cap\mathcal{K}_{\{k\}})\setminus \mathcal{T}_n} \Big)\notag \\
&+H\Big(\Big\{\sum_{v\in [V_u]}\{W_{u,v} + Z_{u,v}\}\Big\}_{u\in[U]\setminus (\mathcal{U}^{(m,n)}\cup\{k\}\cup\mathcal{F}^{(m,n)})},\{W_{u,v} + Z_{u,v}\}_{(u,v)\in \mathcal{K}_{\{k\}}\setminus( \mathcal{S}_m\cup\mathcal{T}_n)}\Big|\notag\\
&\sum_{(u,v)\in \mathcal{K}_{[U]}\setminus(\mathcal{K}_{(\mathcal{U}^{(m,n)}\cup\mathcal{F}^{(m,n)})\setminus \{k\}}\cup((\mathcal{S}_m\cup\mathcal{T}_n)\cap\mathcal{K}_{\{k\}}))}(W_{u,v}+Z_{u,v}) \Big)\notag\\
&-H\Big(\Big\{\sum_{v\in [V_u]}Z_{u,v}\Big\}_{u\in \mathcal{U}^{(m,n)}\setminus \{k\}},\{ Z_{u,v}\}_{(u,v)\in (\mathcal{S}_m\cap\mathcal{K}_{\{k\}})\setminus \mathcal{T}_n}
\Big| \{W_{u,v} \}_{(u,v )\in \mathcal{K}},\{ Z_{u,v } \}_{(u,v )\in \mathcal{T}_n} \Big)\notag\\
&-H\Big(\Big\{\sum_{v\in [V_u]}\{W_{u,v}\}\Big\}_{u\in[U]\setminus (\mathcal{U}^{(m,n)}\cup\{k\}\cup\mathcal{F}^{(m,n)})},\{ W_{u,v}\}_{(u,v)\in \mathcal{K}_{\{k\}}\setminus( \mathcal{S}_m\cup\mathcal{T}_n)}
\Big|\{ Z_{u,v } \}_{(u,v )\in \mathcal{K}},\notag\\
&  \{W_{u,v} \}_{(u,v )\in (\mathcal{S}_m\cup\mathcal{T}_n)},\sum_{(u,v)\in \mathcal{K}_{[U]}\setminus(\mathcal{K}_{(\mathcal{U}^{(m,n)}\cup\mathcal{F}^{(m,n)})\setminus \{k\}}\cup((\mathcal{S}_m\cup\mathcal{T}_n)\cap\mathcal{K}_{\{k\}}))}W_{u,v} \Big)\label{secpfs2t2}\\
=&(|\mathcal{U}^{(m,n)}\setminus \{k\}|+|(\mathcal{S}_m\cap\mathcal{K}_{\{k\}})\setminus \mathcal{T}_n|)L+ (|[U]\setminus(\mathcal{U}^{(m,n)}\cup\{k\}\cup\mathcal{F}^{(m,n)})|+|\mathcal{K}_{\{k\}}\setminus( \mathcal{S}_m\cup\mathcal{T}_n)|-1)L\notag\\
&-(|\mathcal{U}^{(m,n)}\setminus \{k\}|+|(\mathcal{S}_m\cap\mathcal{K}_{\{k\}})\setminus \mathcal{T}_n|)L- (|[U]\setminus(\mathcal{U}^{(m,n)}\cup\{k\}\cup\mathcal{F}^{(m,n)})|+|\mathcal{K}_{\{k\}}\setminus( \mathcal{S}_m\cup\mathcal{T}_n)|-1)L\\
=&0.
\end{align}
In (\ref{secpfs2t1}), we substitute the expressions of the transmitted messages,
namely $X_{u,v}=W_{u,v}+Z_{u,v}$ and $Y_u=\sum_{v\in[V_u]}X_{u,v}$, as defined in
(\ref{eq:achm1X1}), (\ref{eq:achm1Y11}), (\ref{eq:achm2X}), (\ref{eq:achm2Y}),
(\ref{eq:uppm3X}), and (\ref{eq:uppm3Y}).
In (\ref{secpfs2t2}), the inequalities follow from the fact that conditioning reduces entropy,
together with the independence between the message symbols and the key variables.
The second term is obtained using the set inclusions
$(\mathcal{S}_m\cup\mathcal{T}_n)\cap\mathcal{K}_{\{k\}}
\subseteq(\mathcal{S}_m\cap\mathcal{K}_{\{k\}})\cup\mathcal{T}_n$ and
$(\mathcal{U}^{(m,n)}\cup\mathcal{F}^{(m,n)})\setminus\{k\}
\subseteq(\mathcal{U}^{(m,n)}\setminus\{k\})\cup\mathcal{F}^{(m,n)}$.
Moreover, since $\mathcal{K}_{\mathcal{F}^{(m,n)}}\subseteq\mathcal{T}_n$, we have
$\mathcal{K}_{(\mathcal{U}^{(m,n)}\cup\mathcal{F}^{(m,n)})\setminus\{k\}}
\cup((\mathcal{S}_m\cup\mathcal{T}_n)\cap\mathcal{K}_{\{k\}})
\subseteq
\mathcal{K}_{\mathcal{U}^{(m,n)}\setminus\{k\}}
\cup((\mathcal{S}_m\cap\mathcal{K}_{\{k\}})\cup\mathcal{T}_n).$
Similarly, the fourth term follows since
$\mathcal{K}_{(\mathcal{U}^{(m,n)}\cup\mathcal{F}^{(m,n)})\setminus\{k\}}
\cup((\mathcal{S}_m\cup\mathcal{T}_n)\cap\mathcal{K}_{\{k\}})
\subseteq
\mathcal{S}_m\cup\mathcal{T}_n.$
The third term evaluates to
$(|\mathcal{U}^{(m,n)}\setminus\{k\}|+
|(\mathcal{S}_m\cap\mathcal{K}_{\{k\}})\setminus\mathcal{T}_n|)L$,
which will be established in the next lemma.
Finally, the positive and negative entropy terms cancel exactly, implying that
the entropy contributed by the keys matches the remaining uncertainty of the message components, and hence the mutual information equals zero.

To complete the proof, it remains to show the following lemma.

\begin{lemma}\label{lemma:independent}
\tit{Under condition when
Under the achievable schemes constructed in Subsections \ref{sec:ach1}, \ref{sec:ach2}, \ref{sec:ach3}, and Section~\ref{thm2subsec2}, we have}
\begin{eqnarray}
&&H\Big(\Big\{\sum_{v\in [V_u]}Z_{u,v}\Big\}_{u\in \mathcal{U}^{(m,n)}\setminus \{k\}},\{ Z_{u,v}\}_{(u,v)\in (\mathcal{S}_m\cap\mathcal{K}_{\{k\}})\setminus \mathcal{T}_n}
\Big| \{ Z_{u,v } \}_{(u,v )\in \mathcal{T}_n} \Big) \notag\\
&=&(|\mathcal{U}^{(m,n)}\setminus \{k\}|+|(\mathcal{S}_m\cap\mathcal{K}_{\{k\}})\setminus \mathcal{T}_n|)L.
\label{eq:independent}
\end{eqnarray}
\end{lemma}



We prove each case separately. First, consider the scheme in Subsection~\ref{sec:ach1}. 
For any $k,m,n$ such that 
$|((\mathcal{S}_m \cap \mathcal{K}_{\{k\}})\cup \mathcal{T}_n)\cap\overline{\mathcal{S}}| 
+|\mathcal{U}^{(m,n)}\setminus\{k\}| \leq K-1$ and 
$|((\mathcal{S}_m \cap \mathcal{K}_{\{k\}}) 
\cup \mathcal{K}_{\mathcal{U}^{(m,n)}} 
\cup \mathcal{T}_n)\cap\overline{\mathcal{S}}| \leq K-1$, 
we establish (\ref{eq:independent}).
\begin{align}
&H\Big(\Big\{\sum_{v\in [V_u]}Z_{u,v}\Big\}_{u\in \mathcal{U}^{(m,n)}\setminus \{k\}},\{ Z_{u,v}\}_{(u,v)\in (\mathcal{S}_m\cap\mathcal{K}_{\{k\}})\setminus \mathcal{T}_n}
\Big| \{ Z_{u,v } \}_{(u,v )\in \mathcal{T}_n} \Big) \notag\\
=&H\Big(\Big\{\sum_{v\in [V_u]}Z_{u,v}\Big\}_{u\in \mathcal{U}^{(m,n)}\setminus \{k\}},\{ Z_{u,v}\}_{(u,v)\in (\mathcal{S}_m\cap\mathcal{K}_{\{k\}})\setminus \mathcal{T}_n}, \{ Z_{u,v } \}_{(u,v )\in \mathcal{T}_n} \Big)-H(\{ Z_{u,v } \}_{(u,v )\in \mathcal{T}_n}) \\
=&H\Big(\Big\{\sum_{v\in [V_u]}Z_{u,v}\Big\}_{u\in \mathcal{U}^{(m,n)}\setminus \{k\}},\{ Z_{u,v}\}_{(u,v)\in (\mathcal{S}_m\cap\mathcal{K}_{\{k\}})\cup \mathcal{T}_n}\Big)-H(\{ Z_{u,v } \}_{(u,v )\in \mathcal{T}_n}) \\
=&H\Big(\Big\{\sum_{v\in [V_u]}Z_{u,v}\Big\}_{u\in \mathcal{U}^{(m,n)}\setminus \{k\}},\{ Z_{u,v}\}_{(u,v)\in ((\mathcal{S}_m\cap\mathcal{K}_{\{k\}})\cup \mathcal{T}_n)\cap \overline{\mathcal{S}}}\Big)-H(\{ Z_{u,v } \}_{(u,v )\in \mathcal{T}_n\cap \overline{\mathcal{S}}}) \label{lem7pfs1t1}\\
=&(|\mathcal{U}^{(m,n)}\setminus \{k\}|+|((\mathcal{S}_m\cap\mathcal{K}_{\{k\}})\cup \mathcal{T}_n)\cap \overline{\mathcal{S}}|)L- |\mathcal{T}_n\cap \overline{\mathcal{S}}|L\label{lem7pfs1t2}\\
=&(|\mathcal{U}^{(m,n)}\setminus \{k\}|+|(\mathcal{S}_m\cap\mathcal{K}_{\{k\}})\setminus \mathcal{T}_n|)L,
\end{align}
where (\ref{lem7pfs1t1}) follows from the fact that $|\overline{\mathcal{S}}|=K$ (see (\ref{eq:t1t1})). 
To obtain (\ref{lem7pfs1t2}), we use the condition 
$|((\mathcal{S}_m \cap \mathcal{K}_{\{k\}})\cup \mathcal{T}_n)\cap\overline{\mathcal{S}}| 
+|\mathcal{U}^{(m,n)}\setminus\{k\}| \leq K-1$, 
which ensures that any set of at most $K-1$ linear combinations of the $K$ keys is linearly independent (see (\ref{eq:t1t1})).
Note that the case $e^*=K$ has already been treated previously, and here we focus on the regime $e^*\leq K-1$. 
The last equality follows from the fact that $(\mathcal{S}_m \cap \mathcal{K}_{\{k\}}) \subseteq \overline{\mathcal{S}}$.

Second, consider the scheme in Section~\ref{sec:ach2}. We establish (\ref{eq:independent}) as follows:
\begin{align}
&H\Big(\Big\{\sum_{v\in [V_u]}Z_{u,v}\Big\}_{u\in \mathcal{U}^{(m,n)}\setminus \{k\}},\{ Z_{u,v}\}_{(u,v)\in (\mathcal{S}_m\cap\mathcal{K}_{\{k\}})\setminus \mathcal{T}_n}
\Big| \{ Z_{u,v } \}_{(u,v )\in \mathcal{T}_n} \Big) \notag\\
=&H\Big(\Big\{\sum_{v\in [V_u]}Z_{u,v}\Big\}_{u\in \mathcal{U}^{(m,n)}\setminus \{k\}},\{ Z_{u,v}\}_{(u,v)\in (\mathcal{S}_m\cap\mathcal{K}_{\{k\}})\setminus \mathcal{T}_n}, \{ Z_{u,v } \}_{(u,v )\in \mathcal{T}_n} \Big)-H(\{ Z_{u,v } \}_{(u,v )\in \mathcal{T}_n}) \\
=&H\Big(\Big\{\sum_{v\in [V_u]}Z_{u,v}\Big\}_{u\in \mathcal{U}^{(m,n)}\setminus \{k\}},\{ Z_{u,v}\}_{(u,v)\in (\mathcal{S}_m\cap\mathcal{K}_{\{k\}})\cup \mathcal{T}_n}\Big)-H(\{ Z_{u,v } \}_{(u,v )\in \mathcal{T}_n}) \\
=&H\Big(\Big\{\sum_{v\in [V_u]}Z_{u,v}\Big\}_{u\in \mathcal{U}^{(m,n)}\setminus \{k\}},\{ Z_{u,v}\}_{(u,v)\in ((\mathcal{S}_m\cap\mathcal{K}_{\{k\}})\cup \mathcal{T}_n)\cap \overline{\mathcal{S}}}\Big)-H(\{ Z_{u,v } \}_{(u,v )\in \mathcal{T}_n\cap \overline{\mathcal{S}}}) \label{lem7pfs2t1}\\
=&(|\mathcal{U}^{(m,n)}\setminus \{k\}|+|((\mathcal{S}_m\cap\mathcal{K}_{\{k\}})\cup \mathcal{T}_n)\cap \overline{\mathcal{S}}|)L- |\mathcal{T}_n\cap \overline{\mathcal{S}}|L\label{lem7pfs2t2}\\
=&(|\mathcal{U}^{(m,n)}\setminus \{k\}|+|(\mathcal{S}_m\cap\mathcal{K}_{\{k\}})\setminus \mathcal{T}_n|)L, 
\end{align}
where (\ref{lem7pfs2t1}) follows from the fact that 
$Z_{u,v}=0$ for $(u,v)\in \mathcal{K}_{[U]}\setminus\overline{\mathcal{S}}$ (see (\ref{eq:c111})).
To obtain (\ref{lem7pfs2t2}), we use the conditions 
$|\mathcal{A}_{k,m,n}|=
|((\mathcal{S}_m \cap \mathcal{K}_{\{k\}})
\cup \mathcal{K}_{\mathcal{U}^{(m,n)}}
\cup \mathcal{T}_n)\cap\overline{\mathcal{S}}|
\le K-1$
and 
$|((\mathcal{S}_m \cap \mathcal{K}_{\{k\}})\cup \mathcal{T}_n)\cap\overline{\mathcal{S}}|
+|\mathcal{U}^{(m,n)}\setminus\{k\}|
\le e^*$, 
which ensure that any set of at most $K-1$ vectors ${\bf h}_{u,v}$ is linearly independent (see (\ref{eq:sz111})).
The last equality follows from the fact that 
$(\mathcal{S}_m \cap \mathcal{K}_{\{k\}}) \subseteq \overline{\mathcal{S}}$.

Third, consider the scheme in Section~\ref{sec:ach3}. 
We establish (\ref{eq:independent}) by considering two sub-cases.
For the first sub-case, $(u',v') \in \mathcal{T}_n$. Then
$\big|(\mathcal{S}_m \cap \mathcal{K}_{\{k\}}) 
\cup \mathcal{K}_{\mathcal{U}^{(m,n)}} 
\cup \mathcal{T}_n\cap\overline{\mathcal{S}}\big| 
< |\overline{\mathcal{S}}|$
and
$|((\mathcal{S}_m \cap \mathcal{K}_{\{k\}})\cup \mathcal{T}_n)\cap\overline{\mathcal{S}}|
+|\mathcal{U}^{(m,n)}\setminus\{k\}|\leq e^*-1.$
Otherwise, either condition would imply that $(u',v') \in \mathcal{Q}$ according to Definitions~\ref{def:totset1} and~\ref{def:uni2}, which contradicts the choice of $(u',v') \notin \mathcal{Q}$ in Section~\ref{sec:ach2}.
\begin{align}
&H\Big(\Big\{\sum_{v\in [V_u]}Z_{u,v}\Big\}_{u\in \mathcal{U}^{(m,n)}\setminus \{k\}},\{ Z_{u,v}\}_{(u,v)\in (\mathcal{S}_m\cap\mathcal{K}_{\{k\}})\setminus \mathcal{T}_n}
\Big| \{ Z_{u,v } \}_{(u,v )\in \mathcal{T}_n} \Big) \notag\\
=&H\Big(\Big\{\sum_{v\in [V_u]}Z_{u,v}\Big\}_{u\in \mathcal{U}^{(m,n)}\setminus \{k\}},\{ Z_{u,v}\}_{(u,v)\in (\mathcal{S}_m\cap\mathcal{K}_{\{k\}})\setminus \mathcal{T}_n}, \{ Z_{u,v } \}_{(u,v )\in \mathcal{T}_n} \Big)-H(\{ Z_{u,v } \}_{(u,v )\in \mathcal{T}_n}) \\
=&H\Big(\Big\{\sum_{v\in [V_u]}Z_{u,v}\Big\}_{u\in \mathcal{U}^{(m,n)}\setminus \{k\}},\{ Z_{u,v}\}_{(u,v)\in (\mathcal{S}_m\cap\mathcal{K}_{\{k\}})\cup \mathcal{T}_n}\Big)-H(\{ Z_{u,v } \}_{(u,v )\in \mathcal{T}_n}) \\
=&H\Big(\Big\{\sum_{v\in [V_u]}Z_{u,v}\Big\}_{u\in \mathcal{U}^{(m,n)}\setminus \{k\}},\{ Z_{u,v}\}_{(u,v)\in ((\mathcal{S}_m\cap\mathcal{K}_{\{k\}})\cup \mathcal{T}_n)\cap (\overline{\mathcal{S}} \cup \{(u',v')\})}\Big)-H(\{ Z_{u,v } \}_{(u,v )\in \mathcal{T}_n\cap (\overline{\mathcal{S}} \cup \{(u',v')\})}) \\
=&(|\mathcal{U}^{(m,n)}\setminus \{k\}|+|((\mathcal{S}_m\cap\mathcal{K}_{\{k\}})\cup \mathcal{T}_n)\cap (\overline{\mathcal{S}} \cup \{(u',v')\})|)L- |\mathcal{T}_n\cap (\overline{\mathcal{S}} \cup \{(u',v')\})|L\label{lem7pfs3t1}\\
=&(|\mathcal{U}^{(m,n)}\setminus \{k\}|+|(\mathcal{S}_m\cap\mathcal{K}_{\{k\}})\setminus \mathcal{T}_n|)L.
\end{align}
Since 
$|((\mathcal{S}_m \cap \mathcal{K}_{\{k\}})\cup \mathcal{T}_n)\cap\overline{\mathcal{S}}|
+|\mathcal{U}^{(m,n)}\setminus\{k\}|\leq e^*-1$,
the total number of involved vectors is at most $e^*$.
Hence, by (\ref{eq:sz2}), these vectors are linearly independent, and the entropy equals the cardinality.
The last equality follows from the fact that $(u',v') \in \mathcal{T}_n$.

For the second sub-case, $(u',v') \notin \mathcal{T}_n$. 
Note that $(u',v') \notin \overline{\mathcal{S}}$ and $\mathcal{S}_m \subseteq \overline{\mathcal{S}}$. We have
\begin{align}
&H\Big(\Big\{\sum_{v\in [V_u]}Z_{u,v}\Big\}_{u\in \mathcal{U}^{(m,n)}\setminus \{k\}},\{ Z_{u,v}\}_{(u,v)\in (\mathcal{S}_m\cap\mathcal{K}_{\{k\}})\setminus \mathcal{T}_n}
\Big| \{ Z_{u,v } \}_{(u,v )\in \mathcal{T}_n} \Big) \notag\\
=&H\Big(\Big\{\sum_{v\in [V_u]}Z_{u,v}\Big\}_{u\in \mathcal{U}^{(m,n)}\setminus \{k\}},\{ Z_{u,v}\}_{(u,v)\in (\mathcal{S}_m\cap\mathcal{K}_{\{k\}})\setminus \mathcal{T}_n}, \{ Z_{u,v } \}_{(u,v )\in \mathcal{T}_n} \Big)-H(\{ Z_{u,v } \}_{(u,v )\in \mathcal{T}_n}) \\
=&H\Big(\Big\{\sum_{v\in [V_u]}Z_{u,v}\Big\}_{u\in \mathcal{U}^{(m,n)}\setminus \{k\}},\{ Z_{u,v}\}_{(u,v)\in (\mathcal{S}_m\cap\mathcal{K}_{\{k\}})\cup \mathcal{T}_n}\Big)-H(\{ Z_{u,v } \}_{(u,v )\in \mathcal{T}_n}) \\
=&H\Big(\Big\{\sum_{v\in [V_u]}Z_{u,v}\Big\}_{u\in \mathcal{U}^{(m,n)}\setminus \{k\}},\{ Z_{u,v}\}_{(u,v)\in ((\mathcal{S}_m\cap\mathcal{K}_{\{k\}})\cup \mathcal{T}_n)\cap (\overline{\mathcal{S}} \cup \{(u',v')\})}\Big)-H(\{ Z_{u,v } \}_{(u,v )\in \mathcal{T}_n\cap (\overline{\mathcal{S}} \cup \{(u',v')\})}) \\
=&(|\mathcal{U}^{(m,n)}\setminus \{k\}|+|((\mathcal{S}_m\cap\mathcal{K}_{\{k\}})\cup \mathcal{T}_n)\cap \overline{\mathcal{S}} |)L- |\mathcal{T}_n\cap \overline{\mathcal{S}}|L\label{lem7pfs3t2}\\
=&(|\mathcal{U}^{(m,n)}\setminus \{k\}|+|(\mathcal{S}_m\cap\mathcal{K}_{\{k\}})\setminus \mathcal{T}_n|)L.
\end{align}
Since 
$|((\mathcal{S}_m \cap \mathcal{K}_{\{k\}})\cup \mathcal{T}_n)\cap\overline{\mathcal{S}}|
+|\mathcal{U}^{(m,n)}\setminus\{k\}|\leq e^*$,
the corresponding vectors ${\bf h}_{u,v}$ are linearly independent by (\ref{eq:sz2}).
Moreover, $(u',v') \notin \mathcal{T}_n$ implies that it does not contribute to the conditioned set, leading to (\ref{lem7pfs3t2}).

Finally, consider the scheme in Section \ref{thm2subsec2}. 
We prove (\ref{eq:independent}) as follows:
\begin{align}
&H\Big(\Big\{\sum_{v\in [V_u]}Z_{u,v}\Big\}_{u\in \mathcal{U}^{(m,n)}\setminus \{k\}},\{ Z_{u,v}\}_{(u,v)\in (\mathcal{S}_m\cap\mathcal{K}_{\{k\}})\setminus \mathcal{T}_n}
\Big| \{ Z_{u,v } \}_{(u,v )\in \mathcal{T}_n} \Big) \notag\\
=&H\Big(\Big\{\sum_{v\in [V_u]}Z_{u,v}\Big\}_{u\in \mathcal{U}^{(m,n)}\setminus \{k\}},\{ Z_{u,v}\}_{(u,v)\in (\mathcal{S}_m\cap\mathcal{K}_{\{k\}})\setminus \mathcal{T}_n}, \{ Z_{u,v } \}_{(u,v )\in \mathcal{T}_n} \Big)-H(\{ Z_{u,v } \}_{(u,v )\in \mathcal{T}_n}) \\
=&H\Big(\Big\{\sum_{v\in [V_u]}Z_{u,v}\Big\}_{u\in \mathcal{U}^{(m,n)}\setminus \{k\}},\{ Z_{u,v}\}_{(u,v)\in (\mathcal{S}_m\cap\mathcal{K}_{\{k\}})\cup \mathcal{T}_n}\Big)-H(\{ Z_{u,v } \}_{(u,v )\in \mathcal{T}_n}) \\
=&H\Big(\Big\{\sum_{v\in [V_u]}Z_{u,v}\Big\}_{u\in \mathcal{U}^{(m,n)}\setminus \{k\}},\{ Z_{u,v}\}_{(u,v)\in (((\mathcal{S}_m\cap\mathcal{K}_{\{k\}})\cup \mathcal{T}_n)\cap \overline{\mathcal{S}})\cup (\mathcal{T}_n \setminus \overline{\mathcal{S}})}\Big)-H(\{ Z_{u,v } \}_{(u,v )\in (\mathcal{T}_n\cap \overline{\mathcal{S}})\cup (\mathcal{T}_n \setminus \overline{\mathcal{S}})}) \\
=&\Big(|\mathcal{U}^{(m,n)}\setminus \{k\}|\overline{q}+|((\mathcal{S}_m\cap\mathcal{K}_{\{k\}})\cup \mathcal{T}_n)\cap \overline{\mathcal{S}}|\overline{q} + \sum_{(u,v) \in \mathcal{T}_n \setminus \overline{\mathcal{S}}} p_{u,v} \Big)- \Big(|\mathcal{T}_n\cap \overline{\mathcal{S}}|\overline{q} + \sum_{(u,v) \in \mathcal{T}_n \setminus \overline{\mathcal{S}}} p_{u,v} \Big)\label{lem7pfs4t1}\\
=&|\mathcal{U}^{(m,n)}\setminus \{k\}|\overline{q}+|(\mathcal{S}_m\cap\mathcal{K}_{\{k\}})\setminus \mathcal{T}_n|\overline{q}\\
=&(|\mathcal{U}^{(m,n)}\setminus \{k\}|+|(\mathcal{S}_m\cap\mathcal{K}_{\{k\}})\setminus \mathcal{T}_n|)L, 
\end{align}
where $L=\overline q$ for the scheme in Section \ref{thm2subsec2}, which is used in the last step.
To justify (\ref{lem7pfs4t1}), we note that 
$H((Z_{u,v})_{(u,v)\in\mathcal T_n\setminus\overline{\mathcal S}})$ 
is determined by the first $p_{u,v}$ rows, since 
${\bf F}_{u,v}{\bf G}_{u,v}$ has rank $p_{u,v}$ (see (\ref{schpft1})).
It remains to verify that the first term in (\ref{lem7pfs4t1}) does not exceed $(e^*+b^*)\overline q$. 
Indeed, we have
$|\mathcal{U}^{(m,n)}\setminus \{k\}|\overline q
+|((\mathcal S_m\cap\mathcal K_{\{k\}})\cup\mathcal T_n)\cap\overline{\mathcal S}|\overline q
+\sum_{(u,v)\in\mathcal T_n\setminus\overline{\mathcal S}} p_{u,v} 
\le\; e^*\overline q 
+ \overline q \sum_{(u,v)\in\mathcal K\setminus\overline{\mathcal S}} b_{u,v}^*
\overset{(\ref{eq:pq111})}{=}(e^*+b^*)\overline q.$


\section{Conclusion} \label{conclu}
In this paper, we studied the fundamental limits of multi-server secure aggregation under arbitrary collusion and heterogeneous security constraints. For a general two-hop network, we characterized the optimal communication rates for all parameter regimes and established tight information-theoretic bounds on the minimum key rate for most cases. In particular, we showed that the optimal key rate is determined by $e^*$, the maximum number of inputs that must be simultaneously protected under critical server--colluder configurations. For the remaining regime, we derived a general converse bound and proposed a linear-program-based achievable scheme, yielding a bounded-gap characterization.
Our results provide a unified information-theoretic framework that strictly generalizes existing secure aggregation models and reveals the fundamental role of the interaction between topology, collusion, and heterogeneous security requirements in determining communication and randomness costs.

An important open problem is to move beyond the product structure imposed in this work, where security must hold for all pairs $(\mathcal{S}_m,\mathcal{T}_n)$. A more general individual-pair model, where only a prescribed subset of such pairs is enforced, appears significantly more challenging and is not captured by current techniques.
More broadly, extending the model to incorporate practical features in federated learning, such as user dropout, user selection, and groupwise key structures, as well as relaxing the security constraints, may lead to new insights on the trade-offs among security, efficiency, and robustness.

\bibliographystyle{IEEEtran}
\bibliography{references_secagg.bib}
\end{document}

%% file: macros.tex
\long\def\comment#1{}

\newtheorem{example}{Example}
\newtheorem{theorem}{Theorem}
\newtheorem{lemma}{Lemma}
\newtheorem{corollary}{Corollary}
\newtheorem{proposition}{Proposition}
\newtheorem{remark}{Remark}
\makeatletter
\newcommand{\subalign}[1]{%
  \vcenter{%
    \Let@ \restore@math@cr \default@tag
    \baselineskip\fontdimen10 \scriptfont\tw@
    \advance\baselineskip\fontdimen12 \scriptfont\tw@
    \lineskip\thr@@\fontdimen8 \scriptfont\thr@@
    \lineskiplimit\lineskip
    \ialign{\hfil$\m@th\scriptstyle##$&$\m@th\scriptstyle{}##$\crcr
      #1\crcr
    }%
  }
}
\newcommand{\thickhline}{%
    \noalign {\ifnum 0=`}\fi \hrule height 1pt
    \futurelet \reserved@a \@xhline
}
\newcolumntype{"}{@{\hskip\tabcolsep\vrule width 1pt\hskip\tabcolsep}}

\makeatother

\newfont{\bbb}{msbm10 scaled 700}

\newfont{\bb}{msbm10 scaled 1100}

\def \rzsigma{{R_{Z_{\Sigma}}}}

\def \rx{{R_X}} 
\def \ry{{R_Y}}

\newcommand{\af}{as follows\xspace}

\newcommand{\eg}{e.g.}

\newcommand{\secagg}{secure aggregation\xspace}

\newcommand{\indep}{independent\xspace}

\newcommand{\indiv}{individual\xspace}

\newcommand{\comm}{communication\xspace}
\newcommand{\Comm}{Communication\xspace}

\newcommand{\Ac}{{\cal A}}
\newcommand{\Bc}{{\cal B}}

\newcommand{\Pc}{{\cal P}}

\newcommand{\Sc}{{\cal S}}

\newcommand{\Vc}{{\cal V}}

\newcommand{\eqdef}{\stackrel{\Delta}{=}}

\newcommand{\be}{\begin{equation}}
\newcommand{\ee}{\end{equation}}
\newcommand{\bea}{\begin{eqnarray}}
\newcommand{\eea}{\end{eqnarray}}

\let\tbf\textbf
\let\tit\textit

\let\trm\textrm

%% file: author_TIT.tex
\author{
Zhou Li,~\IEEEmembership{Member,~IEEE},
Xiang~Zhang,~\IEEEmembership{Member,~IEEE}, 
Jiguang He,~\IEEEmembership{Senior~Member,~IEEE},
and Giuseppe Caire,~\IEEEmembership{Fellow,~IEEE}


\thanks{Z. Li is with the Guangxi Key Laboratory of Multimedia Communications and Network Technology, 
Guangxi University, Nanning 530004, China (e-mail: lizhou@gxu.edu.cn).}

\thanks{X. Zhang and G. Caire are with the Department of Electrical Engineering and Computer Science, Technical University of Berlin, 10623 Berlin, Germany (e-mail: \{xiang.zhang, caire\}@tu-berlin.de).
}

\thanks{Jiguang He is with the School of Computing and Information Technology, the Dongguan Key Laboratory for Intelligence and Information Technology, and the Great Bay Institute for Advanced Study, Great Bay University, Dongguan 523000, China (e-mail: jiguang.he@gbu.edu.cn).}





}

%% file: references_secagg.bib
@STRING{isit        = {Proc. {IEEE} Int. Symp. on Inform. Theory (ISIT)}}

@STRING{itw         = {Proc. {IEEE} Inform. Theory Workshop (ITW)}}

@inproceedings{zhao2024secure,
  title={Secure Summation with User Selection and Collusion},
  author={Zhao, Yizhou and Sun, Hua},
  booktitle={2024 IEEE Information Theory Workshop (ITW)},
  pages={175--180},
  year={2024},
  organization={IEEE}
}

@article{li2025weakly,
  title={Weakly secure summation with colluding users},
  author={Li, Zhou and Zhao, Yizhou and Sun, Hua},
  journal={IEEE Transactions on Information Theory},
  year={2025},
  publisher={IEEE}
}

@article{zhang2025fundamental,
  title={Fundamental Limits of Hierarchical Secure Aggregation with Cyclic User Association},
  author={Zhang, Xiang and Li, Zhou and Wan, Kai and Sun, Hua and Ji, Mingyue and Caire, Giuseppe},
  journal={arXiv preprint arXiv:2503.04564},
  year={2025}
}

@misc{yuan2025vectorlinearsecureaggregation,
      title={Vector Linear Secure Aggregation}, 
      author={Xihang Yuan and Hua Sun},
      year={2025},
      eprint={2502.09817},
      archivePrefix={arXiv},
      primaryClass={cs.IT},
      url={https://arxiv.org/abs/2502.09817}, 
}

@article{wan2024information,
  title={On the information theoretic secure aggregation with uncoded groupwise keys},
  author={Wan, Kai and Yao, Xin and Sun, Hua and Ji, Mingyue and Caire, Giuseppe},
  journal={IEEE Transactions on Information Theory},
  year={2024},
  publisher={IEEE}
}

@article{jahani2023swiftagg+,
  title={SwiftAgg+: Achieving asymptotically optimal communication loads in secure aggregation for federated learning},
  author={Jahani-Nezhad, Tayyebeh and Maddah-Ali, Mohammad Ali and Li, Songze and Caire, Giuseppe},
  journal={IEEE Journal on Selected Areas in Communications},
  volume={41},
  number={4},
  pages={977--989},
  year={2023},
  publisher={IEEE}
}

@inproceedings{jahani2022swiftagg,
  title={Swiftagg: Communication-efficient and dropout-resistant secure aggregation for federated learning with worst-case security guarantees},
  author={Jahani-Nezhad, Tayyebeh and Maddah-Ali, Mohammad Ali and Li, Songze and Caire, Giuseppe},
  booktitle={2022 IEEE International Symposium on Information Theory (ISIT)},
  pages={103--108},
  year={2022},
  organization={IEEE}
}

@article{bonawitz2016practical,
  title={Practical secure aggregation for federated learning on user-held data},
  author={Bonawitz, Keith and Ivanov, Vladimir and Kreuter, Ben and Marcedone, Antonio and McMahan, H Brendan and Patel, Sarvar and Ramage, Daniel and Segal, Aaron and Seth, Karn},
  journal={arXiv preprint arXiv:1611.04482},
  year={2016}
}

@article{geiping2020inverting,
  title={Inverting gradients-how easy is it to break privacy in federated learning?},
  author={Geiping, Jonas and Bauermeister, Hartmut and Dr{\"o}ge, Hannah and Moeller, Michael},
  journal={Advances in neural information processing systems},
  volume={33},
  pages={16937--16947},
  year={2020}
}

@article{yang2018applied,
  title={Applied federated learning: Improving google keyboard query suggestions},
  author={Yang, Timothy and Andrew, Galen and Eichner, Hubert and Sun, Haicheng and Li, Wei and Kong, Nicholas and Ramage, Daniel and Beaufays, Fran{\c{c}}oise},
  journal={arXiv preprint arXiv:1812.02903},
  year={2018}
}

@inproceedings{karakocc2021secure,
  title={Secure aggregation against malicious users},
  author={Karako{\c{c}}, Ferhat and {\"O}nen, Melek and Bilgin, Zeki},
  booktitle={Proceedings of the 26th ACM Symposium on Access Control Models and Technologies},
  pages={115--124},
  year={2021}
}

@article{sun2023secure,
  title={Secure Aggregation with an Oblivious Server},
  author={Sun, Hua},
  journal={arXiv preprint arXiv:2307.13474},
  year={2023}
}

@inproceedings{zhao2023optimal,
  title={The Optimal Rate of MDS Variable Generation},
  author={Zhao, Yizhou and Sun, Hua},
  booktitle={2023 IEEE International Symposium on Information Theory (ISIT)},
  pages={832--837},
  year={2023},
  organization={IEEE}
}

@article{zhao2022mds,
  title={MDS Variable Generation and Secure Summation with User Selection},
  author={Zhao, Yizhou and Sun, Hua},
  journal={arXiv preprint arXiv:2211.01220},
  year={2022}
}

@article{zhang2025secure,
  title={On Secure Aggregation with Uncoded Groupwise Keys Against User Dropouts and User Collusion},
  author={Zhang, Ziting and Liu, Jiayu and Wan, Kai and Sun, Hua and Ji, Mingyue and Caire, Giuseppe},
  journal={IEEE Transactions on Information Theory},
  year={2025},
  publisher={IEEE}
}

@article{zhao2023secure,
  title={Secure Summation: Capacity Region, Groupwise Key, and Feasibility},
  author={Zhao, Yizhou and Sun, Hua},
  journal={IEEE Transactions on Information Theory},
  year={2023},
  publisher={IEEE}
}

@article{wan2024capacity,
  title={The Capacity Region of Information Theoretic Secure Aggregation with Uncoded Groupwise Keys},
  author={Wan, Kai and Sun, Hua and Ji, Mingyue and Mi, Tiebin and Caire, Giuseppe},
  journal={IEEE Transactions on Information Theory},
  year={2024},
  publisher={IEEE}
}

@article{so2022lightsecagg,
  title={Lightsecagg: a lightweight and versatile design for secure aggregation in federated learning},
  author={So, Jinhyun and He, Chaoyang and Yang, Chien-Sheng and Li, Songze and Yu, Qian and E Ali, Ramy and Guler, Basak and Avestimehr, Salman},
  journal={Proceedings of Machine Learning and Systems},
  volume={4},
  pages={694--720},
  year={2022}
}

@article{kairouz2021advances,
  title={Advances and open problems in federated learning},
  author={Kairouz, Peter and McMahan, H Brendan and Avent, Brendan and Bellet, Aur{\'e}lien and Bennis, Mehdi and Bhagoji, Arjun Nitin and Bonawitz, Kallista and Charles, Zachary and Cormode, Graham and Cummings, Rachel and others},
  journal={Foundations and trends{\textregistered} in machine learning},
  volume={14},
  number={1--2},
  pages={1--210},
  year={2021},
  publisher={Now Publishers, Inc.}
}

@inproceedings{bonawitz2017practical,
  title={Practical secure aggregation for privacy-preserving machine learning},
  author={Bonawitz, Keith and Ivanov, Vladimir and Kreuter, Ben and Marcedone, Antonio and McMahan, H Brendan and Patel, Sarvar and Ramage, Daniel and Segal, Aaron and Seth, Karn},
  booktitle={proceedings of the 2017 ACM SIGSAC Conference on Computer and Communications Security},
  pages={1175--1191},
  year={2017}
}

@article{konecny2016federated,
  title={Federated learning: Strategies for improving communication efficiency},
  author={Konecn{\`y}, Jakub and McMahan, H Brendan and Yu, Felix X and Richt{\'a}rik, Peter and Suresh, Ananda Theertha and Bacon, Dave},
  journal={arXiv preprint arXiv:1610.05492},
  volume={8},
  year={2016}
}

@inproceedings{mcmahan2017communication,
  title={Communication-efficient learning of deep networks from decentralized data},
  author={McMahan, Brendan and Moore, Eider and Ramage, Daniel and Hampson, Seth and y Arcas, Blaise Aguera},
  booktitle={Artificial intelligence and statistics},
  pages={1273--1282},
  year={2017},
  organization={PMLR}
}

@ARTICLE{9834981,
  author={Zhao, Yizhou and Sun, Hua},
  journal={IEEE Transactions on Information Theory}, 
  title={Information Theoretic Secure Aggregation With User Dropouts}, 
  year={2022},
  volume={68},
  number={11},
  pages={7471-7484},
  doi={10.1109/TIT.2022.3192874}}

@article{Zhang_Li_Wan_DSA,
  title={Information-Theoretic Decentralized Secure Aggregation with Collusion Resilience},
  author={Xiang Zhang and Zhou Li and Shuangyang Li and Kai Wan and Derrick Wing Kwan Ng and Giuseppe Caire},
  journal={arXiv preprint arXiv:2508.00596},
  year={2025}
}

@article{Li_Zhang_GroupwiseDSA,
  title={The Capacity of Collusion-Resilient Decentralized Secure Aggregation with Groupwise Keys},
  author={Zhou Li and Xiang Zhang and Yizhou Zhao and Haiqiang Chen and Jihao Fan and Giuseppe Caire},
  journal={arXiv preprint arXiv:2511.14444},
  year={2025}
}

@article{Li_Zhang_WeaklyDSA,
  title={Optimal Key Rates for Decentralized Secure Aggregation with Arbitrary Collusion and Heterogeneous Security Constraints},
  author={Zhou Li and Xiang Zhang and Giuseppe Caire},
  journal={arXiv preprint arXiv:2512.16112},
  year={2025}
}

@article{Li_Zhang_WeaklyHSA,
  title={Hierarchical Secure Aggregation with Heterogeneous Security Constraints and Arbitrary User Collusion},
  author={Zhou Li and Xiang Zhang and Jiawen Lv and Haiqiang Chen and Jihao Fan and Giuseppe Caire},
  journal={arXiv preprint arXiv:2507.14768},
  year={2025}
}

@inproceedings{Zhu2020DLG,
  title={Deep Leakage from Gradients},
  author={Zhu, Ligeng and Liu, Zhijian and Han, Song},
  booktitle={Advances in Neural Information Processing Systems},
  volume={33},
  pages={14774--14784},
  year={2020}
}

@misc{li_zhang_MSSA,
      title={Optimal Rate Region for Multi-server Secure Aggregation with User Collusion}, 
      author={Zhou Li and Xiang Zhang and Kai Wan and Hua Sun and Mingyue Ji and Giuseppe Caire},
      year={2026},
      eprint={2601.06836},
      archivePrefix={arXiv},
      primaryClass={cs.IT},
      url={https://arxiv.org/abs/2601.06836}, 
}

@misc{zhang2026informationtheoreticsecureaggregationregular,
title={Information-Theoretic Secure Aggregation over Regular Graphs}, 
author={Xiang Zhang and Zhou Li and Han Yu and Kai Wan and Hua Sun and Mingyue Ji and Giuseppe Caire},
year={2026},
eprint={2601.19183},
archivePrefix={arXiv},
primaryClass={cs.IT},
url={https://arxiv.org/abs/2601.19183}, 
}

@article{weng2026resilient,
  title={On Resilient and Efficient Linear Secure Aggregation in Hierarchical Federated Learning},
  author={Weng, Shudi and Zhang, Xiang and Zhao, Yizhou and Caire, Giuseppe and Xiao, Ming and Skoglund, Mikael},
  journal={arXiv preprint arXiv:2601.12853},
  year={2026}
}

@misc{xu2026networkfunctioncomputationvector,
      title={Network function computation with vector linear target function and security function}, 
      author={Min Xu and Qian Chen and Gennian Ge},
      year={2026},
      eprint={2602.07316},
      archivePrefix={arXiv},
      primaryClass={cs.IT},
      url={https://arxiv.org/abs/2602.07316}, 
}

@misc{xu2025hierarchicalsecureaggregationrelay,
      title={On hierarchical secure aggregation against relay and user collusion}, 
      author={Min Xu and Xuejiao Han and Kai Wan and Gennian Ge},
      year={2025},
      eprint={2511.20117},
      archivePrefix={arXiv},
      primaryClass={cs.IT},
      url={https://arxiv.org/abs/2511.20117}, 
}

@ARTICLE{Zhang_wan_HSA,
  author={Zhang, Xiang and Wan, Kai and Sun, Hua and Wang, Shiqiang and Ji, Mingyue and Caire, Giuseppe},
  journal={IEEE Transactions on Information Theory}, 
  title={Optimal Communication and Key Rate Region for Hierarchical Secure Aggregation With User Collusion}, 
  year={2026},
  volume={72},
  number={2},
  pages={1030-1050},
  doi={10.1109/TIT.2025.3648767}}

@misc{hu2026capacityregionindividualkey,
      title={On the Capacity Region of Individual Key Rates in Vector Linear Secure Aggregation}, 
      author={Lei Hu and Sennur Ulukus},
      year={2026},
      eprint={2601.03241},
      archivePrefix={arXiv},
      primaryClass={cs.IT},
      url={https://arxiv.org/abs/2601.03241}, 
}

@article{Schwartz,
  title={{Fast Probabilistic Algorithms for Verification of Polynomial Identities}},
  author={Schwartz, Jacob T},
  journal={Journal of the ACM (JACM)},
  volume={27},
  number={4},
  pages={701--717},
  year={1980},
  publisher={ACM}
}

@inproceedings{Zippel,
  title={{Probabilistic Algorithms for Sparse Polynomials}},
  author={Zippel, Richard},
  booktitle={International symposium on symbolic and algebraic manipulation},
  pages={216--226},
  year={1979},
  organization={Springer}
}

@article{Demillo_Lipton,
  title={{A Probabilistic Remark on Algebraic Program Testing}},
  author={Demillo, Richard A and Lipton, Richard J},
  journal={Information Processing Letters},
  volume={7},
  number={4},
  pages={193--195},
  year={1978},
  publisher={Elsevier}
}
